\documentclass[11pt]{article}

\usepackage[T1]{fontenc}
\usepackage{lmodern}
\usepackage[margin=1in]{geometry}
\usepackage{amsmath,amssymb,amsthm}
\usepackage{array}
\usepackage{microtype}
\usepackage{placeins}
\usepackage{algorithm}
\usepackage{algpseudocode}

\usepackage{tikz}
\usetikzlibrary{arrows.meta,positioning,calc,fit,backgrounds}
\usepackage{hyperref}
\hypersetup{
  hidelinks,
  pdftitle={A Simple Algorithm Breaking the 1/3 Barrier for Approximate Nash Equilibria in Bimatrix Games},
  pdfauthor={Dongchen Li and Hanyu Li},
  pdfkeywords={approximate Nash equilibria, bimatrix games, 1/3 barrier, one-third barrier, polynomial-time algorithms, strategy mixing}
}
\newtheorem{theorem}{Theorem}
\newtheorem{lemma}[theorem]{Lemma}
\newtheorem{proposition}[theorem]{Proposition}
\theoremstyle{remark}
\newtheorem*{remark}{Remark}
\theoremstyle{plain}
\newcommand{\sm}{\mathsf{SimpleMixing}}
\newcommand{\om}{\mathsf{OptimalMixing}}
\newcommand{\spoint}{\mathsf{StationaryPoint}}
\newcommand{\br}{\mathsf{BestResponse}}
\newcommand{\um}{\mathsf{UniformMixing}}
\newcommand{\codefootref}{\hyperref[fn:code]{\textsuperscript{\ref*{fn:code}}}}
\newcommand{\conv}{\operatorname{conv}}
\DeclareMathOperator{\supp}{supp}
\DeclareMathOperator{\suppmax}{suppmax}

\title{A Simple Algorithm Breaking the $1/3$ Barrier\\
for Approximate Nash Equilibria in Bimatrix Games}
\author{
  Dongchen Li\thanks{Email: \href{mailto:dongchen.li@connect.hku.hk}{\texttt{dongchen.li@connect.hku.hk}}}\\
  {\small School of Computing and Data Science}\\
  {\small The University of Hong Kong}
  \and
  Hanyu Li\thanks{Email: \href{mailto:lhydave@pku.edu.cn}{\texttt{lhydave@pku.edu.cn}}}\\
  {\small CFCS, School of Computer Science}\\
  {\small Peking University}
}
\date{}

\begin{document}
\maketitle

\begin{abstract}
The PPAD-completeness of computing Nash equilibria in bimatrix games
has driven extensive research into polynomial-time algorithms
for $\epsilon$-approximate Nash equilibria.
Determining the strongest approximation guarantee achievable in
polynomial time is a fundamental question in the complexity of
approximate Nash equilibria.
Despite sustained efforts, breaking the $1/3$ barrier has long been
regarded as a difficult open problem.
In this paper, we give an affirmative answer by proposing a simple,
symmetric deterministic polynomial-time algorithm that has no tunable
hyperparameters and achieves an approximation guarantee of approximately
$0.30954+\delta$.
The main idea is to enrich the optimization framework of Tsaknakis
and Spirakis with building blocks drawn from our work on strategy
mixing and automated approximation analysis, yielding a new algorithmic
structure that provides a conceptual explanation for the improved
guarantee.
\end{abstract}

\section{Introduction}
\label{sec:introduction}

% Paragraph 1: the historical computational question.
\emph{Nash equilibrium} (NE), introduced by Nash~\cite{nash1951}, is a foundational
solution concept in game theory. It describes a strategy profile in which no
player can increase their payoff by changing their strategy unilaterally.
Nash's existence theorem guarantees an equilibrium in mixed strategies for
every finite normal-form game, but computing one has remained a persistent algorithmic
challenge. A major breakthrough came with the complexity results of Daskalakis, Goldberg, and
Papadimitriou~\cite{daskalakis2009} and Chen, Deng, and Teng~\cite{chen2009},
which established $\mathsf{PPAD}$-completeness even for two-player games,
providing strong evidence against polynomial-time algorithms for computing
exact Nash equilibria.

% Paragraph 2: the polynomial-time approximation question.
Approximate Nash equilibria offer a way to relax this computational problem.
For payoffs normalized to $[0,1]$, an \emph{$\epsilon$-approximate Nash equilibrium}
($\epsilon$-NE) allows each player to gain at most $\epsilon$ by a unilateral
deviation. Thus an exact NE is the special case $\epsilon=0$.
Lipton, Markakis, and Mehta~\cite{lipton2003} established
quasi-polynomial-time computation for every fixed $\epsilon>0$, with improved
small-support bounds subsequently given by Babichenko, Barman, and
Peretz~\cite{babichenko2017}. Under the exponential-time hypothesis for
$\mathsf{PPAD}$, Rubinstein~\cite{rubinstein2016} showed that
quasi-polynomial time is necessary for sufficiently small constant errors.
Thus, determining the smallest error achievable in
polynomial time remains a fundamental open question in algorithmic game
theory.

% Polynomial-time approximation history, including the optimization approach.
Most algorithmic work on this question has focused on two-player games,
or bimatrix games~\cite{survey2024}.
Early polynomial-time algorithms achieved guarantees of $0.75$ by Kontogiannis,
Panagopoulou, and Spirakis~\cite{kontogiannis2009} and $0.5$ by Daskalakis,
Mehta, and Papadimitriou~\cite{daskalakisnote2009}. The guarantee was further
improved to $0.382+\delta$ by Daskalakis, Mehta, and
Papadimitriou~\cite{daskalakisprogress2007} and $0.364$ by Bosse, Byrka, and
Markakis~\cite{bbm2010}, before Tsaknakis and Spirakis~\cite{ts2008}
achieved $0.3393+\delta$ through an optimization approach. This guarantee
remained the state of the art for fifteen years, until Deligkas, Fasoulakis,
and Markakis~\cite{dfm2023} reached $1/3+\delta$ by refining the
Tsaknakis--Spirakis optimization framework. Here, $\delta>0$ denotes an arbitrarily small
accuracy parameter.

% Paragraph 5: documented expectations concerning one third.
Despite these advances, going below $1/3$ has long been viewed as a difficult
methodological challenge~\cite{survey2024}. Tsaknakis and Spirakis suggested
that a radically different, possibly randomized, approach might be
needed~\cite{ts2008}. Moreover, the guarantees of the
Tsaknakis--Spirakis and Deligkas--Fasoulakis--Markakis algorithms are
tight~\cite{tightness2023}, so sharper analysis alone cannot improve them.
Achieving a polynomial-time approximation guarantee strictly below $1/3$
has remained a central open problem~\cite{survey2024}.

% Paragraph 6: answer the open problem before explaining the methodology.
In this paper, we resolve this open problem with a deterministic
polynomial-time algorithm that improves the approximation guarantee from
$1/3+\delta$ to approximately $0.30954+\delta$.

\begin{theorem}
\label{thm:main}
There is a deterministic algorithm that, given a bimatrix game with rational
payoffs in $[0,1]$ and a rational $\delta>0$, computes an
$(r_*+\delta)$-approximate Nash equilibrium, where $r_*\approx 0.30954$
is an absolute constant. The running time is polynomial in the input encoding
length and $1/\delta$.
\end{theorem}

% Paragraph 7: LegoNE as the methodological starting point.
Our methodology follows LegoNE~\cite{legone2026,computeraided2026}, a framework
for systematically understanding the construction and analysis of approximate
NE algorithms. LegoNE organizes algorithmic and proof insights into composable
building blocks and combines domain knowledge, large language
models (LLMs),
and computer-aided analysis of worst-case guarantees. It improved the
three-player guarantee from $0.6+\delta$ to $0.5+\delta$~\cite{legone2026,computeraided2026},
using a construction beyond the previously established recursive extension
paradigm~\cite{bbm2010,hemon2008}.

% Codex's role in discovery and analysis, and the connection to earlier mixing work.
Following this perspective, we used OpenAI Codex to explore new algorithms for
bimatrix games. Codex proposed and refined candidate algorithms,
formulated their building-block properties, and developed the proof of the
approximation guarantee. In doing so, Codex used $\sm$~\cite{automating2025} and $\om$~\cite{mixing2025} as building blocks for algorithm design. These
operations were originally formalized to understand mixing and automate
approximation analysis. Combining them with the framework of Tsaknakis and
Spirakis led to a fixed, symmetric construction. The
analysis extends the LegoNE methodology to incorporate the properties of
these optimization operations.

\paragraph{Historical note.}
This research was initiated after an email from Davit Gondauri and subsequent
correspondence concerning his preprint~\cite{gondauri2026}, which reports an
approximation guarantee of \mbox{$1/3-7.5\times10^{-45}$}. At the level of
algorithmic construction, it combines mixing ideas from Bosse, Byrka, and
Markakis~\cite{bbm2010} with the optimization framework of Deligkas,
Fasoulakis, and Markakis~\cite{dfm2023}. The correspondence and our
examination of the manuscript with Codex motivated us to seek a conceptually
new construction yielding a substantial improvement, beyond the very small
reported separation from $1/3$. Our algorithm and approximation analysis were
developed independently.

\section{Preliminaries}
\label{sec:preliminaries}

We first specify the game and the notation used throughout the paper.
The objective is the maximum regret of the two players: minimizing this
quantity gives an approximate Nash equilibrium. We also record the
best-response and convexity properties used in the construction and
its analysis.

For a positive integer $n$, write $[n]=\{1,\ldots,n\}$.

A bimatrix game is specified by two \emph{payoff matrices}
$R,C\in[0,1]^{n_1\times n_2}$, with $n_1$ rows and $n_2$ columns.
Player~1 is the \emph{row player} and
player~2 is the \emph{column player}. A \emph{pure strategy} chooses one
row or column; if the players choose $k\in[n_1]$ and $\ell\in[n_2]$, their payoffs are
$R_{k\ell}$ and $C_{k\ell}$. Thus player~$i\in[2]$ has $n_i$ pure
strategies. A mixed strategy is a probability distribution
over pure strategies. We write $(x^i)_k$ for the probability assigned
to pure strategy $k$ by a mixed strategy $x^i$.
For any positive integer $n$, the \emph{probability simplex} is
\[
 \Delta_n=\left\{z\in\mathbb R_{\ge0}^{n}:
                 \sum_{k\in[n]}z_k=1\right\}.
\]
Thus the strategy set of player~$i$ is $\Delta_{n_i}$.
The \emph{support} of a mixed strategy $z\in\Delta_n$ is
\[
 \supp(z)=\{k\in[n]:z_k>0\}.
\]
The $k$th \emph{standard basis vector} $e_k^i\in\mathbb R^{n_i}$
represents the pure strategy that chooses $k$ with probability one:
$(e_k^i)_k=1$ and $(e_k^i)_\ell=0$ for $\ell\ne k$.
A strategy profile is
$\boldsymbol{x}=(x^1,x^2)$, where $x^i\in\Delta_{n_i}$. Its
\emph{expected payoffs} are
\[
 u_1(\boldsymbol{x})=(x^1)^{\mathsf T}Rx^2,
 \qquad
 u_2(\boldsymbol{x})=(x^1)^{\mathsf T}Cx^2.
\]
We also write $u_i(x^1,x^2)$ for $u_i(\boldsymbol{x})$.
Superscripts on strategies identify the player. Subscripts on named
strategies distinguish outputs of different operations or rounds;
coordinates are written with the subscript outside parentheses, as in
$(x^i)_k$.

For a vector $v\in\mathbb R^n$, $\suppmax(v)$ denotes the index set
of its maximum entries~\cite{tightness2023}. For $\delta\ge0$, we
also write $\suppmax_\delta(v)$ for the indices whose entries are
within $\delta$ of the maximum:
\[
 \begin{aligned}
  \suppmax(v)&=\{k\in[n]:v_k=\max_{\ell\in[n]}v_\ell\},\\
  \suppmax_\delta(v)&=\{k\in[n]:v_k\ge\max_{\ell\in[n]}v_\ell-\delta\}.
 \end{aligned}
\]
Thus $\suppmax_0(v)=\suppmax(v)$.

We write $x^{i\prime}\in\Delta_{n_i}$ for an arbitrary comparison
strategy of player~$i$, and
$\boldsymbol{x}'=(x^{1\prime},x^{2\prime})$ for a comparison profile.
The prime marks a comparison strategy, while the superscript identifies
the player.

Following the notation of LegoNE~\cite{legone2026}, the \emph{regret} of
each player and the \emph{maximum regret} are
\begin{align}
 f_1(\boldsymbol{x})
 &=\max_{x^{1\prime}\in\Delta_{n_1}}u_1(x^{1\prime},x^2)-u_1(\boldsymbol{x}),\notag\\
 f_2(\boldsymbol{x})
 &=\max_{x^{2\prime}\in\Delta_{n_2}}u_2(x^1,x^{2\prime})-u_2(\boldsymbol{x}),\notag\\
 f(\boldsymbol{x})&=\max\{f_1(\boldsymbol{x}),f_2(\boldsymbol{x})\}.
 \label{eq:regret}
\end{align}
For $\epsilon\ge0$, the definition of an $\epsilon$-approximate Nash
equilibrium ($\epsilon$-NE) can therefore be written as
\[
 \boldsymbol{x}\text{ is an }\epsilon\text{-NE}
 \quad\Longleftrightarrow\quad f(\boldsymbol{x})\le\epsilon.
\]
In particular, $\boldsymbol{x}$ is an exact NE if and only if
$f_1(\boldsymbol{x})=f_2(\boldsymbol{x})=0$.
A \emph{best response} attains the corresponding maximum
in~\eqref{eq:regret}; it can be found by comparing pure-strategy payoffs.
The vectors $Rx^2$ and $C^{\mathsf T}x^1$ give the pure-strategy
payoffs of players~1 and~2, respectively. Hence $x^{1\prime}\in\Delta_{n_1}$
and $x^{2\prime}\in\Delta_{n_2}$ are best responses to $x^2$ and $x^1$,
respectively, if and only if
\[
 \supp(x^{1\prime})\subseteq\suppmax(Rx^2),
 \qquad
 \supp(x^{2\prime})\subseteq\suppmax(C^{\mathsf T}x^1).
\]
For a fixed strategy of one player, both regrets, and hence $f$, are
convex functions of the other player's strategy. They need not be jointly convex in the two strategies.

A \emph{convex combination} of strategies is a weighted sum with
nonnegative coefficients summing to one. We write $\conv S$ for the
\emph{convex hull} of a finite set $S=\{s_j:j\in[m]\}$, that is,
\[
 \conv S=\left\{\sum_{j\in[m]}p_j s_j:p\in\Delta_m\right\}.
\]
For two strategies of the same player, these are the mixtures
$(1-p)x^i+p x^{i\prime}$ with $p\in[0,1]$.

\section{The LegoNE framework}
\label{sec:legone-review}

LegoNE provides a symbolic language for composing approximate Nash
equilibrium algorithms and an analyzer for deriving their approximation
guarantees~\cite{legone2026,computeraided2026}. We review these two
components, following~\cite{legone2026}. The
Daskalakis--Mehta--Papadimitriou (DMP) $1/2$-NE
algorithm~\cite{daskalakisnote2009} illustrates both the construction
of a program and its analysis. This viewpoint will also be used to
present and analyze our algorithm in Sections~\ref{sec:algorithm}
and~\ref{sec:analysis}.

\subsection{Language for algorithm design}
\label{sec:legone-blocks}

The LegoNE language specifies algorithms for games with a fixed number
of players by composing high-level building blocks. These blocks
encapsulate operations from the literature, including best responses,
mixing, linear programming, and stationary-point computation. An
algorithm is a sequence of assignments: each call uses previously
available variables as inputs and introduces new variables for its
outputs. The program ends by returning one strategy for each player.
This gives a structured design space in which candidate algorithms
are formed by choosing blocks, their order, and their inputs.

We fix the game throughout a construction. The blocks used in this
paper have strategy inputs and outputs. Their accuracy parameters
are fixed throughout a construction and specified when the blocks
are introduced. Their computational guarantees
are stated when the blocks are introduced. A block also has a
mathematical specification of the properties guaranteed by its
outputs. The analyzer uses these specifications to evaluate a
composition of blocks.

For the DMP example, $\mathsf{Random1}$ supplies an initial row
strategy, $\br$ returns a best response to its input strategy
for the other player, and
$\um$ returns the equal mixture of its input strategies.
In particular,
$\um(x^i,x^{i\prime})=(x^i+x^{i\prime})/2$ for
$x^i,x^{i\prime}\in\Delta_{n_i}$.
Algorithm~\ref{alg:dmp} first constructs a column best response,
then a row best response, and finally mixes the two row strategies.
All the calls in this example take polynomial time.

\begin{algorithm}[htbp]
\caption{The DMP $1/2$-NE algorithm}
\label{alg:dmp}
\begin{algorithmic}[1]
 \Require A bimatrix game.
 \State $x^1=\mathsf{Random1}()$
 \State $y^2=\br(x^1)$
 \State $y^1=\br(y^2)$
 \State $z^1=\um(x^1,y^1)$
 \State \Return $(z^1,y^2)$
\end{algorithmic}
\end{algorithm}

Write $\boldsymbol z=(z^1,y^2)$ for its output. The algorithm uses
only the stated operations; its approximation guarantee is derived
from their properties below. More generally, this modular language
allows new compositions to be proposed and evaluated using the same
library of building blocks.

\subsection{Automated approximation analysis}
\label{sec:legone-instantiation}

The analysis must establish a bound on the output regret for every
game, with payoff matrices of arbitrary size and a continuum of mixed
strategies. LegoNE first translates the program into logical
properties. It then applies two abstraction principles,
\emph{instantiation} and \emph{forgetting}, to turn this proof
obligation into a fixed-size constrained optimization problem.
We describe the translation before explaining these two principles.

\paragraph{Logic encoding.}
Under Floyd--Hoare semantics, each line is represented by the
properties it guarantees. For example, the row best-response call
in Algorithm~\ref{alg:dmp} has the \emph{logic encoding}
\begin{equation}
 \phi[y^1=\br(y^2)]:\quad
 \forall x^{1\prime}\in\Delta_{n_1},\,
 u_1(x^{1\prime},y^2)\le u_1(y^1,y^2).
 \label{eq:dmp-properties}
\end{equation}
The column best-response call has the analogous encoding.
The $\um$ call is encoded by its behavior on payoff and regret
terms: for every $x^{2\prime}\in\Delta_{n_2}$, it guarantees
\[
 \begin{aligned}
  U(z^1,x^{2\prime})
   &=\tfrac12 U(x^1,x^{2\prime})+\tfrac12 U(y^1,x^{2\prime}),\\
  f_2(z^1,x^{2\prime})
   &\le\tfrac12 f_2(x^1,x^{2\prime})+\tfrac12 f_2(y^1,x^{2\prime}).
 \end{aligned}
\]
Here $U$ ranges over the payoff functions and their linear combinations,
as in the LegoNE encoding; the second line is the convexity property
of the other player's regret. Thus the mixing step supplies the
properties used in its analysis, in addition to defining a strategy.
The initialization imposes no further property on $x^1$.

The encoding of the whole program $\Gamma$ combines these line
encodings with \emph{inherent formulas}, denoted by $\phi_0$:
\[
 \phi[\Gamma]
 =\phi_0\wedge\phi[\text{line }1]\wedge\cdots
          \wedge\phi[\text{line }k],
\]
where $k$ is the number of building-block calls. The inherent formulas
include payoff and regret bounds in $[0,1]$, the definition of regret,
and the bounds relating a best-response payoff to the payoff of any
strategy and to~$1$. The proof goal is
\[
 \phi[\Gamma]\quad\Longrightarrow\quad
 f(\boldsymbol{x}_{\mathrm{out}})\le b
\]
for all payoff functions and all values of the strategy variables in
the program, where $\boldsymbol{x}_{\mathrm{out}}$ is its returned
profile and $b$ is the proposed approximation bound. At this stage the formulas still quantify over strategies
and payoff functions, so they are not a finite real-variable
optimization problem.

\paragraph{Instantiation.}
The analyzer substitutes each universally quantified strategy
variable with all strategy variables of the same player appearing
in the program. For the row best-response property
in~\eqref{eq:dmp-properties}, these are $x^1,y^1,z^1$. Besides the
trivial substitution $x^{1\prime}=y^1$, this gives
\[
 u_1(x^1,y^2)\le u_1(y^1,y^2),\qquad
 u_1(z^1,y^2)\le u_1(y^1,y^2).
\]
Universally quantified payoff variables are instantiated with the
payoff functions appearing in the program.

There is also an instantiation using the maximum-payoff operator.
When a quantified strategy occurs only once, in a payoff term, that
term can be replaced by its maximum over the strategy. For the same
best-response property, this gives
\[
 \max_{x^{1\prime}\in\Delta_{n_1}}u_1(x^{1\prime},y^2)
 \le u_1(y^1,y^2).
\]
The inherent maximum-payoff bound supplies the reverse inequality,
so the two terms are equal. Applying the corresponding procedures
to the other line encodings and to $\phi_0$ produces finitely many
constraints involving payoff, regret, and maximum-payoff terms.
If a block property includes an auxiliary real variable, one value
of that variable is shared across all instantiations of the property.

\paragraph{Forgetting.}
The instantiated constraints still involve unknown payoff functions.
Forgetting renames every distinct payoff, regret, and maximum-payoff
term as a symbolic real variable, consistently using the same name
for every occurrence of the same term. For example,
\[
 u_1(x^1,y^2)\mapsto v_{1,x^1,y^2},\qquad
 u_1(y^1,y^2)\mapsto v_{1,y^1,y^2},
\]
so their comparison becomes
$v_{1,x^1,y^2}\le v_{1,y^1,y^2}$.
Each instantiated constraint is rewritten using these new symbols.
Quantification over payoff functions and strategies is replaced by
quantification over these real variables. The functional origin of
the terms is forgotten; their encoded arithmetic relations remain.

Auxiliary existential real variables in the encodings appear
as variables of the constrained program, with their ranges and
relations unchanged. Let $\boldsymbol v$ collect these variables and
the renamed terms, and write $\phi'[\Gamma]$ for the resulting real
constraints. The same renaming turns the output maximum regret into
an arithmetic expression $g_\Gamma(\boldsymbol v)$.
The smallest bound valid for all assignments satisfying the constraints
is the optimal value of
\[
 \begin{aligned}
  \operatorname*{maximize}_{\boldsymbol v}\quad
       &g_\Gamma(\boldsymbol v)\\
  \text{subject to}\quad&\phi'[\Gamma].
 \end{aligned}
\]
The analyzer compiles the constraints and objective into this problem
for an external optimizer. Its size is fixed once the program,
encodings, and instantiation tactic are fixed; it does not grow with
the number of pure strategies.
Every execution satisfies the encoded and instantiated properties,
and therefore gives a feasible assignment of the real variables.
An upper bound on this optimization problem consequently certifies
an approximation guarantee. It is the best bound obtainable from
this constraint system, which may also admit assignments that
are not realized by a game.

For DMP, the instantiated best-response properties and regret
definitions give $f_1(y^1,y^2)=f_2(x^1,y^2)=0$.
The mixing encoding, instantiated at $x^{2\prime}=y^2$, and the same
regret definitions give
\[
 \begin{aligned}
  f_1(\boldsymbol z)&=\tfrac12 f_1(x^1,y^2),\\
  f_2(\boldsymbol z)&\le\tfrac12 f_2(y^1,y^2).
 \end{aligned}
\]
These are relations among terms in the instantiated system, so they
remain arithmetic constraints after forgetting. The inherent regret
bounds then imply $g_\Gamma(\boldsymbol v)\le1/2$. This illustrates how the
analysis derives the guarantee from a program's encodings. For our
construction, Section~\ref{sec:analysis} gives the necessary properties
and their arithmetic consequences; the complete encoding, exact proof
checkers, and replay instructions are available in the
accompanying repository\footnote{\url{https://github.com/lhydave/new-approx-ne-algo}.\label{fn:code}}.

\section{The algorithm}
\label{sec:algorithm}

We now use the building-block viewpoint of
Section~\ref{sec:legone-review} to describe our algorithm. We first
recall the framework of Tsaknakis and Spirakis, the improvement of
Deligkas, Fasoulakis, and Markakis, and the motivation for a further
strategy construction. We then define the mixing operations, state
their properties and computational guarantees, and present the
complete algorithm.

\subsection{The Tsaknakis--Spirakis framework}
\label{sec:ts-framework}
\label{sec:motivation}

The Tsaknakis--Spirakis algorithm (TS) formulates approximate Nash
equilibrium computation as the minimization of $f$ over
$\Delta_{n_1}\times\Delta_{n_2}$. Its idea is to use a descent method
to find a stationary point and then construct a profile with lower
regret~\cite{ts2008,tightness2023}. Accordingly, the algorithm has a
descent phase and a strategy-construction phase.

The descent phase computes a profile
$\boldsymbol{x}_s=(x_s^1,x_s^2)$ and two dual strategies
$y_s^1,y_s^2$. The strategy-construction phase uses these strategies
to choose the output profile.

We first describe the descent outputs and their properties, following
the minimax formulation in~\cite{tightness2023}. This is a classical
technique for approximate Nash equilibrium computation, so we state
the formulation and its output properties without repeating the
derivation; see~\cite{ts2008,tightness2023} for details.
We then explain how
TS and the algorithm of Deligkas, Fasoulakis, and Markakis (DFM) use
these outputs, and why a further improvement requires a different
strategy construction.

Let $\boldsymbol{x}$ be the current profile and $\boldsymbol{x}'$ a
comparison profile, both in $\Delta_{n_1}\times\Delta_{n_2}$.
The \emph{Dini directional derivative} of $f$ at
$\boldsymbol{x}$ in the direction $\boldsymbol{x}'-\boldsymbol{x}$ is
\[
 Df(\boldsymbol{x},\boldsymbol{x}')
 =\lim_{\theta\downarrow0}
   \frac{f((1-\theta)\boldsymbol{x}+\theta\boldsymbol{x}')
                 -f(\boldsymbol{x})}{\theta}.
\]
A profile $\boldsymbol{x}$ is a \emph{stationary point} if
$Df(\boldsymbol{x},\boldsymbol{x}')\ge0$ for every
$\boldsymbol{x}'\in\Delta_{n_1}\times\Delta_{n_2}$.
The descent procedure equalizes the two players' regrets without
increasing $f$. In the following description, we therefore assume
$f_1(\boldsymbol{x})=f_2(\boldsymbol{x})$.

We use $\boldsymbol{y}=(y^1,y^2)$ for the best-response strategies
in the minimax problem. Let
$\boldsymbol{y}\in\Delta_{n_1}\times\Delta_{n_2}$ satisfy
\[
 \supp(y^1)\subseteq\suppmax(Rx^2),
 \qquad
 \supp(y^2)\subseteq\suppmax(C^{\mathsf T}x^1).
\]
Thus $y^1,y^2$ are best responses to $x^2,x^1$, respectively.
For $\rho\in[0,1]$, define
\begin{equation}
\begin{aligned}
 T(\boldsymbol{x},\boldsymbol{x}',\rho,\boldsymbol{y})={}&
 \rho\bigl[u_1(y^1,x^{2\prime})-u_1(x^1,x^{2\prime})
       -u_1(x^{1\prime},x^2)+u_1(\boldsymbol{x})\bigr]\\
 &+(1-\rho)\bigl[u_2(x^{1\prime},y^2)-u_2(x^{1\prime},x^2)
       -u_2(x^1,x^{2\prime})+u_2(\boldsymbol{x})\bigr].
\end{aligned}
 \label{eq:stationary-T}
\end{equation}
The corresponding \emph{minimax value} is
\[
 V(\boldsymbol{x})
   =\min_{\boldsymbol{x}'}\max_{\rho,\boldsymbol{y}}
       T(\boldsymbol{x},\boldsymbol{x}',\rho,\boldsymbol{y}),
\]
where $\boldsymbol{x}'\in\Delta_{n_1}\times\Delta_{n_2}$,
$\rho\in[0,1]$, and
$\boldsymbol{y}\in\Delta_{n_1}\times\Delta_{n_2}$ satisfies the
support conditions above.
The minimum of $Df(\boldsymbol{x},\boldsymbol{x}')$ over
$\boldsymbol{x}'$ equals $V(\boldsymbol{x})-f(\boldsymbol{x})$.
Consequently, $\boldsymbol{x}$ is stationary if and only if
$V(\boldsymbol{x})=f(\boldsymbol{x})$.

The polynomial-time descent procedure uses a rational accuracy
parameter $\delta_f>0$. Its minimax problem replaces the exact
best-response sets by
\[
 \supp(y^1)\subseteq\suppmax_{\delta_f}(Rx^2),
 \qquad
 \supp(y^2)\subseteq\suppmax_{\delta_f}(C^{\mathsf T}x^1).
\]
Let $V_{\delta_f}(\boldsymbol{x})$ denote the value of this enlarged
problem. The procedure terminates at a profile
$\boldsymbol{x}_s$ satisfying
\begin{equation}
 f_1(\boldsymbol{x}_s)=f_2(\boldsymbol{x}_s)=f(\boldsymbol{x}_s),
 \qquad V_{\delta_f}(\boldsymbol{x}_s)\ge f(\boldsymbol{x}_s)-\delta_f.
 \label{eq:stationary-termination}
\end{equation}
This minimax value can be computed by a linear program (LP).
An optimal dual solution at $\boldsymbol{x}_s$ determines
\[
 (\rho,y_s^1,y_s^2)\in[0,1]\times\Delta_{n_1}\times\Delta_{n_2}.
\]
Write $\boldsymbol{y}_s=(y_s^1,y_s^2)$ and $\sigma=1-\rho$.
The weights $\rho,\sigma$ correspond to the row and column players,
respectively. In the enlarged LP, a positive row weight normalizes
dual weights on pure strategies whose payoffs are within $\delta_f$
of the best-response payoff. Averaging these payoffs therefore gives,
when $\rho>0$,
\[
 \begin{aligned}
  u_1(y_s^1,x_s^2)
   &=\sum_{k\in[n_1]}(y_s^1)_k\,u_1(e_k^1,x_s^2)\\
   &\ge\max_{x^{1\prime}\in\Delta_{n_1}}u_1(x^{1\prime},x_s^2)-\delta_f.
 \end{aligned}
\]
The column-player argument is the same. The output properties used
in our analysis can thus be stated directly as payoff comparisons:
\begin{equation}
 \begin{aligned}
  \forall x^{1\prime}\in\Delta_{n_1},\quad
  \rho>0&\ \Longrightarrow\quad
  u_1(x^{1\prime},x_s^2)\le u_1(y_s^1,x_s^2)+\delta_f,\\
  \forall x^{2\prime}\in\Delta_{n_2},\quad
  \sigma>0&\ \Longrightarrow\quad
  u_2(x_s^1,x^{2\prime})\le u_2(x_s^1,y_s^2)+\delta_f.
 \end{aligned}
 \label{eq:ts-best-response}
\end{equation}
When a weight is zero, its corresponding strategy can be chosen
arbitrarily. Moreover, the same dual solution satisfies
\begin{equation}
 f(\boldsymbol{x}_s)-\delta_f\le
 T(\boldsymbol{x}_s,\boldsymbol{x}',\rho,\boldsymbol{y}_s)
 \quad\forall\boldsymbol{x}'\in\Delta_{n_1}\times\Delta_{n_2}.
 \label{eq:stationary}
\end{equation}
We regard this procedure as the $\spoint$ building block of
LegoNE~\cite{legone2026}, with accuracy input $\delta_f$ and four
strategy outputs:
\[
 (x_s^1,x_s^2,y_s^1,y_s^2)=\spoint(\delta_f).
\]
For a fixed bimatrix game, this block has no strategy input;
$\delta_f$ is its only explicit input. It computes the four
strategies in time polynomial in the input encoding
length and $1/\delta_f$~\cite{ts2008,dfm2023}.

We call $\boldsymbol x_s$ the stationary profile below, meaning the
approximate output characterized by~\eqref{eq:stationary-termination}
and~\eqref{eq:stationary}. It need not have small regret. Even when
$y_s^1,y_s^2$ are best responses to $x_s^2,x_s^1$, respectively,
the profile $\boldsymbol{y}_s$ can have large regret. In its
strategy-construction phase, TS compares $\boldsymbol{x}_s$ with
$\boldsymbol{y}_s$ and profiles in which one player mixes her stationary
and dual strategies while the other uses her dual strategy. The mixing
coefficients are chosen using payoff differences between the stationary
and dual strategies. These comparisons give the guarantee
$0.3393+\delta$~\cite{ts2008}, but in the bottleneck cases they do not
guarantee $1/3+\delta$~\cite{dfm2023}.

DFM uses the same descent procedure and replaces the strategy construction
in these bottleneck cases. Its main addition is a best response to an
equal mixture of the opponent's stationary and dual strategies, such as
$(x_s^2+y_s^2)/2$ for the row player. DFM uses this additional candidate
in further convex combinations and uses the resulting
payoff relations to bound both players' regrets. The column-player
construction is symmetric. These modifications improve the guarantee
to $1/3+\delta$~\cite{dfm2023}.

The central role of $\spoint$ can be understood through its logical
encoding. Following the viewpoint of~\cite{computeraided2026}, consider
the strategy variables that can be instantiated in a block property.
The inequality~\eqref{eq:stationary} has \emph{two} universally
quantified strategy variables, one for each player:
\[
 \forall x^{1\prime}\in\Delta_{n_1},
 \forall x^{2\prime}\in\Delta_{n_2},\quad
 f(\boldsymbol{x}_s)-\delta_f\le
 T(\boldsymbol{x}_s,(x^{1\prime},x^{2\prime}),\rho,\boldsymbol{y}_s).
\]
Both players' payoff terms occur in this single inequality, and the
same dual weight $\rho$ is shared by all its instantiations. If the
construction provides $t_1$ row strategies and $t_2$ column strategies,
substituting every pair gives $t_1t_2$ inequalities. Thus, the stationary
and dual strategies can be related to strategies from different calls
through a family of joint payoff constraints.

For a fixed $x^2\in\Delta_{n_2}$, the output
$y^1=\br(x^2)$ instead satisfies
\[
 \forall x^{1\prime}\in\Delta_{n_1},\quad
 u_1(x^{1\prime},x^2)\le u_1(y^1,x^2).
\]
This property has \emph{one} universally quantified strategy variable.
Substituting the $t_1$ row strategies gives $t_1$ comparisons of row
payoffs, all against the same opponent strategy $x^2$. DFM's additional
best response supplies comparisons of this form, with $x^2$ equal to a
prescribed mixture. They improve the bound, but provide limited
additional relations between the two players' strategies. The
tightness results in~\cite{tightness2023} show that the resulting
construction still admits instances attaining $1/3$, or approaching
$1/3$ even in the cases using the additional best response.

This comparison suggests seeking further polynomial-time building
blocks whose properties relate both players' payoffs and regrets, and
composing them so that outputs of one call can be substituted into the
properties of another. Additional strategies also provide further
instantiations of the two strategy variables in~\eqref{eq:stationary}.
The number of quantified variables indicates the available comparisons;
their content and interaction determine the approximation bound. This
motivates our construction while using the same descent procedure.
The following subsections introduce the operations and then explain
their composition.

\subsection{\texorpdfstring{$\om$}{OptimalMixing}}
\label{sec:optimalmixing}

Given $x_s^1,y_s^1,x_s^2,y_s^2$, the objective of $\om$ is to minimize
the maximum regret over mixtures of each player's stationary and dual
strategies~\cite{mixing2025,legone2026}. The two players may use
different mixing coefficients. Formally, write
\[
 x^1(p)=(1-p)x_s^1+p y_s^1,\qquad
 x^2(q)=(1-q)x_s^2+q y_s^2,
\]
and define the corresponding profile by
\[
 \boldsymbol{x}(p,q)
 =\bigl(x^1(p),x^2(q)\bigr),
 \qquad (p,q)\in[0,1]^2.
\]
The optimization problem is
\[
 \min_{(p,q)\in[0,1]^2}f(\boldsymbol{x}(p,q)).
\]

For a fixed accuracy parameter $\tau>0$, $\om$ returns a profile whose
maximum regret is at most $\tau$ above this minimum. We separate the
four strategy inputs from the accuracy input by a semicolon:
\[
 (x_o^1,x_o^2)=\om(x_s^1,y_s^1,x_s^2,y_s^2;\tau).
\]
Formally, the profile $\boldsymbol{x}_o=(x_o^1,x_o^2)$ satisfies
\begin{equation}
 \begin{aligned}
  \boldsymbol{x}_o&\in
      \{\boldsymbol{x}(p,q):(p,q)\in[0,1]^2\},\\
  f(\boldsymbol{x}_o)&\le
       \min_{(p,q)\in[0,1]^2}f(\boldsymbol{x}(p,q))+\tau.
 \end{aligned}
 \label{eq:om-definition}
\end{equation}
We also require the output to have regret no greater than that of any
of the four profiles formed by choosing one input strategy for each
player; namely,
\[
 f(\boldsymbol{x}_o)\le
       \min_{(p,q)\in\{0,1\}^2}f(\boldsymbol{x}(p,q)).
\]
For rational input
and rational $\tau>0$, $\om$ can be computed in time polynomial
in the input encoding length and $1/\tau$~\cite{mixing2025}.
When the accuracies are fixed and clear from context, we omit them
from the notation and write $\spoint()$ and
$\om(x_s^1,y_s^1,x_s^2,y_s^2)$. We display them explicitly in
Algorithm~\ref{alg:main}, where their values are inputs.

To analyze this operation in LegoNE, its optimality comparisons must
be combined with relations between the payoffs and regrets of mixed
strategies~\cite{legone2026}. First, fixing one player's strategy
makes each payoff linear in the other player's strategy. Thus the
following identities hold against every comparison strategy:
\begin{equation}
 \begin{aligned}
  \forall i\in[2],\forall p\in[0,1],
  \forall x^{2\prime}\in\Delta_{n_2},\quad
  u_i(x^1(p),x^{2\prime})
   &=(1-p)u_i(x_s^1,x^{2\prime})+p u_i(y_s^1,x^{2\prime}),\\[2pt]
  \forall i\in[2],\forall q\in[0,1],
  \forall x^{1\prime}\in\Delta_{n_1},\quad
  u_i(x^{1\prime},x^2(q))
   &=(1-q)u_i(x^{1\prime},x_s^2)+q u_i(x^{1\prime},y_s^2).
 \end{aligned}
 \label{eq:mixing-payoffs}
\end{equation}

The best-response payoff is convex in the opponent's strategy.
Write these two payoff functions as
\[
 \begin{aligned}
  \forall x^2\in\Delta_{n_2},\quad
  B_1(x^2)&=\max_{x^{1\prime}\in\Delta_{n_1}}u_1(x^{1\prime},x^2),\\
  \forall x^1\in\Delta_{n_1},\quad
  B_2(x^1)&=\max_{x^{2\prime}\in\Delta_{n_2}}u_2(x^1,x^{2\prime}).
 \end{aligned}
\]
Jensen's inequality bounds a best-response payoff at a mixture by
the mixture of the endpoint best-response payoffs. The nonnegative
differences are the \emph{Jensen gaps}:
\begin{equation}
 \begin{aligned}
  \forall q\in[0,1],\quad
  J_1(q)&=(1-q)B_1(x_s^2)+q B_1(y_s^2)-B_1(x^2(q))\ge0,\\
  \forall p\in[0,1],\quad
  J_2(p)&=(1-p)B_2(x_s^1)+p B_2(y_s^1)-B_2(x^1(p))\ge0.
 \end{aligned}
 \label{eq:mixing-gaps}
\end{equation}
By payoff linearity, each regret is exactly the bilinear interpolation
of its four endpoint values minus its Jensen gap. Formally,
\begin{equation}
 \begin{aligned}
  \forall p,q\in[0,1],\quad
  f_1(\boldsymbol{x}(p,q))={}&
    (1-p)(1-q)f_1(\boldsymbol{x}_s)
    +p(1-q)f_1(y_s^1,x_s^2)\\
   &+(1-p)q f_1(x_s^1,y_s^2)
    +pq f_1(\boldsymbol{y}_s)-J_1(q),\\
  \forall p,q\in[0,1],\quad
  f_2(\boldsymbol{x}(p,q))={}&
    (1-p)(1-q)f_2(\boldsymbol{x}_s)
    +p(1-q)f_2(y_s^1,x_s^2)\\
   &+(1-p)q f_2(x_s^1,y_s^2)
    +pq f_2(\boldsymbol{y}_s)-J_2(p).
 \end{aligned}
 \label{eq:mixing-regrets}
\end{equation}

Each bracket below is the regret of one comparison strategy against
an endpoint opponent strategy. Minimizing the weighted sum of these
two regrets gives the corresponding Jensen gap:
\begin{equation}
 \begin{aligned}
  \forall q\in[0,1],\quad
  J_1(q)=\min_{x^{1\prime}\in\Delta_{n_1}}\bigl\{
   &(1-q)[B_1(x_s^2)-u_1(x^{1\prime},x_s^2)]\\
   &+q[B_1(y_s^2)-u_1(x^{1\prime},y_s^2)]\bigr\},\\
  \forall p\in[0,1],\quad
  J_2(p)=\min_{x^{2\prime}\in\Delta_{n_2}}\bigl\{
   &(1-p)[B_2(x_s^1)-u_2(x_s^1,x^{2\prime})]\\
   &+p[B_2(y_s^1)-u_2(y_s^1,x^{2\prime})]\bigr\}.
 \end{aligned}
 \label{eq:mixing-gap-losses}
\end{equation}
The minimum in each formula is attained by a best response to the
corresponding mixture. A lower bound valid for every comparison
strategy therefore bounds the Jensen gap from below. Joint dual properties
can supply such bounds, giving stronger regret comparisons than
using only $J_1,J_2\ge0$.

\begin{remark}
For two strategies per player, the encoding used in
LegoNE~\cite{legone2026} bounds the regret by considering only the
boundary of $[0,1]^2$. Our analysis of $\om$ is more detailed:
it bounds the Jensen gaps using shared dual constraints and also
applies the $\om$ comparison to interior points
$(p,q)\in(0,1)^2$.
\end{remark}

\subsection{\texorpdfstring{$\sm$}{SimpleMixing} and its dual strategies}
\label{sec:simplemixing}

For a fixed opponent's strategy, a best response maximizes one
player's payoff. Minimizing $f$ with one player's strategy fixed
provides another way to construct a strategy, taking both players'
regrets into account. This is the $\sm$ operation~\cite{automating2025}.
Given a row strategy $x^1\in\Delta_{n_1}$, the operation fixes $x^1$
and chooses a column strategy $x^2$ minimizing the maximum regret
of the resulting profile. Formally,
\begin{equation}
 x^2\in\operatorname*{arg\,min}_{x^{2\prime}\in\Delta_{n_2}}f(x^1,x^{2\prime}).
 \label{eq:sm-definition}
\end{equation}
This problem can be solved by linear programming~\cite{automating2025}.
To write it as an LP, introduce a real variable $h$ that
bounds both players' regrets. The row regret is the largest gain
from a pure row deviation, so we impose one constraint for each
$k\in[n_1]$. The column player's best-response payoff depends only
on the fixed input $x^1$. Thus the LP is
\begin{equation}
 \begin{aligned}
  \operatorname*{minimize}_{x^{2\prime},\,h}\quad &h\\
  \text{subject to}\quad
   &h\ge u_1(e_k^1,x^{2\prime})-u_1(x^1,x^{2\prime}),
      \quad k\in[n_1],\\
   &h\ge\max_{\ell\in[n_2]}u_2(x^1,e_\ell^2)
             -u_2(x^1,x^{2\prime}),\\
   &x^{2\prime}\in\Delta_{n_2},\quad h\in\mathbb R.
 \end{aligned}
 \label{eq:sm-lp}
\end{equation}
The maximum in the second constraint is a constant, so all constraints
are linear in $x^{2\prime}$ and $h$.
An optimal column strategy $x^2$ is the primal solution,
and the optimal value is $h=f(x^1,x^2)$.

An optimal dual solution certifies that no column strategy can achieve
a smaller maximum regret. Let $\lambda_k\ge0$ be the dual variable
for the $k$th row-regret constraint in~\eqref{eq:sm-lp}, and let
$1-\omega$ be the dual variable for the column-regret constraint,
where $\omega$ is defined by the sum below.
These variables satisfy
\[
 \omega=\sum_{k\in[n_1]}\lambda_k\in[0,1],\qquad
 \sum_{k\in[n_1]}\lambda_k+(1-\omega)=1.
\]
They therefore assign weights to the pure row-deviation gains and
the column regret. The total row weight is $\omega$, and the column weight is $1-\omega$.
When $\omega>0$, normalizing the row weights gives the row dual strategy
\[
 y^1=\sum_{k\in[n_1]}\frac{\lambda_k}{\omega}e_k^1
       \in\Delta_{n_1}.
\]
When $\omega=0$, there is no row weight to normalize, and $y^1$ can be
chosen arbitrarily.

The operation therefore has one strategy input, the fixed row
strategy $x^1$, and two strategy outputs: the minimizing column
strategy $x^2$ and the row dual strategy $y^1$. We write
\[
 (x^2,y^1)=\sm(x^1).
\]
Both strategy outputs can be computed exactly in polynomial time for
rational input~\cite{automating2025}.
The dual weight $\omega$ is an auxiliary scalar associated with these
outputs. It appears in their properties below.

The dual strategy $y^1$ combines the row-deviation gains into the
payoff difference $u_1(y^1,x^{2\prime})-u_1(x^1,x^{2\prime})$.
The dual solution combines this difference with the column regret,
using weights $\omega$ and $1-\omega$, respectively.
LP duality guarantees that the resulting expression is at least the
optimal maximum regret $f(x^1,x^2)$ for every comparison column
strategy $x^{2\prime}\in\Delta_{n_2}$. Formally,
\begin{equation}
 \forall x^{2\prime}\in\Delta_{n_2},\quad
 \omega\bigl[u_1(y^1,x^{2\prime})-u_1(x^1,x^{2\prime})\bigr]
 +(1-\omega)f_2(x^1,x^{2\prime})\ge f(x^1,x^2).
 \label{eq:simple-dual}
\end{equation}

Substituting $x^{2\prime}=x^2$ into~\eqref{eq:simple-dual} gives
\[
 \begin{aligned}
 f(x^1,x^2)
  &\le\omega\bigl[u_1(y^1,x^2)-u_1(x^1,x^2)\bigr]
          +(1-\omega)f_2(x^1,x^2)\\
  &\le\omega f_1(x^1,x^2)+(1-\omega)f_2(x^1,x^2)
   \le f(x^1,x^2).
 \end{aligned}
\]
The second inequality follows from the definition of $f_1$, and the
last from $f=\max\{f_1,f_2\}$. All inequalities are therefore
equalities. If $\omega>0$, equality in the second inequality implies
\[
 u_1(y^1,x^2)=\max_{x^{1\prime}\in\Delta_{n_1}}u_1(x^{1\prime},x^2).
\]
Thus $y^1$ is a best response to $x^2$, or equivalently,
\begin{equation}
 \forall x^{1\prime}\in\Delta_{n_1},\quad
 \omega>0\quad\Longrightarrow\quad
 u_1(x^{1\prime},x^2)\le u_1(y^1,x^2).
 \label{eq:alignment}
\end{equation}

Both~\eqref{eq:simple-dual} and~\eqref{eq:alignment} refer to this
same optimal dual solution. In particular, $y^1$ and $\omega$ do not
depend on the strategy $x^{2\prime}$ used in~\eqref{eq:simple-dual}.
The LP formulation and proofs of these properties are given
in~\cite{automating2025}.

These output properties motivate using $\sm$ in strategy construction.
In DFM, an additional strategy is obtained by taking a best response to a
prescribed mixture of the opponent's stationary and dual
strategies~\cite{dfm2023}. This step maximizes one player's payoff
against that mixture. The $\sm$ operation instead minimizes the
maximum regret with one player's strategy fixed, and also supplies a
dual strategy. The additional information is the dual
inequality~\eqref{eq:simple-dual}, which relates the candidate regret
to both players' payoff and regret terms. This suggests forming a
candidate profile from the primal output and using the dual output to
construct further strategies. The associated dual relations may help
such a construction improve on DFM's approximation guarantee.

The operation $\sm(x^2)$ is defined symmetrically: it fixes a
column strategy $x^2\in\Delta_{n_2}$ and chooses a row strategy
$x^1\in\Delta_{n_1}$ minimizing the maximum regret of the resulting
profile. It also returns a dual column strategy $y^2\in\Delta_{n_2}$.
Namely,
\[
 \begin{aligned}
  (x^1,y^2)&=\sm(x^2),\\
  x^1&\in\operatorname*{arg\,min}_{x^{1\prime}\in\Delta_{n_1}}
                   f(x^{1\prime},x^2).
 \end{aligned}
\]
It has the same polynomial-time guarantee. In this symmetric call,
$\omega$ is the total dual weight on the pure column-deviation
constraints; $1-\omega$ is the dual multiplier on the single
row-regret constraint. Here $\omega$ belongs to $\sm(x^2)$,
independently of the weight in $\sm(x^1)$. In particular,
\[
 \forall x^{2\prime}\in\Delta_{n_2},\quad
 \omega>0\quad\Longrightarrow\quad
 u_2(x^1,x^{2\prime})\le u_2(x^1,y^2).
\]

\subsection{Complete algorithm}
\label{sec:complete-algorithm}

With these building blocks, we obtain the algorithm for
Theorem~\ref{thm:main}, presented in Algorithm~\ref{alg:main}.
Starting from the output of $\spoint$, it introduces further
strategies through $\sm$ and uses $\om$ to form another candidate
profile. The same operations are applied to both players, giving a
simple, symmetric construction. Figure~\ref{fig:dag} gives a graphical
illustration.

The $\sm$ calls are organized into two rounds. The first round
uses the stationary dual strategies $y_s^1,y_s^2$ as inputs; the
second uses the dual strategies $y_1^1,y_1^2$ returned by the first.
In $x_j^i,y_j^i$, the subscript $j\in[2]$ indicates the round and
the superscript $i\in[2]$ identifies the player who uses the strategy.

Algorithm~\ref{alg:main} makes one $\spoint$ call, four $\sm$ calls,
and one $\om$ call. The algorithm displays the two accuracy inputs
at the calls that use them. Its running time is therefore polynomial
in the input encoding length, $1/\delta_f$, and $1/\tau$.

\begin{algorithm}[htbp]
\caption{Our algorithm for bimatrix games}
\label{alg:main}
\begin{algorithmic}[1]
 \Require A bimatrix game with rational payoffs and rational $\delta_f,\tau>0$.
 \State $(x_s^1,x_s^2,y_s^1,y_s^2)=\spoint(\delta_f)$
 \State $(x_1^2,y_1^1)=\sm(y_s^1)$
 \State $(x_1^1,y_1^2)=\sm(y_s^2)$
 \State $(x_2^2,y_2^1)=\sm(y_1^1)$
 \State $(x_2^1,y_2^2)=\sm(y_1^2)$
 \State $(x_o^1,x_o^2)=\om(x_s^1,y_s^1,x_s^2,y_s^2;\tau)$
 \State $\mathcal S=\{(y_s^1,x_1^2),\ (x_1^1,y_s^2),\
                   (y_1^1,x_2^2),\ (x_2^1,y_1^2),\ (x_o^1,x_o^2)\}$
 \State \Return $\boldsymbol{x}_{\mathrm{out}}\in
        \operatorname*{arg\,min}_{\boldsymbol{x}\in\mathcal S}
                        f(\boldsymbol{x})$
\end{algorithmic}
\end{algorithm}

\begin{figure}[!htb]
\centering
\begin{tikzpicture}[
 x=1.5cm,y=1cm,
 >=Stealth,line cap=round,line join=round,
 every node/.style={font=\small},
 row/.style={circle,draw=dagrow,fill=dagrow!5,
   line width=.95pt,minimum size=7mm,inner sep=1pt},
 col/.style={circle,draw=dagcol,fill=dagcol!5,
   line width=.95pt,minimum size=7mm,inner sep=1pt},
 profile/.style={draw=dagprofile,fill=dagprofile!2,
   rounded corners=5pt,dash pattern=on 3.5pt off 2.3pt,
   line width=.85pt,inner sep=4pt},
 edge/.style={->,line width=1pt},
 stem/.style={line width=1pt},
 op/.style={fill=white,inner sep=1.5pt,font=\footnotesize}]
 \definecolor{dagrow}{RGB}{24,83,146}
 \definecolor{dagcol}{RGB}{177,76,25}
 \definecolor{dagstation}{RGB}{79,86,96}
 \definecolor{dagmix}{RGB}{15,117,109}
 \definecolor{dagprofile}{RGB}{126,70,134}
 \definecolor{dagink}{RGB}{45,49,54}
 % Each candidate is a fixed strategy paired with its primal response.
 \node[row] (ys1) at (3.55,1.8) {$y_s^1$};
 \node[col] (x12) at (3.55,2.8) {$x_1^2$};
 \node[row] (y11) at (6.1,1.8) {$y_1^1$};
 \node[col] (x22) at (6.1,2.8) {$x_2^2$};
 \node[col] (ys2) at (3.55,-1.8) {$y_s^2$};
 \node[row] (x11) at (3.55,-2.8) {$x_1^1$};
 \node[col] (y12) at (6.1,-1.8) {$y_1^2$};
 \node[row] (x21) at (6.1,-2.8) {$x_2^1$};
 % The initial profile supplies strategy inputs to the mixing calls.
 \node[row] (xs1) at (3,.48) {$x_s^1$};
 \node[col] (xs2) at (3,-.48) {$x_s^2$};
 % Optimal mixing ends in the middle; only the two longer LP chains
 % continue to the right, ending at their last dual strategies.
 \node[row] (xo1) at (7.0,.45) {$x_o^1$};
 \node[col] (xo2) at (7.0,-.45) {$x_o^2$};
 \node[row] (y21) at (8.7,1.8) {$y_2^1$};
 \node[col] (y22) at (8.7,-1.8) {$y_2^2$};
 % Leave small ports in each frame wherever a dependency enters or exits.
 % This keeps the dashed grouping outlines separate from the arrows.
 \node[fit=(ys1)(x12),inner sep=4pt] (p1) {};
 \node[fit=(y11)(x22),inner sep=4pt] (p2) {};
 \node[fit=(ys2)(x11),inner sep=4pt] (p3) {};
 \node[fit=(y12)(x21),inner sep=4pt] (p4) {};
 \node[fit=(xo1)(xo2),inner sep=4pt] (po) {};
 \begin{scope}[on background layer]
  \begin{scope}[overlay,even odd rule]
   \clip (-1,-5) rectangle (13,5)
    (p1.west |- ys1.center) circle (3pt)
    (p1.east |- ys1.center) circle (3pt)
    (p1.south -| ys1.center) circle (3pt)
    (p1.east |- x12.center) circle (3pt);
   \draw[profile] (p1.north west) rectangle (p1.south east);
  \end{scope}
  \begin{scope}[overlay,even odd rule]
   \clip (-1,-5) rectangle (13,5)
    (p2.west |- y11.center) circle (3pt)
    (p2.east |- y11.center) circle (3pt)
    (p2.east |- x22.center) circle (3pt);
   \draw[profile] (p2.north west) rectangle (p2.south east);
  \end{scope}
  \begin{scope}[overlay,even odd rule]
   \clip (-1,-5) rectangle (13,5)
    (p3.west |- ys2.center) circle (3pt)
    (p3.east |- ys2.center) circle (3pt)
    (p3.north -| ys2.center) circle (3pt)
    (p3.east |- x11.center) circle (3pt);
   \draw[profile] (p3.north west) rectangle (p3.south east);
  \end{scope}
  \begin{scope}[overlay,even odd rule]
   \clip (-1,-5) rectangle (13,5)
    (p4.west |- y12.center) circle (3pt)
    (p4.east |- y12.center) circle (3pt)
    (p4.east |- x21.center) circle (3pt);
   \draw[profile] (p4.north west) rectangle (p4.south east);
  \end{scope}
  \begin{scope}[overlay,even odd rule]
   \clip (-1,-5) rectangle (13,5)
    (po.west |- xo1.center) circle (3pt)
    (po.west |- xo2.center) circle (3pt);
   \draw[profile] (po.north west) rectangle (po.south east);
  \end{scope}
 \end{scope}
 % StationaryPoint produces the four original strategies.
 \draw[stem,draw=dagstation] (.35,0) --
   node[op,above=5pt,pos=.35,text=dagstation] {$\spoint$} (2.05,0);
 \draw[edge,draw=dagstation] (2.05,0)
   .. controls (2.4,0) and (2.7,1.8) .. (ys1.west);
 \draw[edge,draw=dagstation] (2.05,0)
   .. controls (2.4,0) and (2.42,.48) .. (xs1.west);
 \draw[edge,draw=dagstation] (2.05,0)
   .. controls (2.4,0) and (2.42,-.48) .. (xs2.west);
 \draw[edge,draw=dagstation] (2.05,0)
   .. controls (2.4,0) and (2.7,-1.8) .. (ys2.west);
 % A labelled common stem precedes each pair of outputs.
 % The primal output loops back into its vertical candidate frame;
 % the dual output continues along the horizontal strategy chain.
 \draw[stem,draw=dagrow] (ys1.east) --
   node[op,below=1pt,text=dagrow] {$\sm$} (5.55,1.8);
 \draw[edge,draw=dagrow] (5.55,1.8)
   .. controls (5.79,1.8) and (5.79,2.8) .. (x12.east);
 \draw[edge,draw=dagrow] (5.55,1.8) -- (y11.west);
 \draw[stem,draw=dagrow] (y11.east) --
   node[op,below=1pt,text=dagrow] {$\sm$} (8.15,1.8);
 \draw[edge,draw=dagrow] (8.15,1.8)
   .. controls (8.39,1.8) and (8.39,2.8) .. (x22.east);
 \draw[edge,draw=dagrow] (8.15,1.8) -- (y21.west);
 \draw[stem,draw=dagcol] (ys2.east) --
   node[op,below=1pt,text=dagcol] {$\sm$} (5.55,-1.8);
 \draw[edge,draw=dagcol] (5.55,-1.8)
   .. controls (5.79,-1.8) and (5.79,-2.8) .. (x11.east);
 \draw[edge,draw=dagcol] (5.55,-1.8) -- (y12.west);
 \draw[stem,draw=dagcol] (y12.east) --
   node[op,below=1pt,text=dagcol] {$\sm$} (8.15,-1.8);
 \draw[edge,draw=dagcol] (8.15,-1.8)
   .. controls (8.39,-1.8) and (8.39,-2.8) .. (x21.east);
 \draw[edge,draw=dagcol] (8.15,-1.8) -- (y22.west);
 % Solid input curves distinguish dependencies from dashed profile frames.
 \draw[draw=dagmix,line width=.7pt] (ys1.south)
   .. controls (3.55,.45) and (4.25,0) .. (4.7,0);
 \draw[draw=dagmix,line width=.7pt] (xs1.east)
   .. controls (3.8,.48) and (4.15,0) .. (4.7,0);
 \draw[draw=dagmix,line width=.7pt] (xs2.east)
   .. controls (3.8,-.48) and (4.15,0) .. (4.7,0);
 \draw[draw=dagmix,line width=.7pt] (ys2.north)
   .. controls (3.55,-.45) and (4.25,0) .. (4.7,0);
 \draw[stem,draw=dagmix] (4.7,0) --
   node[op,above=4pt,text=dagmix] {$\om$} (6.2,0);
 \draw[edge,draw=dagmix] (6.2,0)
   .. controls (6.5,0) and (6.5,.45) .. (xo1.west);
 \draw[edge,draw=dagmix] (6.2,0)
   .. controls (6.5,0) and (6.5,-.45) .. (xo2.west);
 % The profile key contains both player colours.
 \draw[dagrow,fill=dagrow!5,line width=.95pt] (.15,2.75) circle (1.3mm);
 \node[anchor=west,font=\footnotesize,text=dagink] at (.42,2.75) {Player 1};
 \draw[dagcol,fill=dagcol!5,line width=.95pt] (.15,2.18) circle (1.3mm);
 \node[anchor=west,font=\footnotesize,text=dagink] at (.42,2.18) {Player 2};
 \draw[profile,rounded corners=3pt] (-.06,1.36) rectangle (.53,1.79);
 \draw[dagrow,fill=dagrow!5,line width=.8pt] (.1,1.575) circle (1mm);
 \draw[dagcol,fill=dagcol!5,line width=.8pt] (.37,1.575) circle (1mm);
 \node[anchor=west,font=\footnotesize,text=dagink] at (.66,1.575)
   {Candidate profile};
 % Final selection is centred below the entire construction.
 \node[text=dagink,font=\normalsize] at (4.25,-4.0)
  {\textbf{Output:}\quad
   \tikz[baseline=2pt]{
      \draw[draw=dagprofile,fill=dagprofile!2,rounded corners=2pt,
        dash pattern=on 2.8pt off 1.8pt,line width=.85pt]
        (0,0) rectangle (.72,.36);
      \draw[dagrow,fill=dagrow!5,line width=.85pt] (.20,.18) circle (.085);
      \draw[dagcol,fill=dagcol!5,line width=.85pt] (.52,.18) circle (.085);
     }\quad with minimal $f$};
\end{tikzpicture}
\caption{Construction of the candidate profiles. Vertices represent
strategies, and labelled edges indicate building-block calls;
the labels show operation names, while Algorithm~\ref{alg:main}
shows their accuracy inputs.
$\spoint$ has no strategy input and produces four strategies.
Each purple dashed frame pairs the two players' strategies in one
candidate profile. The output minimizes
$f=\max\{f_1,f_2\}$ over the five framed profiles.
The strategies $y_2^1,y_2^2$ are used only in the analysis.}
\label{fig:dag}
\end{figure}
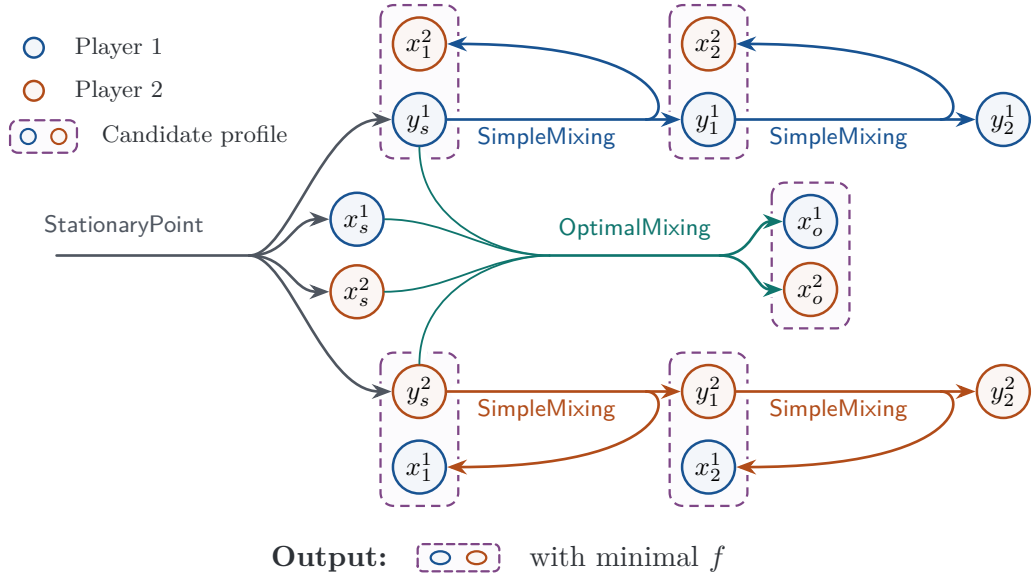

\FloatBarrier
\section{Approximation analysis}
\label{sec:analysis}

We first consider the formal zero-error versions of the two
operations' output properties, written $\delta_f=\tau=0$, and prove
their guarantee by contradiction. The implemented calls use positive
accuracy inputs. Suppose the selected output has
regret above a proposed threshold $r$. Because it selects the best
of five candidates, all five must then exceed $r$. We will show
that this cannot happen when $r$ is above the claimed bound.

The argument proceeds from the game to arithmetic. The properties
of the algorithm's calls first impose restrictions on any such bad
output. We combine and simplify those restrictions until they are
a finite list of inequalities in a few numbers, rather than in the
game's strategies and payoffs. Numerical exploration suggests the
threshold where these inequalities become critical; an exact
calculation defines that threshold, and a global argument shows
that the inequalities cannot hold above it. Finally, we account
for the errors $\delta_f$ and $\tau$ in the actual calls. The main
text explains these transitions; Appendices~\ref{app:instantiation}--\ref{app:errors}
give the detailed calculations and exact-checker paths.

\subsection{Properties needed for the contradiction}
\label{sec:properties}

This subsection gathers the comparisons used later in the proof.
$\spoint$ relates the two players' payoffs at its output; each
$\sm$ call chooses an opponent strategy and compares the result
with other choices; and $\om$ compares its output with profiles formed by
mixing the stationary and best-response strategies. We state the
precise inequalities below and explain what each will contribute.
They follow from the stationary-point results~\cite{ts2008,tightness2023}
and the mixing operations~\cite{automating2025,mixing2025}.

\paragraph{\texorpdfstring{$\spoint$}{StationaryPoint}: one comparison for both players.}
The operation returns a stationary profile $\boldsymbol{x}_s$ and
one auxiliary strategy for each player; write the pair as
$\boldsymbol{y}_s=(y_s^1,y_s^2)$. It makes the two players' regrets at
$\boldsymbol{x}_s$ equal; call their common value $g$:
\begin{equation}
 g=f_1(\boldsymbol{x}_s)=f_2(\boldsymbol{x}_s)=f(\boldsymbol{x}_s).
 \label{eq:properties-g}
\end{equation}
The operation also returns a number $\rho\in[0,1]$ that weighs the
row player's part of a payoff comparison; $1-\rho$ weighs the
column player's part. The same $\rho$ and $\boldsymbol y_s$ work
against every pair of comparison strategies:
\begin{equation}
 \exists\rho\in[0,1],
 \forall x^{1\prime}\in\Delta_{n_1},
 \forall x^{2\prime}\in\Delta_{n_2},\quad
 g-\delta_f\le
 T(\boldsymbol{x}_s,(x^{1\prime},x^{2\prime}),\rho,\boldsymbol{y}_s).
 \label{eq:properties-stationary}
\end{equation}
Here $T$ is the payoff expression in~\eqref{eq:stationary-T}.
The returned strategies also nearly maximize each player's payoff
against the other player's stationary strategy. For the row player,
\begin{equation}
 \forall x^{1\prime}\in\Delta_{n_1},\quad
 \rho>0\quad\Longrightarrow\quad
 u_1(x^{1\prime},x_s^2)\le u_1(y_s^1,x_s^2)+\delta_f.
 \label{eq:properties-stationary-br}
\end{equation}
The column-player property is symmetric. Later we will choose
comparison strategies produced by both $\sm$ rounds. The same
$\rho$ and $\boldsymbol y_s$ must satisfy every resulting inequality.

\paragraph{\texorpdfstring{$\sm$}{SimpleMixing}: comparing a fixed-side choice with alternatives.}
For a fixed row strategy $x^1$, the operation chooses the column
strategy $x^2$ that minimizes the two players' maximum regret. Its
linear program also returns a row strategy $y^1$ derived from the
constraints that certify this choice. Write the outputs as
\[
 (x^2,y^1)=\sm(x^1).
\]
The chosen $x^2$ minimizes maximum regret with $x^1$ fixed, so changing the column strategy alone
cannot reduce the candidate's maximum regret:
\begin{equation}
 \forall x^{2\prime}\in\Delta_{n_2},\quad
 f(x^1,x^2)\le f(x^1,x^{2\prime}).
 \label{eq:properties-primal}
\end{equation}
The same LP assigns a number $\omega\in[0,1]$ to the row
player's possible payoff gain and $1-\omega$ to the column player's
regret, as explained in Section~\ref{sec:simplemixing}. Unlike
$\rho$, this $\omega$ belongs to the particular $\sm$ call. It
gives a second comparison that holds for every column strategy:
\begin{equation}
 \exists\omega\in[0,1],\forall x^{2\prime}\in\Delta_{n_2},\quad
 f(x^1,x^2)\le
 \omega\bigl[u_1(y^1,x^{2\prime})-u_1(x^1,x^{2\prime})\bigr]
 +(1-\omega)f_2(x^1,x^{2\prime}).
 \label{eq:properties-dual}
\end{equation}
The same $y^1$ and $\omega$ work for every $x^{2\prime}$ in this
inequality. This is what lets us compare the output of one call
with strategies from another. At the chosen $x^2$, the inequality
is an equality:
\begin{equation}
 f(x^1,x^2)=
 \omega\bigl[u_1(y^1,x^2)-u_1(x^1,x^2)\bigr]
 +(1-\omega)f_2(x^1,x^2).
 \label{eq:properties-strong-duality}
\end{equation}
If a player's regret has positive total weight, that player's regret
equals the candidate's maximum regret. The LP's complementary-slackness
conditions state this as
\begin{equation}
 \begin{aligned}
  \omega\bigl[f(x^1,x^2)-f_1(x^1,x^2)\bigr]&=0,\\
  (1-\omega)\bigl[f(x^1,x^2)-f_2(x^1,x^2)\bigr]&=0.
 \end{aligned}
 \label{eq:properties-cs}
\end{equation}
It also makes every positively weighted row-deviation constraint
tight. Consequently, when $\omega>0$, the returned $y^1$ is a best
response to $x^2$:
\begin{equation}
 \forall x^{1\prime}\in\Delta_{n_1},\quad
 \omega>0\quad\Longrightarrow\quad
 u_1(x^{1\prime},x^2)\le u_1(y^1,x^2).
 \label{eq:properties-sm-br}
\end{equation}
These are properties of the same LP solution~\cite{automating2025,boyd2004};
the properties of $\sm(x^2)$ are symmetric. We apply them to all
four calls in Algorithm~\ref{alg:main}. The first round's $y$
strategies become the second round's inputs, so the comparisons
from the two rounds constrain some of the same payoffs.

\paragraph{\texorpdfstring{$\om$}{OptimalMixing}: comparing all pairs of mixtures.}
As defined in Section~\ref{sec:optimalmixing},
$\boldsymbol{x}(p,q)$ uses probability $p$ to mix the row player's
two given strategies and probability $q$ for the column player's
two. The two probabilities vary independently over $[0,1]$. The
output is one such profile:
\[
 \exists p_o,q_o\in[0,1],\quad
 \boldsymbol{x}_o=\boldsymbol{x}(p_o,q_o).
\]
The operation minimizes the maximum regret over these mixtures up
to $\tau$, so its output regret is at most $\tau$ above that of any
mixed profile:
\begin{equation}
 \forall p,q\in[0,1],\quad
 f(\boldsymbol{x}_o)\le f(\boldsymbol{x}(p,q))+\tau.
 \label{eq:om}
\end{equation}
To bound regret at these mixtures, we use two identities from
Section~\ref{sec:optimalmixing}. Equation~\eqref{eq:mixing-regrets}
starts with the four unmixed profiles and interpolates their regrets.
It then subtracts a nonnegative correction, called the Jensen gap.
The correction arises because a best response to a mixture must
use one strategy against both endpoint opponents, whereas the
endpoints may have different best responses. Equation~\eqref{eq:mixing-gap-losses}
expresses the correction as a minimum over that one response
strategy. Later comparisons will bound it from below. The
operation also compares with the four unmixed profiles without
an additive error:
\[
 \forall p,q\in\{0,1\},\quad
 f(\boldsymbol{x}_o)\le f(\boldsymbol{x}(p,q)).
\]
In particular, choosing $p=q=0$ gives the stationary profile and
therefore yields
\begin{equation}
 f(\boldsymbol{x}_o)\le g.
 \label{eq:properties-om-g}
\end{equation}

\paragraph{Output selection.}
Finally, the algorithm chooses the candidate of smallest maximum
regret. With $\mathcal S$ as in Algorithm~\ref{alg:main}, this gives
\begin{equation}
 \boldsymbol{x}_{\mathrm{out}}\in\mathcal S,\qquad
 \forall\boldsymbol{z}\in\mathcal S,\quad
 f(\boldsymbol{x}_{\mathrm{out}})\le f(\boldsymbol{z}).
 \label{eq:properties-selection}
\end{equation}
Thus a regret bound for any one candidate suffices. We now assume
all five candidates exceed a proposed threshold and combine these
comparisons to rule out that possibility; the accompanying
materials\codefootref record the full logical description.

\subsection{From a failed output to inequalities it must satisfy}
\label{sec:reduction}

Fix $0<r\le1/3$ and any normalized bimatrix game. In the
zero-error analysis, suppose for contradiction that the selected
output has regret greater than $r$:
\begin{equation}
 f(\boldsymbol x_{\mathrm{out}})>r.
 \label{eq:failure-hypothesis}
\end{equation}
We derive inequalities that any such run must satisfy. We then
eliminate the strategies and payoff values, leaving conditions on
a few real numbers. If these conditions have no solution, the
assumed failure is impossible.

We obtain these conditions by applying one call's inequality to
strategies produced by another call, then keeping only the
consequences needed later. This is the instantiation-and-forgetting
step of LegoNE~\cite{computeraided2026,legone2026}. For each call,
its returned strategies and comparison coefficients stay fixed
while the strategies being compared vary. The resulting inequalities
may also admit numbers that do not describe any game; ruling out
this larger set still rules out every actual failure. Throughout
this reduction we use the zero-error versions of their comparison
properties:
\begin{equation}
 \delta_f=\tau=0.
 \label{eq:ideal-parameters}
\end{equation}
Section~\ref{sec:accuracy-transfer} restores the approximation errors.
For the fixed game, fix one optimal solution of each LP, including
the strategies and comparison coefficients it returns.
Denote the four $\sm$ candidates and their regrets by
\begin{equation}
 \begin{gathered}
 \begin{aligned}
 \boldsymbol z_1^1&=(y_s^1,x_1^2),&\boldsymbol z_1^2&=(x_1^1,y_s^2),\\
 \boldsymbol z_2^1&=(y_1^1,x_2^2),&\boldsymbol z_2^2&=(x_2^1,y_1^2),
 \end{aligned}\\
 \forall i,j\in[2],\quad F_j^i=f(\boldsymbol z_j^i).
 \end{gathered}
 \label{eq:candidate-profiles}
\end{equation}
Here the superscript on $\boldsymbol z_j^i,F_j^i$ identifies the
player whose input is fixed, and the subscript identifies the round.
This differs from the superscript on a single strategy $x_j^i$ or
$y_j^i$, which identifies the player using that strategy.

Each $\sm$ call supplies a coefficient $\omega$ that balances the
two players' regret terms in~\eqref{eq:properties-dual}. We give
the four calls' coefficients different names:
\begin{equation}
 \begin{aligned}
 \sm(y_s^1)&:\quad\omega=\eta,&\sm(y_s^2)&:\quad\omega=\xi,\\
 \sm(y_1^1)&:\quad\omega=\alpha,&\sm(y_1^2)&:\quad\omega=\gamma.
 \end{aligned}
 \label{eq:call-weight-definitions}
\end{equation}

The algorithm chooses the least-regret candidate. Therefore its
failure means that all four fixed-side candidates are bad; the
stationary regret $g$ also exceeds $r$, and every allowable mixture
by $\om$ must exceed $r$. Formally,
\begin{equation}
 \forall i,j\in[2],\ F_j^i>r,\qquad
 g\ge f(\boldsymbol x_o)>r,\qquad
 \forall p,q\in[0,1],\ f(\boldsymbol x(p,q))>r.
 \label{eq:failure-candidates}
\end{equation}
Here $g=f(\boldsymbol x_s)$ and $\boldsymbol x(p,q)$ mixes the
stationary strategies with the auxiliary strategies returned by
$\spoint$, as defined in Section~\ref{sec:optimalmixing}.
The last condition lets us choose the mixing probabilities $p,q$
later to test the supposed failure. Restricting $r$ to $1/3$ loses no counterexample to a
smaller bound, since any output above $1/3$ also exceeds a threshold
strictly below $1/3$.

We now simplify what a failed run would require. First we show
that the comparison coefficients lie strictly between zero and one;
this makes the returned auxiliary strategies best responses. Next
we bound the payoffs used in later comparisons, combine the two
$\sm$ rounds, and test selected mixtures. Finally we remove the
stationary regret $g$ from the conditions. The remaining five
numbers are $\rho,\eta,\xi,\alpha,\gamma$.

First we need positive coefficients; otherwise the best-response
conclusions in~\eqref{eq:properties-stationary-br}
and~\eqref{eq:properties-sm-br} need not apply. The next lemma
derives positivity from~\eqref{eq:failure-candidates} and uses
complementary slackness to identify both regrets at each $\sm$
candidate.

\begin{lemma}[Positive comparison coefficients]
\label{lem:positive-weights}
For an execution satisfying the properties of
Section~\ref{sec:properties}, suppose that there is no approximation
error as in~\eqref{eq:ideal-parameters} and that the selected output
has regret above $r$ as in~\eqref{eq:failure-hypothesis}.
Then all five comparison coefficients satisfy
\[
 r<\rho,\eta,\xi,\alpha,\gamma<1-r.
\]
The auxiliary strategies from $\spoint$ are best responses to the
opponents' stationary strategies, and every auxiliary strategy from
$\sm$ is a best response to that call's chosen strategy:
\[
 \begin{aligned}
 B_1(x_s^2)&=u_1(y_s^1,x_s^2),&B_2(x_s^1)&=u_2(x_s^1,y_s^2),\\
 \forall j\in[2],\quad B_1(x_j^2)&=u_1(y_j^1,x_j^2),&
 B_2(x_j^1)&=u_2(x_j^1,y_j^2).
 \end{aligned}
\]
Both player regrets at each candidate equal its maximum regret:
\[
 \forall i,j\in[2],\quad f_1(\boldsymbol z_j^i)=f_2(\boldsymbol z_j^i)=F_j^i.
\]
\end{lemma}
\begin{proof}
Recall
\[
 \sigma=1-\rho.
\]
Substitute $(y_s^1,x_s^2)$ and $(x_s^1,y_s^2)$ for the comparison
profile in~\eqref{eq:properties-stationary}, with $\delta_f=0$.
One bracket vanishes in each substitution, leaving
\[
 \begin{aligned}
 g&\le\sigma[u_2(y_s^1,y_s^2)-u_2(y_s^1,x_s^2)],\\
 g&\le\rho[u_1(y_s^1,y_s^2)-u_1(x_s^1,y_s^2)].
 \end{aligned}
\]
Each payoff difference is at most one, so $g>r$ from
\eqref{eq:failure-candidates} implies $r<\rho,\sigma$.
This justifies using the stationary auxiliary strategies as best responses.
Nonnegative payoffs also give
\begin{equation}
 u_1(y_s^1,y_s^2)\ge g/\rho,\qquad
 u_2(y_s^1,y_s^2)\ge g/\sigma.
 \label{eq:main-stationary-payoff-lower-bounds}
\end{equation}
For the first-round call $\sm(y_s^1)$, apply
\eqref{eq:properties-dual}, with $\omega=\eta$, first at a
column best response to $y_s^1$, where the column regret is zero,
and then at $x_s^2$, to which $y_s^1$ is a row best response.
Payoff differences are at most one, so these comparisons yield
\[
 F_1^1\le\eta,\qquad F_1^1\le1-\eta.
\]
Since $F_1^1>r$ in~\eqref{eq:failure-candidates}, both coefficients are
positive. Complementary slackness~\eqref{eq:properties-cs} gives
\[
 \eta(F_1^1-f_1(\boldsymbol z_1^1))=0,\qquad
 (1-\eta)(F_1^1-f_2(\boldsymbol z_1^1))=0,
\]
and $y_1^1$ is a best response to $x_1^2$.
For the second-round call $\sm(y_1^1)$, apply the same comparison
inequality, now with $\omega=\alpha$, at $x_1^2$ and at a column
best response to $y_1^1$. The best-response property just established
gives $F_2^1\le1-\alpha$ at $x_1^2$, and zero column regret at the
other comparison gives $F_2^1\le\alpha$. The corresponding
applications for $\sm(y_s^2)$ and $\sm(y_1^2)$ are symmetric.
The complete substitutions are in Appendix~\ref{app:instantiation}.
\end{proof}

\paragraph{Bounding the cross-profile payoffs.}
The best-response identities from Lemma~\ref{lem:positive-weights}
let us bound the payoffs appearing when we compare strategies from
different calls. Because all payoffs lie in $[0,1]$, each player's
regret satisfies $f_i(\boldsymbol x)\le1-u_i(\boldsymbol x)$.
We therefore bound these payoff shortfalls at three profiles made
from first-round strategies. The six bounds below follow by
applying~\eqref{eq:properties-stationary} and
\eqref{eq:properties-dual}; later comparisons use them in place
of game-specific payoffs.
For compactness, collect the threshold, stationary regret, and
first-round comparison coefficients in
\begin{equation}
 \boldsymbol\lambda=(r,g,\rho,\eta,\xi).
 \label{eq:payoff-bound-parameters}
\end{equation}
For fixed $r$, the other entries of $\boldsymbol\lambda$ can vary.
Each entry in the next table bounds the payoff shortfall $1-u_i$
for the player named by its column, and therefore also that
player's regret. The six symbols $s_1,s_2,h_1,h_2,S,T$ depend
only on $\boldsymbol\lambda$; their formulas are given in
\eqref{eq:payoff-bound-s}, \eqref{eq:payoff-bound-h}, and
\eqref{eq:payoff-bound-st}; Appendix~\ref{app:instantiation}
gives the substitutions.
\begin{equation}
 \begin{array}{c|cc}
 \text{profile}&\text{upper bound on }1-u_1&\text{upper bound on }1-u_2\\ \hline
 (y_1^1,y_s^2)&s_1(\boldsymbol\lambda)&h_2(\boldsymbol\lambda)\\
 (y_s^1,y_1^2)&h_1(\boldsymbol\lambda)&s_2(\boldsymbol\lambda)\\
 (y_1^1,y_1^2)&S(\boldsymbol\lambda)&T(\boldsymbol\lambda)
 \end{array}
 \label{eq:main-payoff-bounds}
\end{equation}
Here the numerical bound $T(\boldsymbol\lambda)$ is distinct from the stationary
expression $T(\boldsymbol x,\boldsymbol x',\rho,\boldsymbol y)$;
the number and types of their arguments distinguish them.
To obtain the last row, we abbreviate two ratios
of the first-round comparison coefficients as
\begin{equation}
 u=\frac{1-\eta}{\eta},\qquad v=\frac{1-\xi}{\xi}.
 \label{eq:main-relative-weights}
\end{equation}
Applying~\eqref{eq:properties-dual} for $\sm(y_s^1)$ at
$y_1^2$ and its column-player counterpart for $\sm(y_s^2)$ at
$y_1^1$ gives
\begin{equation}
 S(\boldsymbol\lambda)=h_1(\boldsymbol\lambda)-r+u(s_2(\boldsymbol\lambda)-r),\qquad
 T(\boldsymbol\lambda)=h_2(\boldsymbol\lambda)-r+v(s_1(\boldsymbol\lambda)-r),
 \label{eq:main-payoff-st}
\end{equation}
We have now bounded the payoff shortfalls needed below in terms of
$r,g$ and the first-round coefficients $\rho,\eta,\xi$. The second-round
coefficients enter when we compare the two rounds.

\paragraph{Using both $\sm$ rounds together.}
The auxiliary strategy returned by the first round becomes the
input of the second. Applying both calls' inequalities to the same
comparison strategy gives a lower and an upper bound on one payoff.
Eliminating that payoff will link the rounds. To do so, apply
the comparison inequalities for the first-round call $\sm(y_s^2)$ and
the second-round call $\sm(y_1^2)$ at the same arbitrary row
strategy $x^{1\prime}\in\Delta_{n_1}$. Their coefficients are
$\xi$ and $\gamma$, respectively, by~\eqref{eq:call-weight-definitions}.
The column-player version
of~\eqref{eq:properties-dual}, together with
\eqref{eq:failure-candidates} and payoff bounds in $[0,1]$, gives
\begin{equation}
 \begin{aligned}
 \forall x^{1\prime}\in\Delta_{n_1},\quad
 \xi u_2(x^{1\prime},y_1^2)
   &\ge r-(1-\xi)+(1-\xi)u_1(x^{1\prime},y_s^2),\\
 \forall x^{1\prime}\in\Delta_{n_1},\quad
 \gamma u_2(x^{1\prime},y_1^2)
   &\le1-r-(1-\gamma)u_1(x^{1\prime},y_1^2).
 \end{aligned}
 \label{eq:main-shared-dual-instantiations}
\end{equation}
Multiply the first inequality by $\gamma/\xi>0$ and substitute
its lower bound for the common payoff into the second inequality.
Using $v$ from~\eqref{eq:main-relative-weights}, this eliminates
$u_2(x^{1\prime},y_1^2)$ and gives
\begin{equation}
 \forall x^{1\prime}\in\Delta_{n_1},\quad
 (1-\gamma)u_1(x^{1\prime},y_1^2)+\gamma v u_1(x^{1\prime},y_s^2)
 \le1-r-\frac\gamma\xi[r-(1-\xi)].
 \label{eq:main-shared-comparison}
\end{equation}
This inequality holds for every row strategy, so it bounds the row
best-response payoff against a mixture of $y_1^2$ and $y_s^2$.
Define a mixed column strategy to use as a comparison in the proof:
\begin{equation}
 D_E=1-\gamma+\gamma v>0,\qquad
 y_E^2=\frac{(1-\gamma)y_1^2+\gamma v y_s^2}{D_E}\in\Delta_{n_2}.
 \label{eq:main-column-comparison-mixture}
\end{equation}
By payoff linearity and the definition of $B_1$, inequality
\eqref{eq:main-shared-comparison} implies
\begin{equation}
 B_1(y_E^2)\le
 \frac{1-r-\frac\gamma\xi[r-(1-\xi)]}{D_E}.
 \label{eq:main-column-mixture-br-bound}
\end{equation}
We can therefore apply the comparison inequality
\eqref{eq:properties-dual} for the second-round call
$\sm(y_1^1)$ at $y_E^2$, since that inequality holds throughout
the column simplex. This call has coefficient $\alpha$.
Applying its inequality at $y_s^2$ and $y_1^2$, and that
of $\sm(y_1^2)$ with coefficient $\gamma$ at $y_s^1$ and $y_1^1$,
using~\eqref{eq:main-payoff-bounds}, gives four upper bounds on
second-round candidate regret. Since all those candidates exceed
$r$ under the failure assumption, these bounds must be at least $r$:
\begin{equation}
 \begin{aligned}
 U_g&=\alpha s_1(\boldsymbol\lambda)+(1-\alpha)h_2(\boldsymbol\lambda)\ge r,&
 H_g&=\alpha S(\boldsymbol\lambda)+(1-\alpha)T(\boldsymbol\lambda)\ge r,\\
 V_g&=(1-\gamma)h_1(\boldsymbol\lambda)+\gamma s_2(\boldsymbol\lambda)\ge r,&
 K_g&=(1-\gamma)S(\boldsymbol\lambda)+\gamma T(\boldsymbol\lambda)\ge r.
 \end{aligned}
 \label{eq:main-weighted-regret-bounds}
\end{equation}
These four expressions depend on $\boldsymbol\lambda$ and the
second-round coefficients $\alpha,\gamma$.
At $y_E^2$, the row player's best-response term in the second-round
comparison is bounded by~\eqref{eq:main-column-mixture-br-bound}.
Because $y_E^2$ mixes $y_1^2$ and $y_s^2$, its other payoff terms
are weighted averages of terms already bounded in $H_g$ and $U_g$.
Collecting these bounds gives
\begin{equation}
 F_2^1\le
 \frac{(1-\gamma)H_g+\gamma vU_g
       +\alpha[\gamma(1-r-rv)-r]}{D_E}.
 \label{eq:main-mixture-dual-bound}
\end{equation}
Since $F_2^1>r$, multiplication by $D_E$ yields the further
necessary inequality
\begin{equation}
 E_{1,g}=(1-\gamma)(H_g-r)+\gamma v(U_g-r)
               +\alpha\gamma(1-r-rv)-\alpha r\ge0.
 \label{eq:main-shared-coupling}
\end{equation}
The shared inequality $E_{1,g}\ge0$ is therefore a new necessary
condition for failure: it links both players' $\sm$ calls across
the two rounds. Its remaining variables are the threshold,
stationary regret, and all five comparison coefficients.
Appendix~\ref{app:instantiation} gives
the substitution before regrouping terms and its symmetric counterpart.

\paragraph{Using $\om$ to test selected mixtures.}
The preceding inequalities connect the four $\sm$ calls. We now
use $\om$: every profile $\boldsymbol x(p,q)$ must have regret above
$r$ under~\eqref{eq:failure-candidates}. We derive upper bounds on
both players' regrets at these profiles. If both bounds are at most
$r$ for even one choice of $p,q$, the assumed failure is impossible.
The bounds have simple linear pieces, so we will test four selected
segments rather than carry $p,q$ through the rest of the proof.

By~\eqref{eq:mixing-regrets}, each mixed regret equals an average
of endpoint regrets minus the Jensen gap explained in
Section~\ref{sec:properties}. The gap reflects the fact that one
response strategy must serve against both endpoint opponents.
To bound the column player's gap from below,
consider the regrets of the same column strategy $x^{2\prime}$
against the two row strategies $x_s^1$ and $y_s^1$:
\begin{equation}
 A=f_2(x_s^1,x^{2\prime}),\qquad
 B=f_2(y_s^1,x^{2\prime}).
 \label{eq:main-column-payoff-losses}
\end{equation}
Both numbers are nonnegative, and the same strategy is used at both
endpoints. Apply stationarity~\eqref{eq:properties-stationary}
at $(x_s^1,x^{2\prime})$, then substitute the resulting payoff
lower bound into the comparison inequality~\eqref{eq:properties-dual}
for the first-round call $\sm(y_s^1)$ at $x^{2\prime}$, with
$\omega=\eta$. This gives
\begin{equation}
 \forall x^{2\prime}\in\Delta_{n_2},\quad
 \frac{\eta\sigma}{\rho}A+(1-\eta)B\ge\ell_c(\boldsymbol\lambda),
 \qquad \ell_c(\boldsymbol\lambda)=r-\eta(1-g/\rho).
 \label{eq:main-column-loss-comparison}
\end{equation}
The function $\ell_c$ is the lower bound obtained by these two
comparisons; the index identifies the column player.
If it is positive, every column strategy has positive regret against
at least one of the two row strategies. The Jensen-gap identity
\eqref{eq:mixing-gap-losses} takes the weighted minimum of these same
regrets:
\[
 J_2(p)=\min_{x^{2\prime}\in\Delta_{n_2}}\{(1-p)A+pB\}.
\]
The following lemma therefore converts the joint regret comparison
into the correction needed in the mixed profile's regret bound.

\begin{lemma}[A lower bound on the gain from mixing]
\label{lem:main-jensen}
Consider an execution satisfying the building-block properties of
Section~\ref{sec:properties} with zero errors~\eqref{eq:ideal-parameters},
and suppose that its selected output has regret above $r$ as in
\eqref{eq:failure-hypothesis}. Then, for every $p\in[0,1]$,
the column Jensen gap satisfies
\[
 J_2(p)\ge\ell_c(\boldsymbol\lambda)
       \min\left\{\frac{(1-p)\rho}{\eta\sigma},\frac p{1-\eta}\right\}.
\]
There is a symmetric bound for the row Jensen gap $J_1(q)$.
\end{lemma}
\begin{proof}
Set $c_1=\eta\sigma/\rho$, $c_2=1-\eta$, and
$M=\min\{(1-p)/c_1,p/c_2\}$. Both $c_1,c_2$ are positive.
For each comparison strategy, nonnegativity of $A,B$ gives
\[
 (1-p)A+pB
 \ge M(c_1A+c_2B)\ge M\ell_c(\boldsymbol\lambda).
\]
Minimizing the left side over $x^{2\prime}$ gives $J_2(p)$ by its
regret identity in Section~\ref{sec:optimalmixing}.
Thus a joint restriction on endpoint best responses becomes a
reduction in the regret of the mixed profile. The inequality remains
valid if $\ell_c(\boldsymbol\lambda)$ is nonpositive.
\end{proof}

Subtracting these Jensen-gap bounds from the endpoint interpolation gives
regret upper bounds $R_g(p,q)$ and $C_g(p,q)$ for the row and column
players. They depend on $(\boldsymbol\lambda,p,q)$.
For the row player, define the symmetric Jensen-gap lower bound by
\[
 \ell_r(\boldsymbol\lambda)=r-\xi(1-g/\sigma),\qquad
 j_r(\boldsymbol\lambda,q)=\ell_r(\boldsymbol\lambda)\min\left\{\frac{(1-q)\sigma}{\xi\rho},\frac q{1-\xi}\right\}.
\]
Using $J_1(q)\ge j_r(\boldsymbol\lambda,q)$ gives
\[
 R_g(p,q)=(1-p)(1-q)g+q(1-pg/\rho)-j_r(\boldsymbol\lambda,q).
\]
The first two terms interpolate the four row-regret bounds
$g,0,1,1-g/\rho$; the last term is the improvement from the Jensen gap.
For the column player, the corresponding formulas are
\[
 \begin{aligned}
 j_c(\boldsymbol\lambda,p)
   &=\ell_c(\boldsymbol\lambda)
       \min\left\{\frac{(1-p)\rho}{\eta\sigma},\frac p{1-\eta}\right\},\\
 C_g(p,q)
   &=(1-p)(1-q)g+p(1-qg/\sigma)-j_c(\boldsymbol\lambda,p).
 \end{aligned}
\]
Both formulas and the two Jensen-gap bounds are also stated in
\eqref{eq:rectangle-upper-bounds} and~\eqref{eq:gap-lower-bounds};
Appendix~\ref{app:mixing} gives the derivation.
Since both actual regrets are bounded above by these functions,
failure requires
\[
 \forall p,q\in[0,1],\quad \max\{R_g(p,q),C_g(p,q)\}>r.
\]
Each Jensen-gap bound uses the minimum of two affine expressions.
We choose comparison coefficients where these expressions agree:
\[
 \frac{(1-p)\rho}{\eta\sigma}=\frac p{1-\eta},\qquad
 \frac{(1-q)\sigma}{\xi\rho}=\frac q{1-\xi}.
\]
Solving these equations gives
\[
 p_\dagger=\frac{\rho(1-\eta)}{\rho(1-\eta)+\sigma\eta},\qquad
 q_\dagger=\frac{\sigma(1-\xi)}{\sigma(1-\xi)+\rho\xi}.
\]
It is enough for our contradiction to keep four necessary
\emph{segment conditions} along these cuts:
\begin{equation}
 \begin{aligned}
 \min_{0\le q\le q_\dagger}\max\{R_g(p_\dagger,q),C_g(p_\dagger,q)\}&\ge r,\\
 \min_{q_\dagger\le q\le1}\max\{R_g(p_\dagger,q),C_g(p_\dagger,q)\}&\ge r,\\
 \min_{0\le p\le p_\dagger}\max\{R_g(p,q_\dagger),C_g(p,q_\dagger)\}&\ge r,\\
 \min_{p_\dagger\le p\le1}\max\{R_g(p,q_\dagger),C_g(p,q_\dagger)\}&\ge r.
 \end{aligned}
 \label{eq:main-segment-conditions}
\end{equation}
Both regret bounds are affine on each of these segments, so each
minimum occurs at an endpoint or an intersection of the two bounds.
Symbolic calculation with SymPy~\cite{sympy2017} constructs the
endpoint and intersection inequalities, and an exact machine check
reconstructs their identities. This turns
\eqref{eq:main-segment-conditions} into finite inequalities and
eliminates the mixing parameters.
Appendix~\ref{app:mixing} expands these inequalities.

\paragraph{Removing the stationary regret.}
The payoff bounds have removed the individual payoff variables, and
the segment tests have removed the mixing parameters $p,q$. For
fixed $r$, only $g$ and the five comparison coefficients
$\rho,\eta,\xi,\alpha,\gamma$ remain. We now remove $g$ as well.
Every assumed failure has $g>r$. The next lemma shows that its
necessary inequalities still hold if we substitute $g=r$ while
keeping those coefficients fixed. This is an algebraic substitution
in the inequalities; it does not require another game or algorithm
execution.

\begin{lemma}[Eliminating the stationary regret]
\label{lem:main-g-projection}
Suppose an execution satisfies the properties of
Section~\ref{sec:properties}, with zero errors~\eqref{eq:ideal-parameters}
and output regret above $r$ as in~\eqref{eq:failure-hypothesis},
where $3/10\le r\le1/3$.
Hold $r,\rho,\eta,\xi,\alpha,\gamma$ fixed and replace $g$ by $r$
in the numerical bounds. The six payoff upper bounds remain
nonnegative; the second-round inequalities
\eqref{eq:main-weighted-regret-bounds}, the shared inequality
\eqref{eq:main-shared-coupling}, and its symmetric counterpart
remain valid, as do the four segment conditions
\eqref{eq:main-segment-conditions}.
\end{lemma}
\begin{proof}[Proof idea]
We present only the main proof ideas here;
Lemma~\ref{lem:app-g-projection} in Appendix~\ref{app:scalar}
provides the complete proof of Lemma~\ref{lem:main-g-projection}.
Applying the comparison inequalities for $\sm(y_s^1)$ at $y_s^2$
and for $\sm(y_s^2)$ at $y_s^1$, and using the stationary payoff
bounds~\eqref{eq:main-stationary-payoff-lower-bounds}, gives
\begin{equation}
 \begin{aligned}
 r&\le\eta(1-g/\rho)+(1-\eta)(1-g/\sigma),\\
 r&\le\xi(1-g/\sigma)+(1-\xi)(1-g/\rho).
 \end{aligned}
 \label{eq:main-first-round-g-comparisons}
\end{equation}
The right sides decrease with $g$, so these inequalities bound
the admissible interval for $g$ from above. Together with $g\ge r$,
they specify the interval on which we prove preservation.

Along this interval, each of the six payoff upper bounds increases
or stays unchanged as $g$ decreases. For example,
\[
 \frac{\partial s_1(\boldsymbol\lambda)}{\partial g}=-\frac1\rho-\frac u\sigma<0.
\]
The second-round and shared inequalities combine these bounds with
nonnegative coefficients, so they continue to hold.

For the four mixing conditions, we use their finite endpoint and
intersection cases. On the two segments ending at a mixing
coefficient equal to one, both regret upper bounds increase as
$g$ decreases.

For the segment $p=p_\dagger$, $0\le q\le q_\dagger$, the
first-round comparisons~\eqref{eq:main-first-round-g-comparisons}
give
\[
 C_g(p_\dagger,0)\ge R_g(p_\dagger,0).
\]
The column bound at this outer endpoint and both bounds at the
central endpoint increase as $g$ decreases, so the endpoint
conditions are preserved. Along the segment, the row bound increases
with $q$ and the column bound decreases. If they intersect inside
the segment, clearing the positive difference of their slopes
shows that the common value is at least $r$ exactly when
$\mathcal D(g)\ge0$, where
\begin{equation}
 \begin{aligned}
 \mathcal D(g)=\frac1{q_\dagger}\bigl\{&
 [C_g(p_\dagger,0)-r][R_g(p_\dagger,q_\dagger)-r]\\
 &-[R_g(p_\dagger,0)-r][C_g(p_\dagger,q_\dagger)-r]\bigr\}.
 \end{aligned}
 \label{eq:main-g-intersection-condition}
\end{equation}
We use an exact machine check to establish
\[
 \frac{\partial\mathcal D(g)}{\partial g}<0
\]
on the admissible parameter ranges. After clearing positive
denominators, this becomes a polynomial sign condition, verified
exactly as detailed in Appendix~\ref{app:scalar}~\cite{garloff1986}.
That appendix gives the polynomial, domain, and checker.
Thus $\mathcal D(g)$ increases
as $g$ decreases, preserving the intersection condition.
An intersection enters or leaves the segment through an endpoint,
whose condition is already preserved. Exchanging the two players
gives the segment ending at $p=0$. All four finite segment conditions
therefore continue to hold down to $g=r$.
\end{proof}

After substituting $g=r$, write the remaining parameters as
\begin{equation}
 \boldsymbol\lambda_r=(r,r,\rho,\eta,\xi).
 \label{eq:projected-payoff-parameters}
\end{equation}
We now simplify the first-round regret bounds to obtain explicit
formulas. Applying the comparison inequalities of $\sm(y_s^1)$
at $x_s^2$ and of $\sm(y_s^2)$ at $x_s^1$ gives
\begin{equation}
 f_1(y_1^1,x_s^2)\le\min\{1,(1-\eta-r)/\eta\},\qquad
 f_2(x_s^1,y_1^2)\le\min\{1,(1-\xi-r)/\xi\}.
 \label{eq:main-first-round-regret-bounds}
\end{equation}
Replacing the two right sides by $(1-\eta-r)/\eta$ and
$(1-\xi-r)/\xi$ gives larger upper bounds.
Combining them with stationarity~\eqref{eq:properties-stationary}
at $(y_1^1,x_s^2)$ and $(x_s^1,y_1^2)$ gives relaxed bounds
for $h_1,h_2$.
From this point on, the bare symbols $s_1,s_2,h_1,h_2,S,T$
denote the functions after substituting $g=r$ and making this
relaxation:
\begin{equation}
 \begin{aligned}
 s_1&=1-r-r/\rho+u(1-r-r/\sigma),\\
 s_2&=1-r-r/\sigma+v(1-r-r/\rho),\\
 h_1&=\min\{1,1-r(1+\sigma)/\rho+\sigma(1-r)v/\rho\},\\
 h_2&=\min\{1,1-r(1+\rho)/\sigma+\rho(1-r)u/\sigma\},\\
 S&=h_1-r+u(s_2-r),\qquad T=h_2-r+v(s_1-r).
 \end{aligned}
 \label{eq:main-projected-payoff-bounds}
\end{equation}
These functions depend on $r,\rho,\eta,\xi$.
Using them in~\eqref{eq:main-weighted-regret-bounds} and
\eqref{eq:main-shared-coupling} defines the corresponding expressions
$U,V,H,K,E_1$ after the relaxation. The mixing bounds become
$R_r(p,q)=R_g(p,q)|_{g=r}$ and $C_r(p,q)=C_g(p,q)|_{g=r}$.

The next proposition states the implication we need. A failed run
would assign values to $\rho,\eta,\xi,\alpha,\gamma$ that satisfy
all the remaining necessary inequalities. Call this finite
collection $\Sigma_r$. Proving that no five values satisfy it will
rule out the assumed failure. The complete list is
\eqref{eq:scalar-system} in Appendix~\ref{app:scalar}.

\begin{proposition}[Necessary inequalities for a failed run]
\label{prop:finite-system}
There is an explicit collection of inequalities $\Sigma_r(\boldsymbol\theta)$ in the
five real variables
\[
 \boldsymbol\theta=(\rho,\eta,\xi,\alpha,\gamma)
\]
such that, for an execution with the properties of
Section~\ref{sec:properties} and zero errors~\eqref{eq:ideal-parameters},
\[
 \forall r\in[3/10,1/3],\quad
 f(\boldsymbol x_{\mathrm{out}})>r
 \ \Longrightarrow\ \exists\boldsymbol\theta,\ \Sigma_r(\boldsymbol\theta).
\]
The collection consists of coefficient bounds, the payoff bounds and shared
comparisons derived above, and the four finite segment conditions.
\end{proposition}
\subsection{Finding and defining the candidate threshold}
\label{sec:bound-discovery}

Proposition~\ref{prop:finite-system} has reduced a failed run to
five real numbers $\boldsymbol\theta=(\rho,\eta,\xi,\alpha,\gamma)$
that must satisfy $\Sigma_r$; strategies, payoffs, the mixing
probabilities $p,q$, and $g$ no longer appear. We now search for
where these inequalities come closest to holding at a large $r$.
Numerical exploration suggests the comparisons responsible;
exact equations then define a candidate threshold $r_*$. This
subsection defines that candidate. The next one proves that no
other assignment satisfies $\Sigma_r$ above it.
The exploratory problem is
\begin{equation}
 \begin{aligned}
 \operatorname*{maximize}_{r,\boldsymbol\theta}&\quad r\\
 \text{subject to}&\quad 3/10\le r\le1/3,\quad
                         \Sigma_r(\boldsymbol\theta).
 \end{aligned}
 \label{eq:main-bound-search}
\end{equation}
We explore~\eqref{eq:main-bound-search} from several initial points
using SciPy's SLSQP routine~\cite{scipy2020,kraft1988}. The search
suggests a threshold near $0.30954$ and identifies comparisons close
to equality; Appendix~\ref{app:discovery} records the script and
these observations. They guide the following algebraic definition,
but the numerical search is not used to prove the bound.

Near the numerical candidate, two second-round comparisons equal
$r$, the payoff bound $h_1$ reaches its upper limit of one, and
both a selected mixing comparison and the comparison linking the
two $\sm$ rounds become equalities. The first three observations
let us express $\alpha,\gamma,\xi$ using fewer variables. The
remaining two equalities, plus a condition identifying the top of
their local solution curve, define $r_*$ exactly.

First consider the two second-round comparisons at the profile
$(y_1^1,y_1^2)$ formed by the auxiliary strategies returned in the
first round. They are
\eqref{eq:main-weighted-regret-bounds} after the projection and
relaxation:
\begin{equation}
 H=\alpha S+(1-\alpha)T\ge r,\qquad
 K=(1-\gamma)S+\gamma T\ge r.
 \label{eq:main-discovery-hk}
\end{equation}
Exchanging the players if necessary, we describe the numerical
candidate in the orientation $S>r>T$. Both $H-r$ and $K-r$
are close to zero. We therefore impose $H=r$ and $K=r$ in the
local algebraic calculation. Solving these two equations determines
$\alpha$ and $\gamma$. Denote the resulting value of $\alpha$ by $A_0$:
\begin{equation}
 A_0=\frac{r-T}{S-T},\qquad \alpha=A_0,\qquad\gamma=1-A_0.
 \label{eq:main-discovery-weights}
\end{equation}
The denominator is positive because $S>r>T$.
Thus $A_0$ is the coefficient that makes the weighted average of
$S$ and $T$ equal to $r$.

Another observed equality determines the first-round weight $\xi$.
Recall the upper-bound function from~\eqref{eq:main-projected-payoff-bounds}:
\begin{equation}
 h_1=\min\left\{1,\,
  1-\frac{r(1+\sigma)}\rho+\frac{\sigma(1-r)}\rho v\right\},
 \qquad \sigma=1-\rho,\quad v=\frac{1-\xi}{\xi}.
 \label{eq:main-discovery-h1}
\end{equation}
The numerical points approach the place where this minimum changes
which formula it uses. Equating the second expression to one and
solving for $v$ gives that formula-change value, denoted by $v_0$:
\begin{equation}
 v_0=\frac{r(1+\sigma)}{\sigma(1-r)},\qquad
 v=v_0,\qquad \xi=\frac{\sigma(1-r)}{\sigma+r}.
 \label{eq:main-discovery-cap-switch}
\end{equation}
On the line $v=v_0$, $h_1=1$, and $\xi$ is determined by $r,\rho$.
The same numerical points have $h_2<1$. For the local calculation
we therefore use the second expression in the minimum defining $h_2$:
\begin{equation}
 h_2=1-\frac{r(1+\rho)}\sigma
              +\frac{\rho(1-r)}\sigma u,\qquad u=\frac{1-\eta}{\eta}.
 \label{eq:main-discovery-h2}
\end{equation}
The formulas for $s_1,s_2,S,T$ remain those in
\eqref{eq:main-projected-payoff-bounds}, with these substitutions.

We have now expressed $\xi,\alpha,\gamma$ in terms of
$(r,\rho,\eta)$. Two further equalities appear in the numerical
outputs. One is $E_1=0$ in the shared comparison
\eqref{eq:main-shared-coupling}, after the projection and relaxation.
The other comes from the segment $p=p_\dagger$, $0\le q\le q_\dagger$
in~\eqref{eq:main-segment-conditions}: the two affine regret upper
bounds intersect at value $r$. By~\eqref{eq:main-g-intersection-condition},
this intersection equality is $\mathcal D(r)=0$.
Substituting~\eqref{eq:main-discovery-weights}--\eqref{eq:main-discovery-h2}
into these two equalities leaves rational equations in $(r,\rho,\eta)$.
Clearing their positive denominators gives polynomial equations
$N_1(r,\rho,\eta)=0$ and $N_2(r,\rho,\eta)=0$:
$N_1$ comes from the segment intersection, and $N_2$ from $E_1=0$.
Equation~\eqref{eq:root-polys} in Appendix~\ref{app:discovery}
specifies the positive factors that fix their exact normalization.

The equations $N_1=N_2=0$ define a curve near the candidate in
coordinates $(r,\rho,\eta)$, with linearly independent constraint
gradients. We want its locally largest value of $r$. At such a
point, the Lagrange
multiplier condition~\cite{nocedal2006} gives real
multipliers $\nu_1,\nu_2$ such that
\begin{equation}
 \nabla r=(1,0,0)=\nu_1\nabla N_1+\nu_2\nabla N_2.
 \label{eq:main-lagrange-condition}
\end{equation}
The vector $\nabla N_1\times\nabla N_2$ is orthogonal to both
constraint gradients. Taking its dot product with both sides
of~\eqref{eq:main-lagrange-condition} sets the vector's
$r$-coordinate to zero, giving the third equation
\begin{equation}
 N_1=N_2=0,\qquad
 H_{\mathrm{root}}=(N_1)_\rho(N_2)_\eta-(N_1)_\eta(N_2)_\rho=0,
 \label{eq:main-root-system}
\end{equation}
where subscripts indicate partial derivatives. We define
$(r_*,\rho_*,\eta_*)$ as the unique solution of~\eqref{eq:main-root-system} in the rational isolating box
\eqref{eq:app-root-box}. Appendix~\ref{app:discovery} gives the
polynomials and the checker for this root. Here an exact machine
check uses rational interval arithmetic and the Krawczyk
method~\cite{moore2009,krawczyk1969} to prove
existence and uniqueness in that box, and gives
\[
 (r_*,\rho_*,\eta_*)\approx
 (0.3095399654,\ 0.4799488361,\ 0.5461144474),\qquad
 0.30953996<r_*<0.30953997.
\]
This certifies the local algebraic root suggested by the equalities
seen in the search. It does not yet bound the algorithm: the
inequalities $\Sigma_r$ could have a solution far from this root.
We must exclude every such solution for $r>r_*$.

\subsection{Ruling out every failure above the threshold}
\label{sec:guarantee}

Fix $r_*<r\le1/3$. The previous subsection identified $r_*$ from
equalities near one candidate assignment. We must now rule out
\emph{every} assignment satisfying $\Sigma_r$, including ones far
from that candidate. First, we show that any such assignment would
still satisfy the inequalities needed below after setting the two
second-round coefficients to the values that balance $S$ and $T$.
That leaves $\rho,\eta,\xi$. Next, we move their values, while
preserving the needed inequalities, to the line $v=v_0$ where the
formula for $h_1$ changes. On that line, the comparison linking the
two $\sm$ rounds conflicts with one of two mixing tests. Thus
\[
 \Sigma_r\text{ has a solution}
 \quad\Longrightarrow\quad
 \text{the necessary inequalities hold on }v=v_0
 \quad\Longrightarrow\quad
 \text{contradiction}.
\]
The algebraic root $r_*$ enters in the final exclusion on that line;
the preceding reductions cover the entire domain.

Recall $\sigma=1-\rho$ and the relative weights
$u=(1-\eta)/\eta$, $v=(1-\xi)/\xi$ from
\eqref{eq:main-relative-weights}. The bounds on $\eta,\xi$ in
$\Sigma_r$ become
\begin{equation}
 k=\frac r{1-r},\qquad k\le u,v\le1/k.
 \label{eq:main-global-weight-domain}
\end{equation}
The moves below change numerical values satisfying necessary
inequalities; they need not correspond to new games. Each move
preserves the inequalities used at its destination. If no
destination is possible, no starting assignment is possible.

\paragraph{Removing the two second-round coefficients.}
The comparisons $H,K\ge r$ in~\eqref{eq:main-discovery-hk} are
convex combinations of $S,T$, so at least one of $S,T$ is at least
$r$. The payoff and segment conditions exclude $S,T\ge r$
simultaneously. If $S=r>T$, then $H<r$ because $\alpha<1$;
the symmetric case contradicts $K\ge r$. Hence, exchanging both
players' labels if necessary, we can arrange $S>r>T$.
This exchange also exchanges $\rho,\sigma$, $u,v$, and $\alpha,\gamma$;
it leaves $\rho$ free to range over $[r,1-r]$.
The exclusion of $S,T\ge r$ uses exact machine checks of polynomial
signs together with the payoff comparisons. Appendix~\ref{app:preliminary}
gives the complete argument, including the bound $s_1,s_2<r$,
and identifies the checkers at the sign conditions they verify.

In this orientation, $H,K\ge r$ restrict the two second-round
weights to opposite sides of the balancing value $A_0$ defined
in~\eqref{eq:main-discovery-weights}. Define the positive distances
from the threshold by
\begin{equation}
 m=S-r>0,\qquad n=r-T>0,\qquad
 A_0=\frac n{m+n}=\frac{r-T}{S-T}.
 \label{eq:main-weight-balancing-values}
\end{equation}
Then
\begin{equation}
 H-r=\alpha(m+n)-n,\qquad K-r=m-\gamma(m+n),
 \qquad \alpha\ge A_0,\quad\gamma\le1-A_0.
 \label{eq:main-weight-source-restrictions}
\end{equation}
The numerical discovery used equality in these two comparisons.
For the global proof, we must justify replacing arbitrary allowed
$\alpha,\gamma$ by those balancing values while retaining the
inequalities needed for the contradiction. The next lemma does so.

\begin{lemma}[Eliminating the second-round coefficients]
\label{lem:weight-reduction}
For $r_*<r\le1/3$, any solution of $\Sigma_r$ can, after exchanging
both players' labels if necessary, be changed to an assignment with
\[
 S>r>T,\qquad \alpha=A_0,\qquad\gamma=1-A_0,
 \qquad 0<A_0\le1-r.
\]
Changing only $\alpha,\gamma$ preserves the comparison $U\ge r$
from~\eqref{eq:main-weighted-regret-bounds} and the shared comparison
$E_1\ge0$ from~\eqref{eq:main-shared-coupling}, with the projected
functions of~\eqref{eq:main-projected-payoff-bounds}.
The six payoff bounds in~\eqref{eq:main-projected-payoff-bounds}
and the four segment conditions in~\eqref{eq:main-segment-conditions}
are unchanged.
\end{lemma}
\begin{proof}[Proof idea]
We present only the main proof ideas here; the complete proof of
Lemma~\ref{lem:weight-reduction} is given in Appendix~\ref{app:weights}.
We first lower $\alpha$ to $A_0$. Differentiating the projected
expression in~\eqref{eq:main-shared-coupling} gives
\begin{equation}
 \partial_\alpha E_1
 =(1-\gamma)(S-T)+\gamma(1+T-h_2/\xi)-r<0
 \qquad(r\le\gamma\le1-A_0).
 \label{eq:main-weight-derivative}
\end{equation}
To prove the sign, it suffices to check the two endpoints in
$\gamma$, since this derivative is affine in $\gamma$.
Here exact machine checks establish the endpoint signs and the
bound $h_2>\xi\ge r$, using the original nonnegative payoff bounds,
$H\ge r$, and $\alpha\le1-r$.
Appendix~\ref{app:weights} specifies the domains and code paths
for these checks.
Thus lowering $\alpha$ increases $E_1$.
It also increases $U=\alpha s_1+(1-\alpha)h_2$, since $s_1<r<h_2$.
The inequality $A_0\le\alpha\le1-r$ gives the stated upper bound
on $A_0$.

With $\alpha=A_0$, we have $H=r$ and
\[
 E_1(A_0,0)=-A_0r<0.
\]
The function $E_1(A_0,\gamma)$ is affine in $\gamma$ and
nonnegative at the current positive $\gamma$. Its slope is therefore
positive. Raising $\gamma$ to $1-A_0$ preserves $E_1\ge0$ and
leaves $U$ unchanged. The payoff functions and mixing comparisons
are independent of these two weights.
\end{proof}
Lemma~\ref{lem:weight-reduction} has moved any putative solution to
the balancing values of $\alpha,\gamma$. Only three comparisons
will be needed in the remaining reduction: $U\ge r$, the bound
$A_0\le1-r$, and the shared comparison $E_1\ge0$. We rewrite them
so their signs can be tracked as the other variables move.
Let $z=n/m>0$ be the ratio of the two distances from $r$, and set
$c=1-2r>0$.
Multiply $U(A_0)-r\ge0$ and $1-r-A_0\ge0$ by $m+n>0$,
and divide $E_1(A_0,1-A_0)\ge0$ by $(1-A_0)^2>0$.
Define the resulting nonnegative expressions by
\begin{equation}
 \begin{aligned}
 G&=(m+n)[U(A_0)-r]=m(h_2-r)-n(r-s_1)\ge0,\\
 B_{\mathrm{mass}}&=(m+n)[1-r-A_0]=(1-r)m-rn\ge0,\\
 E&=\frac{E_1(A_0,1-A_0)}{(1-A_0)^2}
   =v(h_2-r)+z[c+v(s_1-2r)]-rz^2\ge0.
 \end{aligned}
 \label{eq:main-reduced-comparisons}
\end{equation}
Thus $G\ge0$ is $U\ge r$ rewritten, $B_{\mathrm{mass}}\ge0$ is
$A_0\le1-r$ rewritten, and $E\ge0$ is the shared comparison
rewritten. We use these three inequalities in the paths below;
the other weighted comparisons are not needed there.

\paragraph{Selecting the mixing tests for the boundary paths.}
After eliminating $\alpha,\gamma$, we need mixing conditions that
remain useful as $u,v$ change. The four segment conditions
\eqref{eq:main-segment-conditions} are still available. We use the
two segments joining $(p_\dagger,q_\dagger)$ to
$(p_\dagger,0)$ and $(0,q_\dagger)$.
At the common endpoint $(p_\dagger,q_\dagger)$, we first need to
know which player's regret upper bound can exceed $r$. The factor
$D_{\mathrm{mix}}$ below is positive, so multiplying by it does not
change signs; $N_R,N_C$ are the resulting differences from $r$:
\begin{equation}
 \begin{aligned}
 D_{\mathrm{mix}}&=(\rho u+\sigma)(\sigma v+\rho)>0,\\
 N_R&=D_{\mathrm{mix}}[R_r(p_\dagger,q_\dagger)-r],\qquad
 N_C=D_{\mathrm{mix}}[C_r(p_\dagger,q_\dagger)-r].
 \end{aligned}
 \label{eq:main-mixing-numerators}
\end{equation}
To keep the two line-segment tests short, write
$\chi_1=1-r-r/\rho$ and $\chi_2=1-r-r/\sigma$, the coefficients
in $s_1=\chi_1+\chi_2u$ and $s_2=\chi_2+\chi_1v$.
Along either selected segment, the two regret upper bounds are
affine and may cross. After fixing the endpoint signs, the next
expressions test whether such a crossing can stay at or above $r$;
positive factors have been removed without changing their signs:
\begin{equation}
 L=(\chi_1+cu)N_R+ruN_C,\qquad
 B_{\mathrm{mix}}=(\chi_2+cv)N_C+rvN_R.
 \label{eq:main-mixing-determinants}
\end{equation}
Here $L$ corresponds to the segment toward $(p_\dagger,0)$ and
$B_{\mathrm{mix}}$ to the segment toward $(0,q_\dagger)$.
The endpoint and intersection conditions imply
\begin{equation}
 (N_R\ge0\ \land\ L\ge0)
 \quad\text{or}\quad
 (N_R<0\le N_C\ \land\ B_{\mathrm{mix}}\ge0).
 \label{eq:main-mixing-branch-guards}
\end{equation}
If $N_R\ge0$, the first segment supplies $L\ge0$; if $N_R<0$,
the central comparison requires $N_C\ge0$, and the second segment
supplies $B_{\mathrm{mix}}\ge0$.
Appendix~\ref{app:preliminary} expands these substitutions and
checks the endpoint signs. The paths use the weaker first
alternative $L\ge0$ alone. The second alternative continues to
require $N_R<0\le N_C$.

We next want one fixed formula for $h_2$, so that only the change
in $h_1$ separates the cases in the path argument. Exact machine
checks combine the coefficient bound
$B_{\mathrm{mass}}\ge0$ with the applicable mixing condition to exclude
$u\ge1$ on the full weight range $r\le\rho\le1-r$.
Appendix~\ref{app:weights} gives the corresponding polynomial signs
and their checkers. This also determines the formula for $h_2$.
By~\eqref{eq:main-projected-payoff-bounds}, reaching its cap would
require
\[
 u\ge\frac{r(1+\rho)}{\rho(1-r)}>1.
\]
Hence $h_2$ always uses the second expression in its minimum in
the reduced problem. The only remaining change of formula is that of $h_1$:
it uses the affine expression in~\eqref{eq:main-discovery-h1}
for $v<v_0$, and equals one for $v\ge v_0$.

Any solution of $\Sigma_r$ would therefore give values
$(r,\rho,u,v)$ satisfying the following shorter list. We keep only
the payoff signs needed by the paths below; dropping other
conditions admits more assignments, so excluding this list still
excludes every failed run.
\begin{equation}
 \begin{gathered}
 r_*<r\le1/3,\quad r\le\rho\le1-r,\quad
 k\le u<1,\quad k\le v\le1/k,\quad S>r>T,\\
 T,s_2\ge0,\qquad (\rho\le1/2\ \Longrightarrow\ s_1\ge0),\qquad
 G,B_{\mathrm{mass}},E\ge0,\\
 L\ge0\quad\text{or}\quad
 (N_R<0\le N_C\ \land\ B_{\mathrm{mix}}\ge0).
 \end{gathered}
 \label{eq:main-reduced-conditions}
\end{equation}
The functions are those defined in
\eqref{eq:main-projected-payoff-bounds},
\eqref{eq:main-reduced-comparisons}, and
\eqref{eq:main-mixing-numerators}--\eqref{eq:main-mixing-determinants}.
For $\rho>1/2$, allowing $s_1$ to become negative enlarges the
list of necessary conditions and permits the paths used next.

\paragraph{Moving solutions to where $h_1$ changes formula.}
The reduced conditions now involve $(r,\rho,u,v)$. For each fixed
pair $(r,\rho)$, the remaining freedom is in $u,v$.
The two formulas for $h_1$ suggest two paths. From $v\ge v_0$,
we can try decreasing $v$ directly to the switch. From $v<v_0$,
we first increase $u$ until a mixing condition becomes an equality,
then use that equality to guide a continuation toward $v_0$.
To justify these moves, we must preserve the inequalities
in~\eqref{eq:main-reduced-conditions} and exclude every possible
stop before reaching the destination. The next proposition records
the resulting reduction.

\begin{proposition}[Boundary reduction]
\label{prop:boundary-reduction}
Any point satisfying~\eqref{eq:main-reduced-conditions} produces
a point satisfying the same conditions on $v=v_0$, with $r,\rho$
unchanged.
\end{proposition}
\begin{proof}[Proof idea]
We present only the main proof ideas here; the complete proof of
Proposition~\ref{prop:boundary-reduction} is given in Appendix~\ref{app:paths}.
Suppose first that $v\ge v_0$. Keep $u$ fixed and decrease $v$.
Exact machine checks establish the derivative signs that make
$z=(r-T)/(S-r)$ decrease during this move. Since
$s_1,h_2$ are fixed, the identities
\[
 G=m[(h_2-r)-z(r-s_1)],\qquad
 B_{\mathrm{mass}}=m[(1-r)-rz]
\]
show that $G,B_{\mathrm{mass}}\ge0$ persist while $m>0$.
The same machine check establishes the sign needed at a zero of
the shared comparison:
\[
 E=0\quad\Longrightarrow\quad -\partial_vE>0.
\]
This excludes a crossing from $E\ge0$ to $E<0$ as $v$ decreases.
Further exact machine checks verify that the applicable determinant
increases and that $T<r$. Together with the decrease of $z$, these
signs exclude the possible stops $S=r$ and $T=r$.
The payoff signs are preserved along this move.
Appendix~\ref{app:paths} gives these checks, their domains, and the
remaining stopping arguments. Thus this path reaches $v_0$.

Suppose now that $v<v_0$. Keep $v$ fixed and increase $u$.
Here exact machine checks of the derivatives show that $z$
decreases and $(h_2-r)/(r-s_1)$ increases. These signs preserve $G,B_{\mathrm{mass}}\ge0$ by the same two
identities. The same derivative checks give
\[
 E=0\quad\Longrightarrow\quad \partial_uE>0,
\]
so it also remains nonnegative. The applicable mixing inequality
can become tight during this move. The other stopping events,
including $u=1$, are excluded by the selected comparisons.
The path therefore reaches $L=0$ or $B_{\mathrm{mix}}=0$.

If $\rho\le1/2$, the payoff conditions force $N_C<0$ before the
switch, so the contact is on $L=0$. At the contact used for
continuation, an exact subdivision check verifies $L_u<0$, $L_v<0$,
and the zero-crossing sign in the following display.
The implicit-function theorem~\cite{rudin1976} then gives a
boundary curve with
\begin{equation}
 \frac{du}{dv}=-\frac{L_v}{L_u}<0,\qquad
 E=0\quad\Longrightarrow\quad
 \frac{dE}{dv}=E_v-E_u\frac{L_v}{L_u}>0.
 \label{eq:main-boundary-continuation}
\end{equation}
Increasing $v$ along this curve decreases $u$ and preserves
$E\ge0$. The same subdivision check, combined with the stopping arguments,
allows the curve to continue to $v_0$.
If $\rho\ge1/2$, two separate exact subdivision checks exclude the
contact systems $L=0$ and $B_{\mathrm{mix}}=0$ under
\eqref{eq:main-reduced-conditions}. Consequently, no starting point
with $v<v_0$ and $\rho\ge1/2$ can satisfy these reduced conditions.

Each of these three subdivision checks verifies exact rational
interval or Bernstein bounds~\cite{moore2009,garloff1986} on every
terminal box and checks that the boxes cover its stated domain.
Appendix~\ref{app:paths} gives the path selection, stopping arguments,
and the code paths for each check.
\end{proof}

\paragraph{Ruling out the switch boundary.}
The path reduction leaves $v=v_0$, so $\xi$ is determined by
$r,\rho$ as in~\eqref{eq:main-discovery-cap-switch}.
We return to $\eta=1/(1+u)>1/2$ and work in $(r,\rho,\eta)$.
Here $h_1=1$, and $h_2$ uses precisely the formula
\eqref{eq:main-discovery-h2} used to define the local root.
This is why the switch was chosen: the polynomials from
Section~\ref{sec:bound-discovery} now describe necessary
inequalities. Clearing their positive denominators gives
\begin{equation}
 L\ge0\ \Longleftrightarrow\ N_1(r,\rho,\eta)\ge0,\qquad
 E\ge0\ \Longleftrightarrow\ N_2(r,\rho,\eta)\ge0.
 \label{eq:main-switch-polynomial-conditions}
\end{equation}
For the first mixing alternative, we must exclude simultaneous
$N_1,N_2\ge0$; the second alternative replaces $N_1$ by the
polynomial for its other segment. In each case, we find the first
value of $\eta$ allowed by the mixing condition, show that the
shared condition is already negative there, and prove it cannot
recover as $\eta$ increases. Thus the first branch reuses the local
equations defining $r_*$; the second uses the same shared comparison
with a different mixing test.

\begin{lemma}[Exclusion on the switch boundary]
\label{lem:switch-exclusion}
The necessary conditions~\eqref{eq:main-reduced-conditions}
have no point on $v=v_0$ with $r_*<r\le1/3$.
\end{lemma}
\begin{proof}[Proof idea]
We present only the main proof ideas here; the complete proof of
Lemma~\ref{lem:switch-exclusion} is given in Appendix~\ref{app:switch}.
First suppose $L\ge0$, so $N_1,N_2\ge0$ are necessary.
The polynomial $N_1$ is a strictly concave quadratic in $\eta$.
Exact signs exclude parts of the $(r,\rho)$ domain where this
quadratic is negative throughout $1/2<\eta\le1-r$.
Whenever $N_1\ge0$ is possible, the quadratic has a smaller real
root $\eta_-(r,\rho)>1/2$, and $N_1\ge0$ requires $\eta\ge\eta_-$.
We will prove $N_2<0$ at this mixing boundary and propagate that
sign to any proposed feasible value of $\eta$.

To prove $N_2<0$ at $\eta_-$, we first rule out a zero there.
The resultant of $N_1$ and $N_2$ with respect to $\eta$ is a
polynomial in $r,\rho$ that must vanish if they have a common
$\eta$-root~\cite{cox2015}. Define its relevant factor by
\begin{equation}
 \operatorname{Res}_\eta(N_1,N_2)
   =\sigma^6\rho^2r^2R_{\mathrm{resultant}}(r,\rho).
 \label{eq:main-switch-resultant}
\end{equation}
The prefactor is positive. A zero of $N_2$ at $\eta_-$ would be a
common root of $N_1,N_2$ and would force
$R_{\mathrm{resultant}}=0$. We therefore seek a strict sign for
this factor above $r_*$. Once common zeros are excluded, one
negative reference value and continuity will fix the sign of
$N_2$ along the whole boundary.

The extremum equation~\eqref{eq:main-root-system}, together with
the regularity check at the isolated root, yields a double factor
at $r=r_*$:
\begin{equation}
 R_{\mathrm{resultant}}(r_*,\rho)=(\rho-\rho_*)^2W(\rho),
 \qquad W(\rho)<0\quad(r_*\le\rho\le1-r_*).
 \label{eq:main-resultant-at-root}
\end{equation}
This identity defines the univariate polynomial $W$.
An exact machine check verifies the regularity condition used for
this factorization and certifies $W<0$ using its Bernstein coefficients
and the rational isolating intervals for $r_*,\rho_*$.
The same checker establishes the preliminary
exclusions restricting possible points with $N_1\ge0$ to $r<8/25$
and $\rho>21/50$. On the remaining parameter range, it verifies
$\partial_rR_{\mathrm{resultant}}<0$.
Hence increasing $r$ from $r_*$ gives
$R_{\mathrm{resultant}}(r,\rho)<0$ at every required pair with
$r>r_*$, and excludes a common root.

Because the resultant is nonzero on the relevant region, $N_2$
cannot cross zero along the smaller-root boundary $\eta_-(r,\rho)$.
An exact check gives
$N_2(r_*,1/2,\eta_-(r_*,1/2))<0$.
The connection paths in Appendix~\ref{app:switch} link every
relevant smaller-root point to this reference within the real-root
domain, including cases of coincident quadratic roots. Continuity
therefore gives $N_2<0$ at every such $\eta_-$.

We must still exclude the inequalities at larger $\eta$.
At fixed $(r,\rho,v_0)$, consider the shared residual
$E_1(A_0,1-A_0)=(1-A_0)^2E$, with all payoff functions and $A_0$
evaluated at the varying $\eta$.
Exact machine checks provide the derivative signs that keep
$S>r>T$, $0<A_0\le1-r$, and $h_2>r$ on the interval from
$\eta_-$ to a proposed feasible $\eta$. At a zero of the shared
residual, the chain rule and the derivative of its quadratic
expression then give
\begin{equation}
 E_1(A_0,1-A_0)=0\quad\Longrightarrow\quad
 \frac{d}{d\eta}E_1(A_0,1-A_0)<0.
 \label{eq:main-switch-zero-derivative}
\end{equation}
By~\eqref{eq:main-switch-polynomial-conditions}, its starting value
at $\eta_-$ is negative. The derivative sign prevents it from
crossing to a nonnegative value as $\eta$ increases.
This excludes $N_1,N_2\ge0$ simultaneously.

Finally suppose $N_R<0\le N_C$ and $B_{\mathrm{mix}}\ge0$.
On the switch, this alternative can occur only for $\rho>1/2$;
the payoff condition $T<r$ excludes $N_C\ge0$ on the other half.
Define the polynomial for this segment by substituting
$u=(1-\eta)/\eta$ and $v=v_0$ into
\begin{equation}
 N_B(r,\rho,\eta)=\sigma^2\eta(1-r)^2B_{\mathrm{mix}}.
 \label{eq:main-second-switch-numerator}
\end{equation}
The multiplier is positive, so this alternative requires
$N_B,N_2\ge0$. The polynomial $N_B$ is affine and strictly
increasing in $\eta$ on $1/2\le\rho\le1-r$.
Its unique zero, denoted $\eta_B(r,\rho)$, is above $1/2$, so
$N_B\ge0$ implies $\eta\ge\eta_B$.
Substitution at this zero gives $N_2(r,\rho,\eta_B)<0$ throughout
the stated domain. Here an exact machine check verifies the
strictly negative Bernstein coefficients after substitution.
The interval conditions and the same zero derivative
\eqref{eq:main-switch-zero-derivative} then exclude a nonnegative
shared residual at every larger feasible $\eta$.
Both mixing alternatives are therefore excluded.
Appendix~\ref{app:switch} specifies the sign domains, the
reference connection paths, and the checker paths beside the
resultant and $N_B$ calculations.
\end{proof}

Proposition~\ref{prop:finite-system} maps an output above $r$ to a
solution of $\Sigma_r$. Lemma~\ref{lem:weight-reduction} reduces
such a solution to~\eqref{eq:main-reduced-conditions},
Proposition~\ref{prop:boundary-reduction} moves it to $v=v_0$,
and Lemma~\ref{lem:switch-exclusion} excludes that boundary.
These implications cover every point with $r_*<r\le1/3$ and give
the following conclusion.
\begin{proposition}[Global arithmetic bound]
\label{prop:global-arithmetic}
For every $r$ with $r_*<r\le1/3$, no five values satisfy the
necessary inequalities:
\[
 \nexists\boldsymbol\theta,\quad\Sigma_r(\boldsymbol\theta).
\]
Consequently, the ideal building-block properties imply
$f(\boldsymbol x_{\mathrm{out}})\le r_*$.
\end{proposition}
The last conclusion follows by choosing a comparison threshold
strictly between $r_*$ and a hypothetical larger output regret,
with the threshold below $1/3$.
The exact checks used in these reductions verify specified arithmetic
identities, signs, and coverage; the implication from an algorithm
failure to the numerical inequalities follows from the comparisons in
Section~\ref{sec:reduction}. All proof computations are offline and
are absent from Algorithm~\ref{alg:main}.
The constant is an upper bound for the algorithm. An exact machine
check also verifies that the necessary inequalities can hold with
equality at $r_*$; Appendix~\ref{app:discovery} identifies that
check. This alone does not show the algorithm attains regret $r_*$:
that would require a game realizing those numerical values.
Section~\ref{sec:accuracy-transfer} now restores the approximation
errors.

\subsection{Restoring the approximation errors}
\label{sec:accuracy-transfer}

The preceding subsection proved that $\Sigma_r$ has no solution for
$r_*<r\le1/3$ and hence established the ideal guarantee.
The actual algorithm uses accuracy $\delta_f$ for $\spoint$ and
$\tau$ for $\om$, as specified in Section~\ref{sec:complete-algorithm}.
To prove the actual guarantee, suppose the output exceeds
$r+\delta_f+\tau$. The $\om$ guarantee then makes every tested
mixture exceed $r+\delta_f$. We next account for the stationary
comparison and its nearly best-response strategies using a single
$\delta_f$ margin. This recovers the same necessary inequalities
$\Sigma_r$, contradicting the preceding subsection. The four
$\sm$ LPs are solved exactly and add no error. The next lemma
states this transfer.

\begin{lemma}[Transfer to approximate building blocks]
\label{lem:error-transfer}
For the actual accuracy parameters $\delta_f,\tau>0$, an execution
with the building-block properties of Section~\ref{sec:properties}
satisfies
\[
 \forall r\in(r_*,1/3),\quad
 f(\boldsymbol x_{\mathrm{out}})>r+\delta_f+\tau
 \ \Longrightarrow\ \exists\boldsymbol\theta,\ \Sigma_r(\boldsymbol\theta).
\]
\end{lemma}
\begin{proof}[Proof idea]
We present only the main proof ideas here; the complete proof of
Lemma~\ref{lem:error-transfer} is given in Appendix~\ref{app:errors}.
Fix $r\in(r_*,1/3)$ and assume the stated failed-output threshold.
Output selection~\eqref{eq:properties-selection} makes every
candidate exceed $r+\delta_f+\tau$.
The comparison of $\om$ in~\eqref{eq:om} gives
$f(\boldsymbol x_o)\le f(\boldsymbol x(p,q))+\tau$ for every
$p,q\in[0,1]$. Therefore
\begin{equation}
 \forall p,q\in[0,1],\quad
 f(\boldsymbol x(p,q))>r+\delta_f.
 \label{eq:main-error-rectangle-failure}
\end{equation}
This is the only subtraction of $\tau$. All four $\sm$ LPs are
exact, so their comparison inequalities and complementary slackness
introduce no additional accuracy term.

The stationary comparison guarantees only $g-\delta_f$ where the
ideal proof used $g$. To reuse the earlier inequalities, call this
conservative value $G_s$:
\begin{equation}
 G_s=g-\delta_f>r.
 \label{eq:main-error-stationary-value}
\end{equation}
We use this same value in every subsequent comparison. The
comparisons at $(x_s^1,y_s^2)$ and
$(y_s^1,x_s^2)$ give $G_s\le\rho,\sigma$, so both stationary
coefficients are above $r$.
The near-best-response error also fits within the same surplus.
For the row player, let $\tau_1$ measure how far $y_s^1$ is from
being a best response to $x_s^2$; it is at most $\delta_f$ and is
unrelated to the $\om$ accuracy $\tau$. Also let $d$ be the row
regret after replacing $y_s^1$ by $y_1^1$:
\begin{equation}
 \tau_1=f_1(y_s^1,x_s^2)=B_1(x_s^2)-u_1(y_s^1,x_s^2),\quad
 0\le\tau_1\le\delta_f,\qquad d=f_1(y_1^1,x_s^2)\ge0.
 \label{eq:main-error-row-shortfall}
\end{equation}
Applying the comparison inequality~\eqref{eq:properties-dual} of
$\sm(y_s^1)$ at a column best response to $y_s^1$ gives
$F_1^1\le\eta$, hence $\eta>r$.
Applying it at $x_s^2$ instead gives
\begin{equation}
 \begin{aligned}
 r+\delta_f<F_1^1&\le\eta(\tau_1-d)+(1-\eta),\\
 \eta d&<1-\eta-r-\delta_f+\eta\tau_1\le1-\eta-r.
 \end{aligned}
 \label{eq:main-error-cancellation}
\end{equation}
Thus $\eta<1-r$ and
$d\le\min\{1,(1-\eta-r)/\eta\}$, the same regret bound used
in~\eqref{eq:main-first-round-regret-bounds}.
The error has canceled before taking the minimum with one.
An exact machine check reconstructs this cancellation identity
and both cases of the minimum.
The symmetric comparisons give $r<\xi<1-r$ and the corresponding
column-regret bound. The auxiliary strategies from the first round
are then exact best responses to those calls' chosen strategies;
the second-round comparisons give $r<\alpha,\gamma<1-r$ as in the
ideal proof.

Using the stationary value $G_s$ and these regret bounds, the six
payoff bounds are now evaluated at
\begin{equation}
 \boldsymbol\lambda_s=(r,G_s,\rho,\eta,\xi).
 \label{eq:main-error-payoff-parameters}
\end{equation}
The same applications of the comparison inequalities from
Section~\ref{sec:reduction}
give their nonnegativity and the weighted and shared comparisons
\eqref{eq:main-weighted-regret-bounds} and
\eqref{eq:main-shared-coupling}, with $g=G_s$.
The joint regret comparisons used for the Jensen-gap bounds have
additional nonnegative terms proportional to $\tau_1$ and its
column counterpart. Dropping these terms and applying the weighted
minimum argument of Lemma~\ref{lem:main-jensen} gives the row
lower bound $j_r(\boldsymbol\lambda_s,q)$ and the symmetric
column bound, both evaluated at $g=G_s$.

It remains to transfer the segment conditions. In the row-regret
interpolation~\eqref{eq:mixing-regrets}, the stationary corner
contributes $g=G_s+\delta_f$ and the corner $(y_s^1,x_s^2)$
contributes $\tau_1$. Their total error is bounded by
\[
 \delta_f(1-p)(1-q)+\tau_1p(1-q)
 \le\delta_f(1-q).
\]
The column calculation is symmetric. With the Jensen-gap lower
bounds just obtained, this gives
\begin{equation}
 \begin{aligned}
 \forall p,q\in[0,1],\quad f_1(\boldsymbol x(p,q))
   &\le R_{G_s}(p,q)+\delta_f(1-q),\\
 \forall p,q\in[0,1],\quad f_2(\boldsymbol x(p,q))
   &\le C_{G_s}(p,q)+\delta_f(1-p).
 \end{aligned}
 \label{eq:main-error-rectangle-bounds}
\end{equation}
Here $R_{G_s},C_{G_s}$ are $R_g,C_g$ with $g$ replaced by $G_s$.
The same machine check reconstructs these error identities and
the nonnegative combinations giving the payoff and regret bounds.
If their maximum were at most $r$ at any $(p,q)$, both actual
regrets would be at most $r+\delta_f$, contradicting
\eqref{eq:main-error-rectangle-failure}.
The four segment conditions~\eqref{eq:main-segment-conditions}
therefore hold at this same adjusted parameter.

The first-round comparisons~\eqref{eq:main-first-round-g-comparisons},
with $g=G_s$, give the interval needed for the substitution in
Lemma~\ref{lem:main-g-projection}. Its preservation argument uses
only the coefficient bounds, nonnegativity, weighted and shared
comparisons, segment conditions, and first-round comparisons
established above. It therefore applies to these adjusted numerical
values. Substituting $G_s=r$ and then using the payoff bounds
\eqref{eq:main-projected-payoff-bounds} produces $\Sigma_r$.
Appendix~\ref{app:errors} identifies the checker path for these
error calculations.
\end{proof}

For every $r\in(r_*,1/3)$, Proposition~\ref{prop:global-arithmetic}
excludes the solution supplied by Lemma~\ref{lem:error-transfer}.
If the output regret exceeded $r_*+\delta_f+\tau$, we could choose
\[
 r_*<r<\min\{f(\boldsymbol x_{\mathrm{out}})-\delta_f-\tau,1/3\},
\]
contradicting that exclusion. Hence
\begin{equation}
 f(\boldsymbol x_{\mathrm{out}})\le r_*+\delta_f+\tau.
 \label{eq:main-error-final-bound}
\end{equation}
Taking $\delta_f=\tau=\delta/2$ gives the claimed
$(r_*+\delta)$ approximation and proves Theorem~\ref{thm:main};
the running time follows from the fixed composition in
Section~\ref{sec:complete-algorithm}.

\section{Discussion}
\label{sec:discussion}

The present work extends the role of LLM agents in the LegoNE
methodology. In the reported LegoNE discovery runs, the programming
language, building-block specifications, and compilation rules were
fixed before the search~\cite{legone2026,computeraided2026}; the LLM
agent explored compositions within that space. Here, an LLM agent
can also develop the search space: identify or propose building
blocks, interleave instantiation and forgetting as the argument
evolves, and choose proof methods and write suitable compilation
or verification routines. The mathematical abstractions and
verification methods can therefore evolve together with the algorithm.

Codex played a central role in discovering the algorithm and developing
its proof, including the specialized routines and exact arithmetic
checkers. In our experience, LLM agents appear indispensable for the
present improvement and may remain so for future advances in
approximate Nash equilibrium algorithms: the proof complexity already
exceeds what we can reliably handle through unaided human reasoning.

An open problem is a polynomial-time $(c+\delta)$-approximation
algorithm for bimatrix games with a constant $c<0.3$ and every
$\delta>0$. We explored this target with Codex, but prolonged work
did not produce a certified bound: the growing complexity of the
proof became the main obstacle. Progress may require new building
blocks or mathematical abstractions that make the analysis simpler.

\section*{Acknowledgments}

OpenAI's Codex was used for manuscript preparation, including language
editing and \LaTeX{} typesetting. It also made substantial contributions
to algorithm design, the development of the approximation proof, and
the implementation of exact arithmetic verification routines. The
authors are responsible for reviewing and finalizing all mathematical
arguments and computational results, and assume full responsibility
for the content.

% A self-contained bibliography keeps this initial draft compilable in the
% built-in LaTeX editor without an external BibTeX file.
% Cite journal versions first, then conference versions, then preprints.

\appendix

\section{Instantiating the building-block properties}
\label{app:instantiation}

This appendix supplies the strategy instantiations used in
Lemma~\ref{lem:positive-weights} and in the payoff elimination of
Section~\ref{sec:reduction}. We analyze the building-block properties
with zero approximation errors, as specified in~\eqref{eq:ideal-parameters}.
The accuracy transfer is proved in Appendix~\ref{app:errors}.
Fix one game with payoffs in $[0,1]$, one execution, and one optimal
dual solution for each of its five LPs, including the stationary LP.
All comparisons below use this same game and these same dual solutions.
The stationary-point and mixing encodings are those
of~\cite{ts2008,tightness2023,automating2025,mixing2025};
their instantiation follows the LegoNE
framework~\cite{computeraided2026,legone2026} and uses LP optimality
and complementary slackness~\cite{boyd2004}.

Use the four candidate profiles $\boldsymbol z_j^i$ and their maximum
regrets $F_j^i$ from~\eqref{eq:candidate-profiles}. In these two symbols,
the superscript identifies the player whose input is fixed and the
subscript identifies the round. In a single strategy $x_j^i$ or $y_j^i$,
the superscript instead identifies the player using that strategy.
The four weights $\eta,\xi,\alpha,\gamma$ are the call-specific
dual weights defined in~\eqref{eq:call-weight-definitions}. For example,
$\eta$ weights the row-deviation gains in the first-round operation
$\sm(y_s^1)$, and $1-\eta$ weights its column regret.

To analyze a possible failure of the approximation guarantee, assume that
the selected output has maximum regret above $r$, with $0<r\le1/3$,
as in~\eqref{eq:failure-hypothesis}. The resulting inequalities
\eqref{eq:failure-candidates} say that all four $F_j^i$ exceed $r$,
that $g\ge f(\boldsymbol x_o)>r$, and that every profile
$\boldsymbol x(p,q)$ in the original mixing rectangle has maximum
regret above $r$. These assumptions allow comparisons at strategies
chosen during the proof, as well as at the returned candidates.

\paragraph{The stationary LP comparison.}
The joint inequality~\eqref{eq:properties-stationary} can be
represented by two additive constants from the stationary LP.
Write $A_s,B_s$ for lower bounds on its two separate comparison
parts. These constants are chosen for the fixed stationary dual
solution, independently of the comparison strategies. With zero
error, their conditions are
\begin{equation}
 \begin{gathered}
 \exists A_s,B_s\in\mathbb R\text{ such that}\\
 \begin{aligned}
 \forall x^{1\prime}\in\Delta_{n_1},\quad
 A_s&\le-\rho u_1(x^{1\prime},x_s^2)
       +\sigma[u_2(x^{1\prime},y_s^2)-u_2(x^{1\prime},x_s^2)],\\
 \forall x^{2\prime}\in\Delta_{n_2},\quad
 B_s&\le-\sigma u_2(x_s^1,x^{2\prime})
       +\rho[u_1(y_s^1,x^{2\prime})-u_1(x_s^1,x^{2\prime})],\\
 g&\le\rho u_1(\boldsymbol x_s)+\sigma u_2(\boldsymbol x_s)+A_s+B_s.
 \end{aligned}
 \end{gathered}
 \label{eq:stationary-lp-offsets}
\end{equation}
Such constants exist: minimize each comparison part over its own
simplex. The sum of these two minima is the minimum of their sum,
since the two comparison strategies are independent; the universal
stationarity inequality gives the last line. Conversely, adding the
three lines gives~\eqref{eq:properties-stationary}.
The quantified registry records this LP representation. We use the
combined inequality directly, eliminating $A_s,B_s$ before forming
the scalar necessary system. With approximation error, the last
line has $g-\delta_f$ in place of $g$.

\paragraph{The four LP comparisons.}
We first record the dual inequalities for the first-round operations
$\sm(y_s^1)$ and $\sm(y_s^2)$ and the second-round operations
$\sm(y_1^1)$ and $\sm(y_1^2)$, so that each subsequent strategy
instantiation has an explicit premise. Instantiating
\eqref{eq:properties-dual} with their inputs and outputs in
Algorithm~\ref{alg:main} and the weight assignments
\eqref{eq:call-weight-definitions} gives
\begin{subequations}
\label{eq:app-four-call-clauses}
\begin{align}
 \forall x^{2\prime}\in\Delta_{n_2},\quad
 F_1^1&\le
 \eta[u_1(y_1^1,x^{2\prime})-u_1(y_s^1,x^{2\prime})]
 +(1-\eta)f_2(y_s^1,x^{2\prime}),
 \label{eq:first-row-instance}\\
 \forall x^{1\prime}\in\Delta_{n_1},\quad
 F_1^2&\le
 \xi[u_2(x^{1\prime},y_1^2)-u_2(x^{1\prime},y_s^2)]
 +(1-\xi)f_1(x^{1\prime},y_s^2),
 \label{eq:app-first-column-instance}\\
 \forall x^{2\prime}\in\Delta_{n_2},\quad
 F_2^1&\le
 \alpha[u_1(y_2^1,x^{2\prime})-u_1(y_1^1,x^{2\prime})]
 +(1-\alpha)f_2(y_1^1,x^{2\prime}),
 \label{eq:app-second-row-instance}\\
 \forall x^{1\prime}\in\Delta_{n_1},\quad
 F_2^2&\le
 \gamma[u_2(x^{1\prime},y_2^2)-u_2(x^{1\prime},y_1^2)]
 +(1-\gamma)f_1(x^{1\prime},y_1^2).
 \label{eq:app-second-column-instance}
\end{align}
\end{subequations}
Each inequality holds throughout the stated comparison simplex,
with one fixed dual strategy and weight for its specified operation.

\paragraph{Positive weights and best responses.}
The best-response conclusions
\eqref{eq:properties-stationary-br} and~\eqref{eq:properties-sm-br}
require positive dual weights. We establish that positivity before
using the conclusions to compare payoffs for the second-round operations
$\sm(y_1^1)$ and $\sm(y_1^2)$.
Recall that $\rho$ is the stationary row weight and $\sigma=1-\rho$
is the stationary column weight.
Define the two stationary/dual payoff differences by
\begin{equation}
 \lambda_{\mathrm{payoff}}=u_1(y_s^1,y_s^2)-u_1(x_s^1,y_s^2),\qquad
 \mu_{\mathrm{payoff}}=u_2(y_s^1,y_s^2)-u_2(y_s^1,x_s^2).
 \label{eq:stationary-payoff-differences}
\end{equation}
Instantiate~\eqref{eq:properties-stationary}, with $\delta_f=0$,
at the comparison profiles $(x_s^1,y_s^2)$ and $(y_s^1,x_s^2)$,
respectively:
\begin{equation}
 \begin{aligned}
 g&\le\rho[u_1(y_s^1,y_s^2)-u_1(x_s^1,y_s^2)],\\
 g&\le\sigma[u_2(y_s^1,y_s^2)-u_2(y_s^1,x_s^2)].
 \end{aligned}
 \label{eq:app-stationary-weight-tests}
\end{equation}
The column bracket is zero in the first instantiation and the row
bracket is zero in the second. Each displayed payoff difference is
at most one. Since $g>r$ in~\eqref{eq:failure-candidates}, this gives
\begin{equation}
 r<\rho<1-r,\qquad \rho,\sigma>0.
 \label{eq:app-positive-stationary-weights}
\end{equation}
Thus~\eqref{eq:properties-stationary-br}, with zero error, implies
\begin{equation}
 B_1(x_s^2)=u_1(y_s^1,x_s^2),\qquad
 B_2(x_s^1)=u_2(x_s^1,y_s^2).
 \label{eq:app-stationary-best-responses}
\end{equation}
The functions $B_1,B_2$ are the best-response payoff functions
defined in Section~\ref{sec:optimalmixing}.
Nonnegativity of $u_1(x_s^1,y_s^2)$ and $u_2(y_s^1,x_s^2)$ in
\eqref{eq:app-stationary-weight-tests} also gives
\begin{equation}
 u_1(y_s^1,y_s^2)\ge g/\rho,\qquad
 u_2(y_s^1,y_s^2)\ge g/\sigma.
 \label{eq:app-stationary-payoff-lower-bounds}
\end{equation}
Equivalently, the two payoff differences satisfy the lower bounds
\begin{equation}
 \lambda_{\mathrm{payoff}}\ge\lambda_0=g/\rho,\qquad
 \mu_{\mathrm{payoff}}\ge\mu_0=g/\sigma.
 \label{eq:stationary-gap-lower-values}
\end{equation}
These lower bounds will bound regrets from above, since all
best-response payoffs are at most one.

For the first-round operation $\sm(y_s^1)$, choose a column best
response $x_{\mathrm{br}}^2$ to $y_s^1$. It exists because payoffs are
continuous and the simplex is compact. Instantiate
\eqref{eq:first-row-instance} with $x^{2\prime}=x_{\mathrm{br}}^2$.
The column regret is zero, so
\begin{equation}
 \begin{aligned}
 f_2(y_s^1,x_{\mathrm{br}}^2)&=0,\\
 F_1^1&\le\eta[u_1(y_1^1,x_{\mathrm{br}}^2)
                   -u_1(y_s^1,x_{\mathrm{br}}^2)]\le\eta.
 \end{aligned}
 \label{eq:app-first-row-best-response-test}
\end{equation}
Next instantiate the same inequality with $x^{2\prime}=x_s^2$.
By~\eqref{eq:app-stationary-best-responses}, its row payoff
difference is nonpositive; the column regret is at most one. Hence
\begin{equation}
 F_1^1\le\eta[u_1(y_1^1,x_s^2)-u_1(y_s^1,x_s^2)]
                  +(1-\eta)f_2(y_s^1,x_s^2)\le1-\eta.
 \label{eq:app-first-row-stationary-test}
\end{equation}
Together with $F_1^1>r$ in~\eqref{eq:failure-candidates}, these
instantiations give
$r<\eta<1-r$.

To transfer a best-response conclusion to the second round,
instantiate the strong-duality equality
\eqref{eq:properties-strong-duality} with the input $y_s^1$,
primal output $x_1^2$, dual strategy $y_1^1$, and weight $\eta$
of the first-round operation $\sm(y_s^1)$.
The regret definition~\eqref{eq:regret} then gives
\begin{equation}
 \begin{aligned}
 F_1^1
 &=\eta[u_1(y_1^1,x_1^2)-u_1(y_s^1,x_1^2)]
                       +(1-\eta)f_2(y_s^1,x_1^2)\\
 &\le\eta f_1(y_s^1,x_1^2)+(1-\eta)f_2(y_s^1,x_1^2)
 \le F_1^1.
 \end{aligned}
 \label{eq:app-first-row-tightness}
\end{equation}
Every inequality is an equality. Since $\eta>0$, $y_1^1$ attains
the row best-response payoff against $x_1^2$. Since $\eta$ and
$1-\eta$ are both positive, both regrets equal the maximum regret:
\begin{equation}
 B_1(x_1^2)=u_1(y_1^1,x_1^2),\qquad
 f_1(\boldsymbol z_1^1)=f_2(\boldsymbol z_1^1)=F_1^1.
 \label{eq:app-first-row-best-response-output}
\end{equation}

For the first-round operation $\sm(y_s^2)$, instantiate
\eqref{eq:app-first-column-instance} at a row best response to
$y_s^2$; the row regret vanishes and $F_1^2\le\xi$.
Instantiate the same inequality with $x^{1\prime}=x_s^1$;
the column payoff difference is nonpositive by
\eqref{eq:app-stationary-best-responses}, so $F_1^2\le1-\xi$.
Apply the strong-duality equality and regret definition at the
primal output $x_1^1$, now with input $y_s^2$, dual strategy
$y_1^2$, and weight $\xi$ of $\sm(y_s^2)$. This gives
\begin{equation}
 r<\xi<1-r,\qquad
 B_2(x_1^1)=u_2(x_1^1,y_1^2),\qquad
 f_1(\boldsymbol z_1^2)=f_2(\boldsymbol z_1^2)=F_1^2.
 \label{eq:app-first-column-best-response-output}
\end{equation}

For the second-round operation $\sm(y_1^1)$, instantiate
\eqref{eq:app-second-row-instance} with $x^{2\prime}=x_1^2$.
Its input $y_1^1$ is a best response to $x_1^2$ by
\eqref{eq:app-first-row-best-response-output}, giving
$F_2^1\le1-\alpha$. Instantiating the same inequality at a column
best response to $y_1^1$ instead gives $F_2^1\le\alpha$.
For the second-round operation $\sm(y_1^2)$, instantiate
\eqref{eq:app-second-column-instance} first with $x^{1\prime}=x_1^1$
and then at a row best response to $y_1^2$. These give
$F_2^2\le1-\gamma$ and $F_2^2\le\gamma$, respectively.
The failures in~\eqref{eq:failure-candidates} therefore make all
four weights and their complements positive. Instantiating
\eqref{eq:properties-strong-duality} at the respective primal outputs
of $\sm(y_1^1)$ and $\sm(y_1^2)$ and repeating the equality chain
\eqref{eq:app-first-row-tightness} gives
\begin{equation}
 \begin{gathered}
 r<\eta,\xi,\alpha,\gamma<1-r,\\
 B_1(x_2^2)=u_1(y_2^1,x_2^2),\qquad
 B_2(x_2^1)=u_2(x_2^1,y_2^2),\\
 f_1(\boldsymbol z_2^1)=f_2(\boldsymbol z_2^1)=F_2^1,\qquad
 f_1(\boldsymbol z_2^2)=f_2(\boldsymbol z_2^2)=F_2^2.
 \end{gathered}
 \label{eq:app-positive-call-weights}
\end{equation}
Thus the first-round best-response conclusion is proved before
it is used to instantiate the dual inequalities of the second-round
operations $\sm(y_1^1)$ and $\sm(y_1^2)$.

\paragraph{Lower bounds on payoffs used by the second round.}
Recall the parameter tuple $\boldsymbol\lambda=(r,g,\rho,\eta,\xi)$
from~\eqref{eq:payoff-bound-parameters}. All six payoff-bound
functions below are evaluated at this full tuple.
Dividing the dual inequalities of the first-round operations
$\sm(y_s^1)$ and $\sm(y_s^2)$ by their positive weights repeatedly
produces the same ratios. Recall the relative weights $u,v$ from
\eqref{eq:main-relative-weights}. The following display uses exactly
these same $u,v$ and records their common lower endpoint:
\begin{equation}
 u=\frac{1-\eta}{\eta},\qquad
 v=\frac{1-\xi}{\xi},\qquad
 k=\frac r{1-r},\qquad k<u,v<1/k.
 \label{eq:relative-weights}
\end{equation}
The inequalities follow from~\eqref{eq:app-positive-call-weights}.
These are scalar ratios; the payoff functions continue to be
written as $u_1,u_2$.

We seek lower bounds on the payoffs of the first-round dual
strategies, because they are the inputs of $\sm(y_1^1)$ and
$\sm(y_1^2)$. Instantiate the dual inequality
\eqref{eq:first-row-instance} of $\sm(y_s^1)$ with
$x^{2\prime}=y_s^2$, and use
\eqref{eq:app-stationary-payoff-lower-bounds}:
\begin{equation}
 \begin{aligned}
 r<F_1^1
 &\le\eta[u_1(y_1^1,y_s^2)-u_1(y_s^1,y_s^2)]
                          +(1-\eta)f_2(y_s^1,y_s^2)\\
 &\le\eta[u_1(y_1^1,y_s^2)-g/\rho]
                          +(1-\eta)(1-g/\sigma).
 \end{aligned}
 \label{eq:app-first-row-dual-profile-test}
\end{equation}
Solving for $u_1(y_1^1,y_s^2)$, and then subtracting from one,
gives the first of the following upper-bound functions.
The second follows by instantiating the dual inequality
\eqref{eq:app-first-column-instance} of $\sm(y_s^2)$ with
$x^{1\prime}=y_s^1$:
\begin{equation}
 \begin{aligned}
 s_1(\boldsymbol\lambda)&=1-g/\rho-\frac{r-(1-\eta)(1-g/\sigma)}{\eta},\\
 s_2(\boldsymbol\lambda)&=1-g/\sigma-\frac{r-(1-\xi)(1-g/\rho)}{\xi}.
 \end{aligned}
 \label{eq:payoff-bound-s}
\end{equation}
They satisfy
\begin{equation}
 \begin{aligned}
 1-u_1(y_1^1,y_s^2)&\le s_1(\boldsymbol\lambda),&
 f_1(y_1^1,y_s^2)&\le s_1(\boldsymbol\lambda),\\
 1-u_2(y_s^1,y_1^2)&\le s_2(\boldsymbol\lambda),&
 f_2(y_s^1,y_1^2)&\le s_2(\boldsymbol\lambda).
 \end{aligned}
 \label{eq:app-payoff-bound-s-meaning}
\end{equation}
The bounds on $1-u_i$ give payoff information as well as regret
information; later inequalities will use both.

We also need the row regret of $y_1^1$ against $x_s^2$. Rearrange
\eqref{eq:app-first-row-stationary-test}, using $F_1^1>r$,
$f_2\le1$, and~\eqref{eq:app-stationary-best-responses}.
This regret is at most $(1-\eta-r)/\eta$ and also at most one.
Instantiating the dual inequality
\eqref{eq:app-first-column-instance} of $\sm(y_s^2)$ with
$x^{1\prime}=x_s^1$ gives the second bound:
\begin{equation}
 \begin{aligned}
 d_1&=\min\{1,(1-\eta-r)/\eta\},&
 B_1(x_s^2)-u_1(y_1^1,x_s^2)&\le d_1,\\
 d_2&=\min\{1,(1-\xi-r)/\xi\},&
 B_2(x_s^1)-u_2(x_s^1,y_1^2)&\le d_2.
 \end{aligned}
 \label{eq:app-stationary-payoff-loss-bounds}
\end{equation}
Both $d_i$ are nonnegative by~\eqref{eq:app-positive-call-weights}.

To turn these regrets into payoff bounds for the other player,
instantiate~\eqref{eq:properties-stationary} at
$(x_s^1,y_1^2)$ and $(y_1^1,x_s^2)$. The zero-error comparisons are
\begin{equation}
 \begin{aligned}
 g&\le\rho[u_1(y_s^1,y_1^2)-u_1(x_s^1,y_1^2)]
              +\sigma[B_2(x_s^1)-u_2(x_s^1,y_1^2)],\\
 g&\le\rho[B_1(x_s^2)-u_1(y_1^1,x_s^2)]
              +\sigma[u_2(y_1^1,y_s^2)-u_2(y_1^1,x_s^2)].
 \end{aligned}
 \label{eq:app-cross-stationary-tests}
\end{equation}
Use~\eqref{eq:app-stationary-payoff-loss-bounds} and nonnegative
payoffs to obtain $u_1(y_s^1,y_1^2)\ge g/\rho-(\sigma/\rho)d_2$
and $u_2(y_1^1,y_s^2)\ge g/\sigma-(\rho/\sigma)d_1$.
Together with $u_i\ge0$, these define
\begin{equation}
 \begin{aligned}
 h_1(\boldsymbol\lambda)&=\min\{1,1-g/\rho+(\sigma/\rho)d_2\},\\
 h_2(\boldsymbol\lambda)&=\min\{1,1-g/\sigma+(\rho/\sigma)d_1\}.
 \end{aligned}
 \label{eq:payoff-bound-h}
\end{equation}
Their meanings are
\begin{equation}
 1-u_1(y_s^1,y_1^2)\le h_1(\boldsymbol\lambda),\qquad
 1-u_2(y_1^1,y_s^2)\le h_2(\boldsymbol\lambda).
 \label{eq:app-payoff-bound-h-meaning}
\end{equation}
Thus $(s_1(\boldsymbol\lambda),h_2(\boldsymbol\lambda))$ bounds the two regrets at $(y_1^1,y_s^2)$,
and $(h_1(\boldsymbol\lambda),s_2(\boldsymbol\lambda))$ bounds them at $(y_s^1,y_1^2)$.

We next compare the dual strategies returned by
$\sm(y_s^1)$ and $\sm(y_s^2)$.
Instantiate the dual inequality~\eqref{eq:first-row-instance}
of $\sm(y_s^1)$ with $x^{2\prime}=y_1^2$ and the dual inequality
\eqref{eq:app-first-column-instance} of $\sm(y_s^2)$ with
$x^{1\prime}=y_1^1$.
Using~\eqref{eq:failure-candidates}, then the four payoff bounds
\eqref{eq:app-payoff-bound-s-meaning} and
\eqref{eq:app-payoff-bound-h-meaning}, gives
\begin{equation}
 \begin{aligned}
 r&\le\eta[u_1(y_1^1,y_1^2)-u_1(y_s^1,y_1^2)]
                           +(1-\eta)f_2(y_s^1,y_1^2),\\
 r&\le\xi[u_2(y_1^1,y_1^2)-u_2(y_1^1,y_s^2)]
                           +(1-\xi)f_1(y_1^1,y_s^2).
 \end{aligned}
 \label{eq:app-first-round-cross-tests}
\end{equation}
Solving for the actual payoffs at $(y_1^1,y_1^2)$ defines
\begin{equation}
 S(\boldsymbol\lambda)=h_1(\boldsymbol\lambda)-r+u(s_2(\boldsymbol\lambda)-r),\qquad
 T(\boldsymbol\lambda)=h_2(\boldsymbol\lambda)-r+v(s_1(\boldsymbol\lambda)-r).
 \label{eq:payoff-bound-st}
\end{equation}
In particular,
\begin{equation}
 1-u_1(y_1^1,y_1^2)\le S(\boldsymbol\lambda),\qquad
 1-u_2(y_1^1,y_1^2)\le T(\boldsymbol\lambda).
 \label{eq:app-payoff-bound-st-meaning}
\end{equation}
Here the one-argument notation $T(\boldsymbol\lambda)$ always denotes the scalar
upper-bound function in~\eqref{eq:payoff-bound-st}.
The four-argument expression used in~\eqref{eq:properties-stationary}
continues to denote the stationary payoff comparison.
The six functions $s_i(\boldsymbol\lambda),h_i(\boldsymbol\lambda),S(\boldsymbol\lambda),T(\boldsymbol\lambda)$ are nonnegative at the
scalar values supplied by the execution, since each bounds
a nonnegative quantity of the form $1-u_i$.

\paragraph{Second-round comparisons at four specified strategies.}
The preceding six payoff bounds allow four instantiations of the
dual inequalities of the second-round operations $\sm(y_1^1)$ and
$\sm(y_1^2)$ to be written without additional payoff variables. Define
\begin{equation}
 \begin{aligned}
 U_g&=\alpha s_1(\boldsymbol\lambda)+(1-\alpha)h_2(\boldsymbol\lambda),&
 V_g&=(1-\gamma)h_1(\boldsymbol\lambda)+\gamma s_2(\boldsymbol\lambda),\\
 H_g&=\alpha S(\boldsymbol\lambda)+(1-\alpha)T(\boldsymbol\lambda),&
 K_g&=(1-\gamma)S(\boldsymbol\lambda)+\gamma T(\boldsymbol\lambda).
 \end{aligned}
 \label{eq:app-weighted-regret-bounds}
\end{equation}
Instantiate the dual inequality~\eqref{eq:app-second-row-instance}
of $\sm(y_1^1)$ with $x^{2\prime}=y_s^2$ and $y_1^2$; instantiate
the dual inequality~\eqref{eq:app-second-column-instance} of
$\sm(y_1^2)$ with $x^{1\prime}=y_s^1$ and $y_1^1$.
Each deviation payoff is
at most one, and each regret is at most one minus the actual
payoff. The four bounds therefore imply
\begin{equation}
 U_g\ge r,\qquad V_g\ge r,\qquad H_g\ge r,\qquad K_g\ge r.
 \label{eq:app-second-round-comparisons}
\end{equation}
For example, instantiating the dual inequality of $\sm(y_1^1)$
with $x^{2\prime}=y_s^2$ gives an expression bounded above by
$\alpha[1-u_1(y_1^1,y_s^2)]+(1-\alpha)[1-u_2(y_1^1,y_s^2)]$,
which is at most $U_g$. The lower bound $r$ comes from
$F_2^1>r$ in~\eqref{eq:failure-candidates}.

\paragraph{Comparisons combining the two rounds.}
Instantiating two dual inequalities at the same arbitrary strategy
can give a restriction that separate instantiations do not supply.
We use the first-round operation $\sm(y_s^2)$ and the second-round
operation $\sm(y_1^2)$ to bound a row best-response payoff at a
mixture, and then instantiate the dual inequality of the second-round
operation $\sm(y_1^1)$ at that mixture.
Fix $x^{1\prime}\in\Delta_{n_1}$. Equations
\eqref{eq:app-first-column-instance} and
\eqref{eq:app-second-column-instance}, the failures
\eqref{eq:failure-candidates}, and the regret definition give
\begin{equation}
 \begin{aligned}
 r&\le\xi[u_2(x^{1\prime},y_1^2)-u_2(x^{1\prime},y_s^2)]
             +(1-\xi)[B_1(y_s^2)-u_1(x^{1\prime},y_s^2)],\\
 r&\le\gamma[u_2(x^{1\prime},y_2^2)-u_2(x^{1\prime},y_1^2)]
             +(1-\gamma)[B_1(y_1^2)-u_1(x^{1\prime},y_1^2)].
 \end{aligned}
 \label{eq:app-shared-column-tests}
\end{equation}
Use $u_2(x^{1\prime},y_s^2)\ge0$, $B_1\le1$, and
$u_2(x^{1\prime},y_2^2)\le1$, then rearrange:
\begin{equation}
 \begin{aligned}
 \xi u_2(x^{1\prime},y_1^2)
     &\ge r-(1-\xi)+(1-\xi)u_1(x^{1\prime},y_s^2),\\
 \gamma u_2(x^{1\prime},y_1^2)
     &\le1-r-(1-\gamma)u_1(x^{1\prime},y_1^2).
 \end{aligned}
 \label{eq:app-shared-column-payoff-inequalities}
\end{equation}
Multiply the first inequality by the positive number $\gamma/\xi$
and substitute its lower bound for the common column payoff into
the second inequality. With $v$ from~\eqref{eq:relative-weights},
this gives
\begin{equation}
 \forall x^{1\prime}\in\Delta_{n_1},\quad
 (1-\gamma)u_1(x^{1\prime},y_1^2)+\gamma v u_1(x^{1\prime},y_s^2)
 \le1-r-\frac{\gamma}{\xi}[r-(1-\xi)].
 \label{eq:common-payoff}
\end{equation}
Both occurrences of the eliminated payoff concern the same
strategy and the same fixed dual solutions. This is why the
universal comparison holds with the same $\xi,\gamma$.

Recall $D_E$ and $y_E^2$ from~\eqref{eq:main-column-comparison-mixture}.
They normalize the coefficients on the left of~\eqref{eq:common-payoff}:
\begin{equation}
 D_E=1-\gamma+\gamma v>0,\qquad
 y_E^2=\frac{(1-\gamma)y_1^2+\gamma v y_s^2}{D_E}\in\Delta_{n_2}.
 \label{eq:app-column-comparison-mixture}
\end{equation}
This is a strategy chosen for the proof. It gives a valid instantiation
of the dual inequality~\eqref{eq:app-second-row-instance} of
$\sm(y_1^1)$, whose comparison domain is the entire column simplex. Payoff linearity and~\eqref{eq:common-payoff}
give
\begin{equation}
 B_1(y_E^2)\le
 \frac{1-r-(\gamma/\xi)[r-(1-\xi)]}{D_E}.
 \label{eq:app-column-mixture-best-response-bound}
\end{equation}
Instantiate the dual inequality~\eqref{eq:app-second-row-instance}
of $\sm(y_1^1)$ with $x^{2\prime}=y_E^2$. Bound the row-deviation
payoff by~\eqref{eq:app-column-mixture-best-response-bound} and the
column best-response payoff by one. Expand each actual payoff
using~\eqref{eq:app-column-comparison-mixture}, and apply
\eqref{eq:app-payoff-bound-s-meaning},
\eqref{eq:app-payoff-bound-h-meaning}, and
\eqref{eq:app-payoff-bound-st-meaning}.
Multiplying by the positive denominator $D_E$ gives
\begin{equation}
 \begin{aligned}
 rD_E\le{}&
 \alpha\{1-r-(\gamma/\xi)[r-(1-\xi)]
       -(1-\gamma)[1-S(\boldsymbol\lambda)]-\gamma v[1-s_1(\boldsymbol\lambda)]\}\\
 &+(1-\alpha)[(1-\gamma)T(\boldsymbol\lambda)+\gamma v h_2(\boldsymbol\lambda)].
 \end{aligned}
 \label{eq:app-column-mixture-second-row-test}
\end{equation}
Grouping by the functions in~\eqref{eq:app-weighted-regret-bounds}
defines the left side of an additional necessary inequality:
\begin{equation}
 E_{1,g}=(1-\gamma)(H_g-r)+\gamma v(U_g-r)
               +\alpha\gamma(1-r-rv)-\alpha r\ge0.
 \label{eq:common-envelope}
\end{equation}

For completeness, the other player's calculation starts by
instantiating the dual inequalities~\eqref{eq:first-row-instance}
of $\sm(y_s^1)$ and~\eqref{eq:app-second-row-instance} of
$\sm(y_1^1)$ at the same $x^{2\prime}\in\Delta_{n_2}$. Eliminating their common
row payoff after multiplying the first rearranged inequality by
$\alpha/\eta>0$ gives
\begin{equation}
 \forall x^{2\prime}\in\Delta_{n_2},\quad
 (1-\alpha)u_2(y_1^1,x^{2\prime})+\alpha u u_2(y_s^1,x^{2\prime})
 \le1-r-\frac{\alpha}{\eta}[r-(1-\eta)].
 \label{eq:app-shared-row-payoff-bound}
\end{equation}
Define the row comparison strategy and its positive denominator by
\begin{equation}
 D_{E,1}=1-\alpha+\alpha u>0,\qquad
 y_E^1=\frac{(1-\alpha)y_1^1+\alpha u y_s^1}{D_{E,1}}\in\Delta_{n_1}.
 \label{eq:app-row-comparison-mixture}
\end{equation}
The best-response payoff $B_2(y_E^1)$ is bounded by the right side
of~\eqref{eq:app-shared-row-payoff-bound} divided by $D_{E,1}$.
Instantiate the dual inequality~\eqref{eq:app-second-column-instance}
of $\sm(y_1^2)$ with $x^{1\prime}=y_E^1$, expand payoffs, and
multiply by $D_{E,1}$:
\begin{equation}
 \begin{aligned}
 rD_{E,1}\le{}&
 \gamma\{1-r-(\alpha/\eta)[r-(1-\eta)]
       -(1-\alpha)[1-T(\boldsymbol\lambda)]-\alpha u[1-s_2(\boldsymbol\lambda)]\}\\
 &+(1-\gamma)[(1-\alpha)S(\boldsymbol\lambda)+\alpha u h_1(\boldsymbol\lambda)].
 \end{aligned}
 \label{eq:app-row-mixture-second-column-test}
\end{equation}
Grouping these terms gives
\begin{equation}
 E_{2,g}=(1-\alpha)(K_g-r)+\alpha u(V_g-r)
               +\alpha\gamma(1-r-ru)-\gamma r\ge0.
 \label{eq:app-second-shared-comparison}
\end{equation}
Thus $\sm(y_s^2)$ and $\sm(y_1^2)$ jointly bound the row
best-response payoff at the mixture $y_E^2$, and $\sm(y_s^1)$
and $\sm(y_1^1)$ jointly bound the column best-response payoff
at the mixture $y_E^1$. Instantiating the dual inequalities of
$\sm(y_1^1)$ and $\sm(y_1^2)$ at these respective mixtures yields
$E_{1,g},E_{2,g}\ge0$, which connect both first-round weights with
both second-round weights. The arithmetic regroupings and fixed
strategy bindings in this appendix are checked by
\nolinkurl{quantified-interface-identities}\footnote{Checker:
\nolinkurl{proof/bound/quantified/replay_quantified_identities.py}.\label{fn:checker-quantified-interface-identities}} and
\nolinkurl{ideal-cap-provenance-ledger}\footnote{Checker:
\nolinkurl{proof/bound/global/replay_ideal_cap_provenance_ledger.py}.\label{fn:checker-ideal-cap-provenance-ledger}}. Their inputs are the specified
building-block instantiations; LP duality and the strategy semantics
are supplied by the mathematical argument above. All checker paths
below are relative to the repository root.

\section{Jensen gaps and finite mixing comparisons}
\label{app:mixing}

We now analyze the $\om$ comparisons~\cite{mixing2025} over all
stationary/dual mixtures, using Jensen's inequality for convex
best-response payoffs~\cite{boyd2004}.
The zero-error and failed-output assumptions are
\eqref{eq:ideal-parameters} and~\eqref{eq:failure-hypothesis};
the resulting comparisons are~\eqref{eq:failure-candidates}.
The positive weights proved in
\eqref{eq:app-positive-stationary-weights} and
\eqref{eq:app-positive-call-weights} will justify every division below.

\paragraph{Joint bounds on endpoint regrets.}
The Jensen-gap identity~\eqref{eq:mixing-gap-losses} minimizes
the weighted endpoint regrets of one comparison strategy.
To bound the column Jensen gap, define its two regret functions on
the comparison simplex:
\begin{equation}
 \begin{aligned}
 \forall x^{2\prime}\in\Delta_{n_2},\quad
 A_c(x^{2\prime})&=f_2(x_s^1,x^{2\prime})\ge0,\\
 B_c(x^{2\prime})&=f_2(y_s^1,x^{2\prime})\ge0.
 \end{aligned}
 \label{eq:app-column-endpoint-losses}
\end{equation}
These are the same regrets introduced in
\eqref{eq:main-column-payoff-losses}, now written as functions of
the comparison strategy:
\[
 A=A_c(x^{2\prime}),\qquad B=B_c(x^{2\prime}).
\]
The subscript $c$ identifies column-player regrets.
The same comparison strategy occurs at the two row-strategy endpoints.
Instantiate~\eqref{eq:properties-stationary} at
$(x_s^1,x^{2\prime})$, and use
\eqref{eq:app-stationary-best-responses} and nonnegative payoffs:
\begin{equation}
 \begin{aligned}
 g&\le\rho[u_1(y_s^1,x^{2\prime})-u_1(x_s^1,x^{2\prime})]
                         +\sigma A_c(x^{2\prime}),\\
 u_1(y_s^1,x^{2\prime})
 &\ge g/\rho-(\sigma/\rho)A_c(x^{2\prime}).
 \end{aligned}
 \label{eq:app-column-loss-stationary-test}
\end{equation}
Instantiate the dual inequality~\eqref{eq:first-row-instance}
of the first-round operation $\sm(y_s^1)$ at the same
$x^{2\prime}$, and use this lower bound and
$u_1(y_1^1,x^{2\prime})\le1$. Since $F_1^1>r$, this gives
\begin{equation}
 r\le\eta[1-g/\rho+(\sigma/\rho)A_c(x^{2\prime})]
                         +(1-\eta)B_c(x^{2\prime}).
 \label{eq:app-column-loss-first-row-test}
\end{equation}
Rearrangement yields a joint lower bound valid for every column strategy.
For the row counterpart, define the functions
\begin{equation}
 \begin{aligned}
 \forall x^{1\prime}\in\Delta_{n_1},\quad
 A_r(x^{1\prime})&=f_1(x^{1\prime},x_s^2)\ge0,\\
 B_r(x^{1\prime})&=f_1(x^{1\prime},y_s^2)\ge0.
 \end{aligned}
 \label{eq:app-row-endpoint-losses}
\end{equation}
Instantiate~\eqref{eq:properties-stationary} at
$(x^{1\prime},x_s^2)$. Then instantiate the dual inequality
\eqref{eq:app-first-column-instance} of the first-round operation
$\sm(y_s^2)$ at the same $x^{1\prime}$, and use the resulting
column-payoff lower bound. Explicitly, the two instantiations give
\begin{equation}
 \begin{aligned}
 g&\le\rho A_r(x^{1\prime})
       +\sigma[u_2(x^{1\prime},y_s^2)-u_2(x^{1\prime},x_s^2)],\\
 u_2(x^{1\prime},y_s^2)
       &\ge g/\sigma-(\rho/\sigma)A_r(x^{1\prime}),\\
 r&\le\xi[1-g/\sigma+(\rho/\sigma)A_r(x^{1\prime})]
                         +(1-\xi)B_r(x^{1\prime}).
 \end{aligned}
 \label{eq:app-row-loss-tests}
\end{equation}
The two conclusions and their right sides are
\begin{equation}
 \begin{aligned}
 \forall x^{2\prime}\in\Delta_{n_2},\quad
 \frac{\eta\sigma}{\rho}A_c(x^{2\prime})
       +(1-\eta)B_c(x^{2\prime})&\ge\ell_c(\boldsymbol\lambda),&
 \ell_c(\boldsymbol\lambda)&=r-\eta(1-g/\rho),\\
 \forall x^{1\prime}\in\Delta_{n_1},\quad
 \frac{\xi\rho}{\sigma}A_r(x^{1\prime})
       +(1-\xi)B_r(x^{1\prime})&\ge\ell_r(\boldsymbol\lambda),&
 \ell_r(\boldsymbol\lambda)&=r-\xi(1-g/\sigma).
 \end{aligned}
 \label{eq:app-joint-endpoint-loss-bounds}
\end{equation}
When a right side is positive, no comparison strategy has zero
regret at both endpoints.

We next convert a joint lower bound into a lower bound on the
weighted regrets in the Jensen-gap identity
\eqref{eq:mixing-gap-losses}. The elementary calculation must work
even when $\ell_c(\boldsymbol\lambda)$ or $\ell_r(\boldsymbol\lambda)$ is negative.

\begin{lemma}[A weighted lower bound on two regrets]
\label{lem:app-weighted-losses}
Let $a,b\ge0$, $c_1,c_2>0$, and $\ell\in\mathbb R$ satisfy
$c_1a+c_2b\ge\ell$. For $t\in[0,1]$, define
$M(t)=\min\{(1-t)/c_1,t/c_2\}$. Then
\[
 (1-t)a+tb\ge M(t)\ell.
\]
\end{lemma}
\begin{proof}
Both $1-t-c_1M(t)$ and $t-c_2M(t)$ are nonnegative. Therefore
\begin{equation}
 \begin{aligned}
 (1-t)a+tb-M(t)\ell
  &=[1-t-c_1M(t)]a+[t-c_2M(t)]b\\
  &\quad+M(t)[c_1a+c_2b-\ell]\ge0.
 \end{aligned}
 \label{eq:app-weighted-loss-identity}
\end{equation}
The calculation uses $M(t)\ge0$ and does not require $\ell\ge0$.
\end{proof}

For the column regrets, apply the lemma with
$a=A_c(x^{2\prime})$, $b=B_c(x^{2\prime})$,
$c_1=\eta\sigma/\rho$, $c_2=1-\eta$, $\ell=\ell_c(\boldsymbol\lambda)$,
and $t=p$. For the row regrets, use
$a=A_r(x^{1\prime})$, $b=B_r(x^{1\prime})$,
$c_1=\xi\rho/\sigma$, $c_2=1-\xi$, $\ell=\ell_r(\boldsymbol\lambda)$,
and $t=q$. Minimize over the comparison strategy using
\eqref{eq:mixing-gap-losses}. This defines explicit Jensen-gap lower bounds:
\begin{equation}
 \begin{aligned}
 \forall p\in[0,1],\quad J_2(p)&\ge j_c(\boldsymbol\lambda,p),&
 j_c(\boldsymbol\lambda,p)&=\ell_c(\boldsymbol\lambda)
       \min\left\{\frac{(1-p)\rho}{\eta\sigma},\frac p{1-\eta}\right\},\\
 \forall q\in[0,1],\quad J_1(q)&\ge j_r(\boldsymbol\lambda,q),&
 j_r(\boldsymbol\lambda,q)&=\ell_r(\boldsymbol\lambda)
       \min\left\{\frac{(1-q)\sigma}{\xi\rho},\frac q{1-\xi}\right\}.
 \end{aligned}
 \label{eq:gap-lower-bounds}
\end{equation}
The functions $j_c,j_r$ are lower bounds on the actual Jensen gaps $J_2,J_1$.
A negative lower bound is still valid. We use these signed formulas
throughout the arithmetic reduction. The unit
\nolinkurl{full-BR-rectangle-bridge}\footnote{Checker:
\nolinkurl{proof/bound/quantified/replay_fullbr_rectangle_bridge.py}.\label{fn:checker-full-BR-rectangle-bridge}} reconstructs the joint-regret
combinations and the signed weighted-regret identities used to obtain
\eqref{eq:gap-lower-bounds}.

\paragraph{Upper bounds on both regrets throughout the rectangle.}
To turn the universal mixing comparison
\eqref{eq:failure-candidates} into finitely many conditions, we need
upper bounds that can be evaluated along specified mixing segments.
The stationary best-response identities
\eqref{eq:app-stationary-best-responses} and the payoff lower bounds
\eqref{eq:app-stationary-payoff-lower-bounds} give the endpoint bounds
\begin{equation}
 \begin{array}{c|cccc}
 (p,q)&(0,0)&(1,0)&(0,1)&(1,1)\\ \hline
 \text{row-regret upper bound}&g&0&1&1-g/\rho\\
 \text{column-regret upper bound}&g&1&0&1-g/\sigma
 \end{array}
 \label{eq:app-rectangle-endpoint-bounds}
\end{equation}
For example, the row regret at $(1,0)$ is zero because $y_s^1$
is a best response to $x_s^2$. At $(1,1)$ it is at most
$1-u_1(y_s^1,y_s^2)\le1-g/\rho$.
Insert~\eqref{eq:app-rectangle-endpoint-bounds} into the exact
regret interpolation~\eqref{eq:mixing-regrets}, and subtract
the lower bounds~\eqref{eq:gap-lower-bounds}. The resulting functions are
\begin{equation}
 \begin{aligned}
 \forall p,q\in[0,1],\quad f_1(\boldsymbol x(p,q))&\le R_g(p,q),&
 R_g(p,q)&=(1-p)(1-q)g+q(1-pg/\rho)-j_r(\boldsymbol\lambda,q),\\
 \forall p,q\in[0,1],\quad f_2(\boldsymbol x(p,q))&\le C_g(p,q),&
 C_g(p,q)&=(1-p)(1-q)g+p(1-qg/\sigma)-j_c(\boldsymbol\lambda,p).
 \end{aligned}
 \label{eq:rectangle-upper-bounds}
\end{equation}
Thus $R_g,C_g$ denote upper bounds on the row and column regrets;
they are distinct from the payoff matrices $R,C$.

Each minimum in~\eqref{eq:gap-lower-bounds} changes formula where
its two ratios are equal. Define the positive denominators and
the corresponding mixing coefficients by
\begin{equation}
 \begin{aligned}
 D_c&=\rho(1-\eta)+\sigma\eta>0,&
 p_\dagger&=\frac{\rho(1-\eta)}{D_c},\\
 D_r&=\sigma(1-\xi)+\rho\xi>0,&
 q_\dagger&=\frac{\sigma(1-\xi)}{D_r}.
 \end{aligned}
 \label{eq:app-mixing-cuts}
\end{equation}
Both coefficients lie strictly between zero and one.
They are determined by the ratios even when $\ell_c(\boldsymbol\lambda)$ or
$\ell_r(\boldsymbol\lambda)$ is zero. The two cuts partition $[0,1]^2$ into
four rectangles. On each rectangle $R_g,C_g$ are bilinear,
so they are affine along a segment with one coordinate fixed.
There are twelve distinct rectangle edges: each of three
vertical levels and three horizontal levels is split into two segments.

We will use the four segments joining the central cut point to
the boundary. Name their endpoints now for later reference:
\begin{equation}
 \begin{aligned}
 e_0&=(p_\dagger,q_\dagger),&
 e_1&=(p_\dagger,0),&
 e_2&=(p_\dagger,1),\\
 e_3&=(0,q_\dagger),&
 e_4&=(1,q_\dagger),&
 \boldsymbol c_j(\boldsymbol\lambda)&=(R_g(e_j),C_g(e_j))
       \quad(j\in\{0,\ldots,4\}).
 \end{aligned}
 \label{eq:app-mixing-segment-points}
\end{equation}
Here $e_j$ denotes a pair of mixing coefficients. Evaluation at $e_j$ means evaluation at its two
coordinates. The pair $\boldsymbol c_j(\boldsymbol\lambda)$ consists of regret
upper bounds. Omitting the other eight edge conditions enlarges
the necessary system that the proof must exclude.

\paragraph{The finite condition on one segment.}
We eliminate the continuous segment parameter by identifying the
possible minimizers of the maximum of two affine functions.
For this calculation, let
$\boldsymbol a=(R_0,C_0),\boldsymbol b=(R_1,C_1)\in\mathbb R^2$
be the two endpoint upper-bound pairs. Parameterize the segment
by $t\in[0,1]$ and define
\begin{equation}
 \begin{aligned}
 L_1(t)&=R_0+(R_1-R_0)t,&
 L_2(t)&=C_0+(C_1-C_0)t,\\
 D&=(R_1-R_0)-(C_1-C_0),&
 N&=C_0-R_0,&
 W&=R_1C_0-R_0C_1.
 \end{aligned}
 \label{eq:app-affine-segment-data}
\end{equation}
The symbols $R_0,R_1,C_0,C_1,D,N,W$ are local to this affine
calculation. If $D\ne0$, the lines intersect at $t=N/D$ with
common value $W/D$. No intersection condition is needed outside $[0,1]$.
When $D=0$, endpoint conditions suffice.
The maximum is affine on either side of an intersection, so its
minimum occurs at an endpoint or an intersection in the segment.
Define the segment predicate by
\begin{equation}
 \mathcal T_r(\boldsymbol a,\boldsymbol b)
 \quad\Longleftrightarrow\quad
 \begin{gathered}
 \max\{R_0,C_0\}\ge r,\qquad \max\{R_1,C_1\}\ge r,\\
 \text{if }D\ne0\text{ and }0\le N/D\le1,\text{ then }W/D\ge r.
 \end{gathered}
 \label{eq:segment-test}
\end{equation}
The ratios in the last condition are evaluated only when $D\ne0$.
Equivalently, the predicate says that
$\min_{t\in[0,1]}\max\{L_1(t),L_2(t)\}\ge r$.

Every actual profile on one of the twelve edges has maximum regret
above $r$ by~\eqref{eq:failure-candidates}. The upper bounds
\eqref{eq:rectangle-upper-bounds} therefore satisfy the strict
affine minimax inequality there; continuity on the compact
segment justifies taking its minimum. Weakening it to
\eqref{eq:segment-test} gives a necessary condition.
When clearing denominators, we must check their signs.
For example, if $R_0<r<R_1$ and $C_1<r<C_0$, then $D>0$,
the intersection lies in the segment, and its non-strict condition is
\begin{equation}
 (C_0-r)(R_1-r)-(R_0-r)(C_1-r)=W-rD\ge0.
 \label{eq:app-positive-denominator-intersection-test}
\end{equation}
Thus the inequality is equivalent to $W/D\ge r$ in this case.
A symbolic routine constructs endpoint conditions, intersection
conditions, and sign-correct polynomial inequalities for all
twelve edges; the independent unit
\nolinkurl{parametric-OM-linearizer}\footnote{Checker:
\nolinkurl{proof/bound/quantified/replay_parametric_optimalmix.py}.\label{fn:checker-parametric-OM-linearizer}} reconstructs 131 identities
for these constructions.
The four segment conditions~\eqref{eq:main-segment-conditions}
are expressed by this finite predicate as
\begin{equation}
 \forall j\in[4],\qquad
 \mathcal T_r(\boldsymbol c_0(\boldsymbol\lambda),\boldsymbol c_j(\boldsymbol\lambda)).
 \label{eq:app-four-segment-comparisons}
\end{equation}
This calculation eliminates one-dimensional affine upper-bound
problems. It supplies necessary conditions for a failed execution
on the selected segments.

\section{Forgetting and the finite necessary system}
\label{app:scalar}

We now eliminate the remaining game-dependent quantities from the
comparisons proved above. Throughout, the zero-error assumption is
\eqref{eq:ideal-parameters}, and the failed-output assumption is
\eqref{eq:failure-hypothesis}. The purpose is to obtain necessary
conditions involving only $r$ and five dual weights.

\paragraph{Consistent replacement of payoff and regret terms.}
An instantiated building-block clause still concerns strategies
in the fixed game. Forgetting replaces each distinct payoff,
regret, or best-response payoff term by a real variable, using the
same variable at every occurrence of the same term,
as in the LegoNE abstraction~\cite{computeraided2026,legone2026}.
For this example only, introduce the scalar names
\begin{equation}
 \widehat t=F_1^1,\qquad
 \widehat a_1=u_1(y_1^1,x_s^2),\qquad
 \widehat a_0=u_1(y_s^1,x_s^2),\qquad
 \widehat b=f_2(y_s^1,x_s^2).
 \label{eq:app-scalar-renaming}
\end{equation}
These scalar names represent the displayed payoff and regret values.
The instantiation of the dual inequality
\eqref{eq:first-row-instance} of $\sm(y_s^1)$ with
$x^{2\prime}=x_s^2$, the payoff range, and
\eqref{eq:app-stationary-best-responses} become
\begin{equation}
 \widehat t\le\eta(\widehat a_1-\widehat a_0)
                      +(1-\eta)\widehat b,\qquad
 0\le\widehat a_0,\widehat a_1,\widehat b\le1,\qquad
 \widehat a_1\le\widehat a_0.
 \label{eq:app-scalar-first-row-test}
\end{equation}
The additional instantiation at a column best response to $y_s^1$,
recorded in~\eqref{eq:app-first-row-best-response-test}, and the
failure~\eqref{eq:failure-candidates} give
$r<\widehat t\le\min\{\eta,1-\eta\}$.
Likewise, both occurrences of $u_2(x^{1\prime},y_1^2)$ in
\eqref{eq:app-shared-column-tests} must be replaced by the same
scalar for their cancellation to remain valid.
Relations proved from the regret definition, payoff linearity,
LP duality, and complementary slackness~\cite{boyd2004}
remain conditions on
the replacement variables.

We simplify further by eliminating many payoff variables through
their lower bounds. The resulting regret upper bounds depend
only on weights, $r$, and $g$. Every failed execution supplies
values satisfying the resulting conditions; a solution of the
scalar conditions need not be realizable by a game.
Deleting a constraint, increasing an upper bound, or changing
a scalar point must therefore be justified in the necessary
direction. The unit \nolinkurl{ideal-cap-provenance-ledger}\textsuperscript{\ref{fn:checker-ideal-cap-provenance-ledger}}
checks the strategy types and fixed dual solutions, then
reconstructs the nonnegative combinations, positive divisions,
and payoff mixing identities used here.

\paragraph{Eliminating the stationary regret.}
The remaining game-dependent scalar is the stationary regret $g$.
Setting it to $r$ requires proof: the term $(1-p)(1-q)g$ in
\eqref{eq:rectangle-upper-bounds} increases with $g$, whereas
other terms decrease. We first identify the admissible interval
and then prove preservation of the selected inequalities while
lowering $g$. The sign certificates cover $3/10\le r\le1/3$.

Instantiate the dual inequality~\eqref{eq:first-row-instance}
of the first-round operation $\sm(y_s^1)$ with $x^{2\prime}=y_s^2$
and the dual inequality~\eqref{eq:app-first-column-instance}
of the first-round operation $\sm(y_s^2)$ with $x^{1\prime}=y_s^1$.
These instantiations use the stationary dual profile $(y_s^1,y_s^2)$.
Use~\eqref{eq:app-stationary-payoff-lower-bounds}, bound the other
deviation payoff by one, and use~\eqref{eq:failure-candidates}.
The resulting two inequalities are
\begin{equation}
 \begin{aligned}
 r&\le\eta(1-g/\rho)+(1-\eta)(1-g/\sigma),\\
 r&\le\xi(1-g/\sigma)+(1-\xi)(1-g/\rho).
 \end{aligned}
 \label{eq:app-first-round-g-comparisons}
\end{equation}
With the positive $D_c,D_r$ from~\eqref{eq:app-mixing-cuts},
solving these inequalities bounds $g$ from above.
Weakening the strict weight and regret inequalities from
\eqref{eq:app-positive-stationary-weights},
\eqref{eq:app-positive-call-weights}, and
\eqref{eq:failure-candidates} gives the domain used in the projection:
\begin{equation}
 \begin{gathered}
 3/10\le r\le1/3,\qquad
 r\le\rho,\eta,\xi,\alpha,\gamma\le1-r,\qquad \sigma=1-\rho,\\
 r\le g\le
 \min\left\{\frac{\rho\sigma(1-r)}{D_c},
             \frac{\rho\sigma(1-r)}{D_r}\right\}.
 \end{gathered}
 \label{eq:g-projection-domain}
\end{equation}

The projection will preserve the nonnegative payoff upper bounds,
the four comparisons from the second-round operations $\sm(y_1^1)$
and $\sm(y_1^2)$, the two comparisons obtained by instantiating
their dual inequalities at mixtures, and the four selected segment
conditions.
For a general scalar point, record these premises as
\begin{equation}
 \begin{gathered}
 s_1(\boldsymbol\lambda),s_2(\boldsymbol\lambda),h_1(\boldsymbol\lambda),h_2(\boldsymbol\lambda),S(\boldsymbol\lambda),T(\boldsymbol\lambda)\ge0,\\
 U_g,V_g,H_g,K_g\ge r,\qquad E_{1,g},E_{2,g}\ge0,\\
 \forall j\in[4],\qquad
 \mathcal T_r(\boldsymbol c_0(\boldsymbol\lambda),\boldsymbol c_j(\boldsymbol\lambda)).
 \end{gathered}
 \label{eq:app-g-projection-premises}
\end{equation}
The functions and coefficient points are defined in
\eqref{eq:payoff-bound-s}, \eqref{eq:payoff-bound-h},
\eqref{eq:payoff-bound-st},
\eqref{eq:app-weighted-regret-bounds},
\eqref{eq:common-envelope},
\eqref{eq:app-second-shared-comparison}, and
\eqref{eq:app-mixing-segment-points}. The predicate is
\eqref{eq:segment-test}. A failed execution satisfies these premises
by the payoff interpretations, the instantiations of the dual
inequalities of $\sm(y_1^1)$ and $\sm(y_1^2)$, and
\eqref{eq:app-four-segment-comparisons}.

\begin{lemma}[Lowering the stationary regret]
\label{lem:app-g-projection}
Let a scalar point satisfy the domain~\eqref{eq:g-projection-domain}
and the inequalities~\eqref{eq:app-g-projection-premises}.
Hold $r,\rho,\eta,\xi,\alpha,\gamma$ fixed and evaluate every
function in~\eqref{eq:app-g-projection-premises} with $g=r$.
All the inequalities in~\eqref{eq:app-g-projection-premises}
continue to hold.
\end{lemma}
\begin{proof}
We first treat the payoff upper bounds. Direct differentiation gives
\begin{equation}
 \frac{\partial s_1(\boldsymbol\lambda)}{\partial g}=-\frac1\rho-\frac u\sigma<0,\qquad
 \frac{\partial s_2(\boldsymbol\lambda)}{\partial g}=-\frac1\sigma-\frac v\rho<0.
 \label{eq:app-g-payoff-bound-derivatives}
\end{equation}
Each $h_i(\boldsymbol\lambda)$ is the minimum of one and an affine function decreasing
in $g$, by~\eqref{eq:payoff-bound-h}. Thus lowering $g$ cannot decrease
any $s_i(\boldsymbol\lambda)$ or $h_i(\boldsymbol\lambda)$.
The formulas~\eqref{eq:payoff-bound-st}, with $u,v>0$, imply the
same for $S(\boldsymbol\lambda),T(\boldsymbol\lambda)$. The nonnegative coefficients in
\eqref{eq:app-weighted-regret-bounds} then give the same conclusion
for $U_g,V_g,H_g,K_g$.
In~\eqref{eq:common-envelope} and
\eqref{eq:app-second-shared-comparison}, only these weighted bounds
depend on $g$, and their coefficients are nonnegative.
Therefore $E_{1,g},E_{2,g}$ also cannot decrease.
This proves preservation of the first two lines of
\eqref{eq:app-g-projection-premises}.

We next treat the four segment predicates, beginning with their
shared central endpoint $e_0=(p_\dagger,q_\dagger)$.
Differentiate~\eqref{eq:rectangle-upper-bounds}, using the equality
of the two ratios at each cut in~\eqref{eq:app-mixing-cuts}:
\begin{equation}
 \begin{aligned}
 \partial_gR_g(p_\dagger,q_\dagger)
 &=\frac{\xi}{D_r}[\rho(1-p_\dagger)-1]
                       -\frac{p_\dagger q_\dagger}{\rho}<0,\\
 \partial_gC_g(p_\dagger,q_\dagger)
 &=\frac{\eta}{D_c}[\sigma(1-q_\dagger)-1]
                       -\frac{p_\dagger q_\dagger}{\sigma}<0.
 \end{aligned}
 \label{eq:app-central-g-derivatives}
\end{equation}
The signs follow from $0<\rho,\sigma,p_\dagger,q_\dagger<1$.
At the two endpoints with a coefficient equal to one, the Jensen gap
associated with that coefficient is zero and direct differentiation gives
\begin{equation}
 \begin{aligned}
 \partial_gR_g(p_\dagger,1)&=-p_\dagger/\rho<0,\\
 \partial_gC_g(p_\dagger,1)
 &=-p_\dagger/\sigma-\eta p_\dagger/[\rho(1-\eta)]<0,\\
 \partial_gR_g(1,q_\dagger)
 &=-q_\dagger/\rho-\xi q_\dagger/[\sigma(1-\xi)]<0,\\
 \partial_gC_g(1,q_\dagger)&=-q_\dagger/\sigma<0.
 \end{aligned}
 \label{eq:app-coefficient-one-g-derivatives}
\end{equation}
Both upper-bound functions are affine along each selected segment.
Their derivatives at an interior point are therefore convex
combinations of their endpoint derivatives.
Equations~\eqref{eq:app-central-g-derivatives} and
\eqref{eq:app-coefficient-one-g-derivatives} prove that lowering
$g$ cannot decrease either bound on the segments from $e_0$ to
$e_2$ and $e_4$. Their predicates are preserved.

For the segment from $e_1=(p_\dagger,0)$ to $e_0$, we must also
consider an intersection of the two affine bounds. Recall
$\mathcal D(g)$ from~\eqref{eq:main-g-intersection-condition}, which
expresses the intersection condition:
\begin{equation}
 \begin{aligned}
 \mathcal D(g)=\frac1{q_\dagger}\bigl\{&
 [C_g(p_\dagger,0)-r][R_g(p_\dagger,q_\dagger)-r]\\
 &-[R_g(p_\dagger,0)-r][C_g(p_\dagger,q_\dagger)-r]\bigr\}.
 \end{aligned}
 \label{eq:app-zero-endpoint-determinant}
\end{equation}
Division by $q_\dagger>0$ preserves its sign.
The endpoint ordering follows from
\eqref{eq:app-first-round-g-comparisons}, which are equivalent to
the upper bounds on $g$ in~\eqref{eq:g-projection-domain}.
To see the ordering explicitly, define two temporary functions
for the affine row formula on this segment:
\begin{equation}
 B_{\mathrm e}(g)=(1-p_\dagger)g,\qquad
 U_{\mathrm e}(g)=1-p_\dagger g/\rho-\ell_r(\boldsymbol\lambda)/(1-\xi).
 \label{eq:app-lower-segment-affine-values}
\end{equation}
$B_{\mathrm e}(g)$ is the row upper bound at $e_1$;
$U_{\mathrm e}(g)$ is its affine extrapolation to coefficient one.
It is distinct from the weighted bound $U_g$ in
\eqref{eq:app-weighted-regret-bounds}.
For $0\le q\le q_\dagger$, the formula is
\begin{equation}
 R_g(p_\dagger,q)
       =(1-q)B_{\mathrm e}(g)+qU_{\mathrm e}(g).
 \label{eq:app-lower-segment-row-formula}
\end{equation}
Direct subtraction and the first-round comparisons give
\begin{equation}
 \begin{aligned}
 \relax[U_{\mathrm e}(g)-B_{\mathrm e}(g)]
        -(1-p_\dagger)(g/\rho-g)
 &=\frac{1-r-g((1-\xi)/\rho+\xi/\sigma)}{1-\xi}\ge0,\\
 C_g(p_\dagger,0)-R_g(p_\dagger,0)
 &=\frac{\rho(1-r)-\eta g}{D_c}\ge0.
 \end{aligned}
 \label{eq:app-lower-segment-endpoint-order}
\end{equation}
For the second sign, the first upper bound on $g$ gives
$\eta g\le\eta\rho\sigma(1-r)/D_c\le\rho(1-r)$,
because $\eta\sigma\le D_c$.
Since $g>0$, $\rho<1$, and $p_\dagger<1$,
the first line implies $U_{\mathrm e}(g)-B_{\mathrm e}(g)>0$.
The row bound increases with $q$, whereas the column bound decreases:
\begin{equation}
 \partial_qC_g(p_\dagger,q)
        =-(1-p_\dagger)g-p_\dagger g/\sigma<0,\qquad
 \partial_gC_g(p_\dagger,0)=-\rho\eta/D_c<0.
 \label{eq:app-lower-segment-column-derivatives}
\end{equation}
At $e_1$ the column bound is the larger one, so its maximum
cannot decrease when $g$ is lowered. The central endpoint is
protected by~\eqref{eq:app-central-g-derivatives}.
If the increasing row line and decreasing column line intersect
inside the segment, their common value is at least $r$ exactly
when $\mathcal D(g)\ge0$, by~\eqref{eq:segment-test} and
\eqref{eq:app-positive-denominator-intersection-test}.
We will prove $\partial_g\mathcal D(g)<0$ throughout the admissible
interval. This preserves the intersection condition as $g$ decreases.
An intersection can enter or leave the segment only through an
endpoint, whose condition is already preserved. Continuity consequently
preserves the full segment predicate.

The derivative sign reduces to a polynomial sign.
Define the function $\mathcal Q(g)$ and the positive coefficient $Z_Q$ by
\begin{equation}
 \begin{aligned}
 \partial_g\mathcal D(g)
    &=-\frac{\mathcal Q(g)}{\sigma D_c^2(1-\xi)},\\
 Z_Q&=\rho\sigma\eta\xi+\sigma^2\eta(1-\xi)
       +\rho^2(1-\eta)\xi+2\rho\sigma(1-\eta)(1-\xi)>0.
 \end{aligned}
 \label{eq:app-determinant-derivative-numerator}
\end{equation}
The numerator can be written explicitly as
\begin{equation}
 \begin{aligned}
 \mathcal Q(g)={}&
 [\sigma(1-\xi)(1-\rho\eta)+D_c\xi]
                   [\rho(1-r)-rD_c]\\
 &+\sigma D_c\rho\eta(1-r)
       +rD_c(1-\xi)[\sigma^2\eta+\rho(1-\eta)]
       -2\eta Z_Q g .
 \end{aligned}
 \label{eq:app-determinant-numerator-polynomial}
\end{equation}
The unit \nolinkurl{g-identities}\footnote{Checker:
\nolinkurl{proof/analytic_scalar_bound/cross_projection_symbolic_checks.py}.\label{fn:checker-g-identities}} reconstructs this derivative
identity and the endpoint and intersection identities above.
This formula follows by expanding
\eqref{eq:app-zero-endpoint-determinant} using
\eqref{eq:app-lower-segment-row-formula} and
\eqref{eq:rectangle-upper-bounds}.
It gives $\partial_g\mathcal Q(g)=-2\eta Z_Q<0$.
Thus it suffices to prove positivity at the largest admissible $g$.
Define that endpoint, its feasibility margin, and a cleared numerator by
\begin{equation}
 \begin{aligned}
 D_*&=\max\{D_c,D_r\},&
 G_*&=\frac{\rho\sigma(1-r)}{D_*},\\
 \mathcal B_*&=\rho\sigma(1-r)-rD_*\ge0,&
 \mathcal H_*&=D_*\mathcal Q(G_*).
 \end{aligned}
 \label{eq:app-g-upper-endpoint}
\end{equation}
The sign of $\mathcal B_*$ is exactly $G_*\ge r$, which follows
from~\eqref{eq:g-projection-domain}.

Before invoking the certificates, specify the four parameter
regions and their unit-box coordinates. Split the weight domain
at $\rho=1/2$ and $\eta+\xi=1$. The identity
\begin{equation}
 D_c-D_r=(\rho-\sigma)(1-\eta-\xi)
 \label{eq:app-denominator-order}
\end{equation}
determines which denominator equals $D_*$ in each region.
Let $t_r,t_\rho,t_\eta,t_\xi\in[0,1]$, and substitute
\begin{equation}
 \begin{aligned}
 r&=3/10+t_r/30,&
 \eta&=r+(1-2r)t_\eta,\\
 \rho&=
 \begin{cases}
 r+(1/2-r)t_\rho,&\rho\le1/2,\\
 1/2+(1/2-r)t_\rho,&\rho\ge1/2,
 \end{cases}
 &
 \xi&=
 \begin{cases}
 r+(1-\eta-r)t_\xi,&\eta+\xi\le1,\\
 1-\eta+(\eta-r)t_\xi,&\eta+\xi\ge1.
 \end{cases}
 \end{aligned}
 \label{eq:app-g-certificate-coordinates}
\end{equation}
These substitutions cover all four regions, including their
boundaries. After clearing positive denominators, the checker
reconstructs the resulting polynomials and their rational
Bernstein coefficients~\cite{garloff1986}. For degree at most $d$,
the Bernstein basis on $[0,1]$ satisfies
\begin{equation}
 b_i^{(d)}(t)=\binom di t^i(1-t)^{d-i},\quad i=0,\ldots,d,\quad
 b_i^{(d)}(t)\ge0,\quad\sum_{i=0}^d b_i^{(d)}(t)=1.
 \label{eq:bernstein-basis}
\end{equation}
On a rational box $\prod_j[a_j,b_j]$, substitute
$x_j=a_j+(b_j-a_j)t_j$ and expand in products of these basis
functions. The polynomial's value is a convex combination of its
coefficients, so their signs certify its sign throughout the box.
The checker uses Python~3 rational arithmetic with
\nolinkurl{fractions.Fraction}~\cite{python3} to verify these signs
exactly.

On the two regions with $\rho\le1/2$, choose the positive rational
certificate multiplier $m_Q=12609/50000$. Subtracting
$m_Q\mathcal B_*$ produces a polynomial whose strict positivity
can be checked through its Bernstein coefficients; the known
$\mathcal B_*\ge0$ then recovers positivity of $\mathcal H_*$.
On the other two regions the certificate proves positivity
of $\mathcal H_*$ directly. The certified inequalities are
\begin{equation}
 \begin{cases}
 \mathcal H_*-m_Q\mathcal B_*>0,\quad m_Q=12609/50000,
                         &\rho\le1/2,\\
 \mathcal H_*>0,          &\rho\ge1/2.
 \end{cases}
 \label{eq:app-g-numerator-sign-certificates}
\end{equation}
All reconstructed Bernstein coefficients for these stated
polynomials are strictly positive on the respective unit boxes.
The unit \nolinkurl{g-derivative-small}\footnote{Checker:
\nolinkurl{proof/analytic_scalar_bound/replay_cross_derivative_certificate.py}.\label{fn:checker-g-derivative-small}} reconstructs the two
$\rho\le1/2$ certificates, and \nolinkurl{g-derivative-large}\footnote{Checker:
\nolinkurl{proof/analytic_scalar_bound/replay_cross_derivative_large_mass_certificate.py}.\label{fn:checker-g-derivative-large}}
reconstructs the two $\rho\ge1/2$ certificates.
Both certificates imply $\mathcal Q(G_*)>0$ because $D_*>0$.
Since $\mathcal Q$ decreases with $g$ and $g\le G_*$,
\begin{equation}
 \mathcal Q(g)\ge\mathcal Q(G_*)>0,\qquad
 \partial_g\mathcal D(g)<0.
 \label{eq:app-g-determinant-sign}
\end{equation}
This completes the segment from $e_1$ to $e_0$.

The other coefficient-zero segment, from $e_3$ to $e_0$, follows
from the same identities under the explicit interchange
\begin{equation}
 (\rho,\sigma,\eta,\xi,p,q,R_g,C_g)
       \longleftrightarrow
 (\sigma,\rho,\xi,\eta,q,p,C_g,R_g).
 \label{eq:app-player-interchange}
\end{equation}
It exchanges $D_c$ with $D_r$ and the two first-round inequalities
\eqref{eq:app-first-round-g-comparisons}. The domain
\eqref{eq:g-projection-domain} is invariant under the interchange,
so the preceding endpoint, determinant, and certificate arguments
cover this segment as well. Together with the coefficient-one
segments, all four predicates in~\eqref{eq:app-g-projection-premises}
are preserved while $g$ decreases to $r$.
\end{proof}

\paragraph{The projected upper-bound functions.}
Lemma~\ref{lem:app-g-projection} associates each failed execution
with a scalar point at $g=r$. We now evaluate the arithmetic
functions at this point.
There is one further relaxation: replace $d_1,d_2$ in
\eqref{eq:payoff-bound-h} by $(1-\eta-r)/\eta$ and
$(1-\xi-r)/\xi$, respectively. Each replacement increases or leaves
unchanged the corresponding $h_i$. With nonnegative $u,v$ and
dual weights, it also increases or leaves unchanged
$S,T,U,V,H,K,E_1,E_2$.
For compact formulas, define
\begin{equation}
 \chi_1=1-r-r/\rho,\qquad
 \chi_2=1-r-r/\sigma,\qquad c=1-2r.
 \label{eq:app-projected-coefficients}
\end{equation}
The scalar $c$ is a recurring coefficient.
The symbols without an argument below are working names for
the projected functions, including the stated relaxation of $h_i$:
\begin{equation}
 \begin{aligned}
 s_1&=\chi_1+\chi_2u,&
 s_2&=\chi_2+\chi_1v,\\
 h_1&=\min\{1,1-r(1+\sigma)/\rho+\sigma(1-r)v/\rho\},&
 h_2&=\min\{1,1-r(1+\rho)/\sigma+\rho(1-r)u/\sigma\},\\
 S&=h_1-r+u(s_2-r),&
 T&=h_2-r+v(s_1-r).
 \end{aligned}
 \label{eq:scalar-caps}
\end{equation}
In terms of $\boldsymbol\lambda_r$ from
\eqref{eq:projected-payoff-parameters}, the relation to the original
functions is
\[
 \forall i\in[2],\quad s_i=s_i(\boldsymbol\lambda_r),\qquad
 h_i\ge h_i(\boldsymbol\lambda_r).
\]
The larger $h_i$ result from using the uncapped bounds in their
affine formulas. The functions $S,T$ are then recalculated using
these relaxed $h_i$. For $s_i,S,T$, we use the displayed formulas directly.
All six displayed upper bounds are nonnegative at the scalar
point produced by the execution and the preceding projection.
Later arithmetic transformations must establish the signs they need.

Use the same substitutions in the four weighted functions
\eqref{eq:app-weighted-regret-bounds} and the two inequalities
\eqref{eq:common-envelope} and~\eqref{eq:app-second-shared-comparison}.
Their projected working names and formulas are
\begin{equation}
 \begin{aligned}
 U&=\alpha s_1+(1-\alpha)h_2\ge r,&
 V&=(1-\gamma)h_1+\gamma s_2\ge r,\\
 H&=\alpha S+(1-\alpha)T\ge r,&
 K&=(1-\gamma)S+\gamma T\ge r,\\
 E_1&=(1-\gamma)(H-r)+\gamma v(U-r)
                         +\alpha\gamma(1-r-rv)-\alpha r\ge0,\\
 E_2&=(1-\alpha)(K-r)+\alpha u(V-r)
                         +\alpha\gamma(1-r-ru)-\gamma r\ge0.
 \end{aligned}
 \label{eq:scalar-sources}
\end{equation}
Each symbol denotes the displayed arithmetic function of $r$
and the five weights.

\paragraph{Definition of the finite necessary system.}
The preceding projection removes $g$ and preserves the selected
necessary conditions. We can now specify their conjunction
explicitly in five scalar variables.
With $e_j$ from~\eqref{eq:app-mixing-segment-points}, define
\begin{equation}
 \boldsymbol c_j=\boldsymbol c_j(\boldsymbol\lambda_r)
       =(R_r(e_j),C_r(e_j))
       \quad(j\in\{0,\ldots,4\}),\qquad
 \boldsymbol\theta=(\rho,\eta,\xi,\alpha,\gamma)\in\mathbb R^5.
 \label{eq:app-projected-segment-pairs}
\end{equation}
Here $R_r,C_r$ mean the functions $R_g,C_g$ with $g=r$;
the bold pairs record numerical regret upper bounds.
For $3/10\le r\le1/3$, let $\Sigma_r(\boldsymbol\theta)$ be
the conjunction
\begin{equation}
 \Sigma_r(\boldsymbol\theta)
 \quad\Longleftrightarrow\quad
 \begin{gathered}
 r\le\rho,\eta,\xi,\alpha,\gamma\le1-r,\\
 s_1,s_2,h_1,h_2,S,T\ge0,\\
 U,V,H,K\ge r,\qquad E_1,E_2\ge0,\\
 \forall j\in[4],\qquad
 \mathcal T_r(\boldsymbol c_0,\boldsymbol c_j).
 \end{gathered}
 \label{eq:scalar-system}
\end{equation}
The first line bounds the five weights. The second expresses
nonnegativity of the six payoff upper bounds. The third records
the four second-round comparisons and the two mixture comparisons.
The last line uses the finite affine segment predicate
\eqref{eq:segment-test}.
All quantities are the functions defined in
\eqref{eq:relative-weights}, \eqref{eq:app-mixing-cuts},
\eqref{eq:scalar-caps}, \eqref{eq:scalar-sources}, and
\eqref{eq:app-projected-segment-pairs}, with $\sigma=1-\rho$.

The first-round inequalities used in the projection are already
implied by this conjunction. To check that implication explicitly,
subtract $r$ from their projected upper bounds:
\begin{equation}
 \begin{aligned}
 \eta(1-r/\rho)+(1-\eta)(1-r/\sigma)-r&=\eta s_1,\\
 \xi(1-r/\sigma)+(1-\xi)(1-r/\rho)-r&=\xi s_2,\\
 \eta h_1+(1-\eta)s_2-r&=\eta S,\\
 \xi h_2+(1-\xi)s_1-r&=\xi T.
 \end{aligned}
 \label{eq:app-first-round-projected-identities}
\end{equation}
Their nonnegative right sides give precisely the needed
non-strict inequalities.
Instantiating the dual inequality of the second-round operation
$\sm(y_1^1)$ with $x^{2\prime}=x_s^2$ gives the upper bound
$\alpha d_1+1-\alpha$, which is at least $r$ because $d_1\ge0$
and $\alpha\le1-r$. Instantiating the dual inequality of the
second-round operation $\sm(y_1^2)$ with $x^{1\prime}=x_s^1$
gives $\gamma d_2+1-\gamma\ge r$ in the same way. At the respective first-round primal outputs $x_1^2,x_1^1$,
the exact best-response conclusions instead give the upper bounds
$1-\alpha$ and $1-\gamma$, also at least $r$ by the weight ranges.
Thus these instantiations impose no additional condition on
\eqref{eq:scalar-system}. Further mixture comparisons could
strengthen the system, but are not part of its definition.

Every failed execution first yields
\eqref{eq:g-projection-domain} and~\eqref{eq:app-g-projection-premises}.
Lemma~\ref{lem:app-g-projection} and the subsequent increase of
the payoff upper bounds then yield~\eqref{eq:scalar-system}.
This proves the implication
\begin{equation}
 f(\boldsymbol x_{\mathrm{out}})>r
 \quad\Longrightarrow\quad
 \exists\boldsymbol\theta\in\mathbb R^5,\ 
             \Sigma_r(\boldsymbol\theta),
 \qquad 3/10\le r\le1/3.
 \label{eq:app-failure-implies-scalar-system}
\end{equation}
The system has fixed dimension, independent of the game.
Replacing strict failure inequalities by non-strict necessary
conditions and increasing upper bounds both enlarge the system.
Consequently, infeasibility of $\Sigma_r$ excludes a failed
execution; realizability of every scalar point is not needed.

\section{The algebraic definition of the bound}
\label{app:discovery}

This appendix describes how we compute a candidate bound and use
the returned active comparisons to define it by exact polynomial
equations. Its starting point is the finite
necessary system~\eqref{eq:scalar-system}:
\begin{equation}
 \operatorname*{maximize}_{r,\boldsymbol\theta}\ r
 \quad\text{subject to}\quad 3/10\le r\le1/3,\quad
                              \Sigma_r(\boldsymbol\theta).
 \label{eq:app-bound-search}
\end{equation}
We minimize $-r$ subject to all the inequalities in
\eqref{eq:scalar-system}, using SciPy's multi-start SLSQP
routine~\cite{scipy2020,kraft1988}. The segment minima are evaluated
by the endpoint and intersection cases in~\eqref{eq:segment-test}.
The computation is implemented by \nolinkurl{find_bound.py}\footnote{Script:
\nolinkurl{tools/find_bound.py}. Parameters and reproduction instructions:
\nolinkurl{tools/README.md}.}.
Its returned points and constraint residuals give a candidate
threshold near $0.3095399654$ and identify the active comparisons
used in the following algebraic calculation. Exchanging the players
if necessary, we describe the candidate in the orientation $S>r>T$.
An inequality is called active when it holds with
equality. Near this value, five inequalities appear active: the
uncapped expression for $h_1$ equals one, $H=K=r$, one interior
mixing comparison holds with equality, and $E_1=0$. Here $H,K,E_1$
are the expressions in~\eqref{eq:scalar-sources}. These observations
motivate the equations below; they place no restriction on an
arbitrary solution of $\Sigma_r$.

Recall $\sigma=1-\rho$ and the relative weights $u,v$ from
\eqref{eq:relative-weights}. Setting the uncapped formula for $h_1$
in~\eqref{eq:scalar-caps} equal to one gives its change of formula:
\begin{equation}
 v_0=\frac{r(1+\sigma)}{\sigma(1-r)},\qquad
 v=v_0\quad\Longleftrightarrow\quad
 \xi=\frac{\sigma(1-r)}{\sigma+r},\qquad h_1=1.
 \label{eq:app-cap-switch}
\end{equation}
At the numerical candidate $S>r>T$. Under this condition, substituting
$H=K=r$ into their linear formulas in~\eqref{eq:scalar-sources}
determines both second-round weight variables:
\begin{equation}
 A_0=\frac{r-T}{S-T},\qquad \alpha=A_0,\qquad\gamma=1-A_0.
 \label{eq:app-observed-weights}
\end{equation}
The candidate also has $h_2<1$. For the following local calculation,
we therefore fix $v=v_0$, use $h_1=1$, and substitute
\begin{equation}
 \sigma=1-\rho,\quad u=\frac{1-\eta}{\eta},\quad
 \xi=\frac{\sigma(1-r)}{\sigma+r},\quad
 h_2=1-\frac{r(1+\rho)}{\sigma}
                +\frac{\rho(1-r)}{\sigma}u.
 \label{eq:app-root-substitution}
\end{equation}
The formulas for $s_1,s_2,S,T$ remain those
in~\eqref{eq:scalar-caps}, with these substitutions. This specifies
which rational functions are used before defining their polynomial
numerators. The global argument must later justify the use of this
formula for $h_2$ at the points left by its reductions.

Substitute~\eqref{eq:app-observed-weights} into $E_1$
in~\eqref{eq:scalar-sources}. Define the quadratic $\Psi$ with a real
argument $a$; evaluation at the balancing coefficient $A_0$ gives
\begin{equation}
 \begin{aligned}
 \forall a\in\mathbb R,\quad
 \Psi(a)&=a(1-a)K_0+(1-a)^2v h_2-ra-r(1-a)v,\\
 K_0&=1-r/\xi+vs_1,\qquad E_1(A_0,1-A_0)=\Psi(A_0).
 \end{aligned}
 \label{eq:app-root-comparison}
\end{equation}
The symbol $K_0$ denotes only the coefficient in this polynomial;
it is distinct from the weighted bound $K$. Equation
\eqref{eq:app-root-comparison} expresses the necessary comparison
$E_1\ge0$ in a short factored form. We have not yet proved that every
possible scalar point can be moved to these weight values.

The observed interior mixing equality comes from the segment joining
$(p_\dagger,0)$ and $(p_\dagger,q_\dagger)$. On~\eqref{eq:app-cap-switch},
the two cut coefficients have the following denominators and values:
\begin{equation}
 D_c=\rho+(1-2\rho)\eta,\qquad J=\rho+2\sigma r,\qquad
 p_\dagger=\frac{\rho(1-\eta)}{D_c},\qquad
 q_\dagger=\frac{r(2-\rho)}{J}.
 \label{eq:app-root-cuts}
\end{equation}
The expression $D_c$ equals $\rho(1-\eta)+\sigma\eta$ from the
stationary-regret projection. The symbol $J$ is a denominator,
distinct from the Jensen gaps. Evaluating their lower bounds
\eqref{eq:gap-lower-bounds} at the cuts gives
\begin{equation}
 J_C=j_c(\boldsymbol\lambda_r,p_\dagger)=\frac{r(\rho+\eta)-\rho\eta}{D_c},\qquad
 J_R=j_r(\boldsymbol\lambda_r,q_\dagger)=\frac{r(3-2\rho)-\sigma}{J}.
 \label{eq:app-root-gap-values}
\end{equation}
Thus $J_C,J_R$ denote the Jensen-gap lower bounds evaluated at the two
comparison coefficients. Substituting $g=r$ and the two
endpoint coordinates into~\eqref{eq:rectangle-upper-bounds} gives
\begin{equation}
 \begin{aligned}
 R_\ell&=(1-p_\dagger)r,&
 C_\ell&=(1-p_\dagger)r+p_\dagger-J_C,\\
 R_k&=(1-p_\dagger)(1-q_\dagger)r
                    +q_\dagger(1-p_\dagger r/\rho)-J_R,\\
 C_k&=(1-p_\dagger)(1-q_\dagger)r
                    +p_\dagger(1-q_\dagger r/\sigma)-J_C.
 \end{aligned}
 \label{eq:app-root-endpoints}
\end{equation}
The subscripts $\ell,k$ denote the outer endpoint and the central
cut point, respectively; these are four scalar regret upper bounds.

We need polynomial equations with specified normalization to define
the constant reproducibly. The intersection calculation
\eqref{eq:segment-test}, applied to~\eqref{eq:app-root-endpoints},
and the comparison~\eqref{eq:app-root-comparison} define their
polynomial numerators by
\begin{equation}
 \begin{aligned}
 (C_\ell-r)(R_k-r)-(R_\ell-r)(C_k-r)
       &=\frac{N_1(r,\rho,\eta)}{\sigma D_c^2J},\\
 \Psi\!\left(\frac{r-T}{S-T}\right)
       &=\frac{N_2(r,\rho,\eta)}{D_2},\\
 D_2&=\rho\eta\sigma(1-r)
            [\rho\eta\sigma^2(1-r)(S-T)]^2.
 \end{aligned}
 \label{eq:root-polys}
\end{equation}
Every quantity in these identities uses
\eqref{eq:app-root-substitution}. On $3/10\le r\le1/3$,
$r\le\rho,\eta\le1-r$, and $S>r>T$, the factors $\rho,\eta,\sigma$,
$1-r,D_c,J$, and $S-T$ are positive. Thus both displayed
denominators are positive. Exact expansion shows that $N_1$ is
quadratic and $N_2$ is cubic in $\eta$. The identities specify these
polynomials without a large coefficient list.

The equations $N_1=N_2=0$ suggest one further equation at an extremum
of $r$. A gradient $\nabla N_i$ is the vector of partial derivatives
in coordinates $(r,\rho,\eta)$. An intersection is regular here when
the two gradients are linearly independent. At a regular constrained
local extremum, the Lagrange multiplier
condition~\cite{nocedal2006} gives real multipliers
$\nu_1,\nu_2$ such that
\[
 \nabla r=(1,0,0)=\nu_1\nabla N_1+\nu_2\nabla N_2.
\]
Taking the dot product with
$\nabla N_1\times\nabla N_2$ shows that the cross product's
$r$-coordinate is zero, giving
\begin{equation}
 N_1=N_2=0,\qquad
 H_{\mathrm{root}}=(N_1)_\rho(N_2)_\eta-(N_1)_\eta(N_2)_\rho=0.
 \label{eq:natural-root}
\end{equation}
A subscript on a polynomial denotes a partial derivative.
The symbol $H_{\mathrm{root}}$ is distinct from the weighted bound $H$.
Solving near the numerical candidate gives
\[
 (r_*,\rho_*,\eta_*)\approx
 (0.3095399654,\ 0.4799488361,\ 0.5461144474).
\]
For an exact definition, take the rational center
\begin{equation}
 \begin{aligned}
 c_r&=0.30953996543144345279052787573817601463814080254235,\\
 c_\rho&=0.47994883614366958777809825446105296705426289894925,\\
 c_\eta&=0.54611444737225673894884919253273990511184474731359,
 \end{aligned}
 \qquad \boldsymbol c=(c_r,c_\rho,c_\eta).
 \label{eq:app-root-center}
\end{equation}
These terminating decimals denote exact rationals. Define
$(r_*,\rho_*,\eta_*)$ as the unique solution
of~\eqref{eq:natural-root} in the closed box
\begin{equation}
 X=\{(r,\rho,\eta):
 |r-c_r|\le10^{-30},\
 |\rho-c_\rho|\le10^{-30},\
 |\eta-c_\eta|\le10^{-30}\}.
 \label{eq:app-root-box}
\end{equation}
The existence and uniqueness check below makes this definition
well posed and also gives
\begin{equation}
 0.30953996<r_*<0.30953997.
 \label{eq:app-root-value}
\end{equation}

To explain exactly what this local check establishes, let
$\mathcal F=(N_1,N_2,H_{\mathrm{root}})$. Its Jacobian $D\mathcal F$ is
the $3\times3$ matrix of partial derivatives, with one row for each
component and columns for $(r,\rho,\eta)$. Let $Y$ be an invertible
rational matrix approximating the inverse Jacobian at $\boldsymbol c$.
Rational interval arithmetic~\cite{moore2009} computes lower and upper rational bounds
containing all values of each entry on $X$; write the resulting
Jacobian enclosure as $D\mathcal F(X)$. The verification unit
\nolinkurl{natural-root}\footnote{Checker:
\nolinkurl{proof/bound/critical/local_root_certificate.py}.\label{fn:checker-natural-root}} checks
\begin{equation}
 \boldsymbol c-Y\mathcal F(\boldsymbol c)
   +[I-YD\mathcal F(X)](X-\boldsymbol c)
       \subset\operatorname{int}X,\qquad
 \sup_{\boldsymbol z\in X}
       \|I-YD\mathcal F(\boldsymbol z)\|_\infty<1.
 \label{eq:app-root-contraction}
\end{equation}
The first expression is the interval enclosure commonly called the
Krawczyk enclosure~\cite{krawczyk1969}. Here $I$ is the identity matrix,
$\operatorname{int}X$ is the box interior, the vector norm
$\|\boldsymbol z\|_\infty$ is the largest absolute coordinate, and
the matrix norm is the largest sum of absolute entries in a row.
The interval products enclose every possible value on $X$.
Invertibility of $Y$ makes zeros of $\mathcal F$ exactly the fixed
points of $\boldsymbol z\mapsto\boldsymbol z-Y\mathcal F(\boldsymbol z)$.
The first condition maps the box into itself, and the second makes
the map a contraction: distances are reduced by one uniform factor
strictly below one. The contraction therefore has exactly one fixed
point in $X$ by the contraction mapping
theorem~\cite{rudin1976}. This proves the local algebraic definition and the
interval~\eqref{eq:app-root-value}. Excluding feasible points elsewhere
requires the separate global argument; the numerical search
contributes no assumption to that argument.
At this isolated root, the unit \nolinkurl{scalar-root-fixture}\footnote{Checker:
\nolinkurl{proof/bound/critical/local_signs.py}.\label{fn:checker-scalar-root-fixture}}
checks the selected scalar equality point: it verifies 58 strict
signs and four active identities in the necessary relaxation.
This check concerns scalar inequalities; an algorithmic tightness
claim would additionally require a game realizing the point.

\section{Global reduction and exclusion}
\label{app:global}

\subsection{Exact sign checks and preliminary exclusions}
\label{app:preliminary}

The global proof needs polynomial signs throughout parameter regions.
For this purpose we use the Bernstein representation defined
in~\eqref{eq:bernstein-basis}, following the polynomial range
method of~\cite{garloff1986}. Its basis functions are nonnegative
and sum to one, so a polynomial $P$ with coefficients $b_0,\ldots,b_d$
satisfies
\begin{equation}
 \min_jb_j\le P(t)\le\max_jb_j\qquad(0\le t\le1).
 \label{eq:bernstein-sign-bound}
\end{equation}
In several variables, products of the one-variable bases have the
same properties. A rational affine substitution maps a rational box
to a unit box, meaning a Cartesian product of intervals $[0,1]$.
Strictly negative coefficients then prove negativity throughout
the box. If one coefficient list is insufficient, the certificate
subdivides the box. The independent checker verifies that every
subdivided box equals the union of its child boxes and that every
terminal box is excluded by a polynomial sign or incompatible
inequalities. This covers the whole region, including its boundaries.
For expressions involving the isolated root, the checker uses rational
interval enclosures~\cite{moore2009} from~\eqref{eq:app-root-box}.

It is enough to exclude solutions of $\Sigma_r$ for $r_*<r\le1/3$.
The interval~\eqref{eq:app-root-value} permits the slightly larger
threshold range below. The weight bounds in~\eqref{eq:scalar-system}
and the relative-weight definitions~\eqref{eq:relative-weights} give
the following basic domain and coefficients:
\begin{equation}
 \begin{gathered}
 r_0=\frac{30953996}{10^8},\qquad r_0\le r\le1/3,\qquad
 r\le\rho\le1-r,\qquad \sigma=1-\rho,\\
 k=\frac r{1-r},\qquad k\le u,v\le1/k,\qquad
 \chi_1=1-r-r/\rho,\quad\chi_2=1-r-r/\sigma,\quad c=1-2r.
 \end{gathered}
 \label{eq:global-basic-domain}
\end{equation}
The cutoff $r_0$ fixes the range of this proof.
Throughout this domain, $\rho,\sigma,r,1-r,u,v$ and $c$ are positive.
The two clipped lower bounds corresponding to $h_1,h_2$ in
\eqref{eq:scalar-caps} can be written as positive parts:
\begin{equation}
 \begin{aligned}
 \kappa_1&=\left[\frac{r(1+\sigma)-\sigma(1-r)v}{\rho}\right]_+=1-h_1,\\
 \kappa_2&=\left[\frac{r(1+\rho)-\rho(1-r)u}{\sigma}\right]_+=1-h_2.
 \end{aligned}
 \label{eq:cap-payoff-lower-bounds}
\end{equation}
Here $[t]_+=\max\{0,t\}$. Taking the positive part combines a
derived payoff lower bound with the inherent bound that payoffs
are nonnegative. It is equivalent to taking the minimum with one
in the definitions of $h_1,h_2$.

Two elementary coefficient bounds that will be used below are
\begin{equation}
 \rho\sigma\le1/4,\qquad
 \chi_1-r,\ \chi_2-r
   \le\frac{1-4r+2r^2}{1-r}<0.
 \label{eq:global-coefficient-signs}
\end{equation}
The first follows from $(\rho-\sigma)^2\ge0$ and $\rho+\sigma=1$.
For the second, use $\rho,\sigma\le1-r$ in their definitions.
Its numerator is negative at $r=3/10$ and has derivative $4r-4<0$
on the threshold interval.

For the odds coordinates $u,v$, the cut denominators have the
following equivalent expressions:
\begin{equation}
 D_{c,\mathrm{odds}}=\rho u+\sigma=D_c/\eta,\qquad
 D_{r,\mathrm{odds}}=\sigma v+\rho=D_r/\xi.
 \label{eq:odds-cut-denominators}
\end{equation}
Thus $p_\dagger=\rho u/D_{c,\mathrm{odds}}$ and
$q_\dagger=\sigma v/D_{r,\mathrm{odds}}$ are the same cuts as
in~\eqref{eq:app-mixing-cuts}.

Before eliminating weights, we specify the two segment inequalities
that later paths will use. Substituting the cut coefficients into
\eqref{eq:rectangle-upper-bounds} with $g=r$ gives the positive
denominator and the central numerators
\begin{equation}
 \begin{aligned}
 D_{\mathrm{mix}}&=D_{c,\mathrm{odds}}D_{r,\mathrm{odds}}>0,\\
 N_R&=\sigma^2\chi_2+\rho(\sigma-2r)u+\sigma^2cv
                         +\rho\sigma(\chi_1-r)uv,\\
 N_C&=\rho^2\chi_1+\sigma(\rho-2r)v+\rho^2cu
                         +\rho\sigma(\chi_2-r)uv,\\
 R_r(p_\dagger,q_\dagger)-r&=N_R/D_{\mathrm{mix}},\qquad
 C_r(p_\dagger,q_\dagger)-r=N_C/D_{\mathrm{mix}}.
 \end{aligned}
 \label{eq:mixing-numerators}
\end{equation}
Thus the signs of $N_R,N_C$ specify which central regret upper
bound is above the comparison threshold. Define
\begin{equation}
 L=(\chi_1+cu)N_R+ruN_C,\qquad
 B_{\mathrm{mix}}=(\chi_2+cv)N_C+rvN_R.
 \label{eq:mixing-determinants}
\end{equation}
These are positive multiples of the intersection determinants for
the segments toward $(p_\dagger,0)$ and $(0,q_\dagger)$,
respectively. The other endpoint differences are
\begin{equation}
 \begin{aligned}
 R_r(p_\dagger,0)-r&=-\frac{\rho ur}{\rho u+\sigma},&
 C_r(p_\dagger,0)-r&=\frac{\rho(\chi_1+cu)}{\rho u+\sigma},\\
 R_r(0,q_\dagger)-r&=\frac{\sigma(\chi_2+cv)}{\sigma v+\rho},&
 C_r(0,q_\dagger)-r&=-\frac{\sigma vr}{\sigma v+\rho}.
 \end{aligned}
 \label{eq:mixing-outer-endpoints}
\end{equation}
The segment predicates in~\eqref{eq:scalar-system} force
$\chi_1+cu,\chi_2+cv\ge0$: the other bound at each outer endpoint
is strictly below $r$. If $N_C<0$, the central comparison forces
$N_R\ge0$, and the intersection condition~\eqref{eq:segment-test} gives
$L\ge0$. If $N_R,N_C\ge0$, the definition of $L$ and the endpoint
coefficient signs give $L\ge0$ directly. Both numerators negative
would violate the central comparison. The remaining case
$N_R<0\le N_C$ gives $B_{\mathrm{mix}}\ge0$ from the second segment.
Consequently, the initial necessary system implies the explicit
guarded alternatives
\begin{equation}
 \Sigma_r\ \Longrightarrow\
 \begin{gathered}
 \chi_1+cu\ge0,\qquad\chi_2+cv\ge0,\\
 (N_R\ge0\ \land\ L\ge0)
 \quad\text{or}\quad
 (N_R<0\le N_C\ \land\ B_{\mathrm{mix}}\ge0).
 \end{gathered}
 \label{eq:mixing-branch-guards}
\end{equation}
These substitutions and determinant identities are reconstructed by
the units \nolinkurl{compact-identities}\footnote{Checker:
\nolinkurl{proof/bound/simplify/verify_compact_identities.py}.\label{fn:checker-compact-identities}} and
\nolinkurl{mass-identities}\footnote{Checker:
\nolinkurl{proof/bound/global/check_reduction_identities.py}.\label{fn:checker-mass-identities}}.

We next show that exactly one of $S,T$ is above $r$. The inequalities
$H,K\ge r$ in~\eqref{eq:scalar-sources} imply $\max\{S,T\}\ge r$,
because $H,K$ are convex combinations of $S,T$. A preliminary
uniform bound needed for this exclusion is
\begin{equation}
 s_1<r,\qquad s_2<r
 \qquad\text{on the full domain~\eqref{eq:global-basic-domain}}.
 \label{eq:small-payoff-bounds}
\end{equation}
Indeed, $s_2=\chi_2+\chi_1v$ in~\eqref{eq:scalar-caps} is affine
in $v$, so its maximum on $[k,1/k]$ occurs at an endpoint.
Cauchy--Schwarz, applied to the two terms with denominators
$\sigma,\rho$, gives
\begin{equation}
 \frac1\sigma+\frac v\rho\ge(1+\sqrt v)^2.
 \label{eq:reciprocal-weight-bound}
\end{equation}
At $v=k$, substitute~\eqref{eq:reciprocal-weight-bound} into $s_2$
and use $r\ge3/10$, $k\ge3/7>(13/20)^2$ to obtain
\[
 s_2\le1-k-2r\sqrt{k}
       \le1-3/7-39/100=127/700<r.
\]
At $v=1/k$, the same substitution gives
$s_2\le1/r-2-2\sqrt{r(1-r)}<r$. To verify the last inequality,
both sides of $1/r-2-r<2\sqrt{r(1-r)}$ are positive on this interval;
squaring and multiplying by $r^2$ gives
$5r^4+2r^2-4r+1<0$. This polynomial is negative at $r_0$, and its
derivative $20r^3+4r-4$ is at most $-52/27$ for $r\le1/3$.
The proof for $s_1$ exchanges $\rho,\sigma$ and $u,v$. In particular,
\eqref{eq:small-payoff-bounds} does not assume $\rho\le1/2$.

To exclude $S,T\ge r$ simultaneously, first exchange all player
labels, if necessary, so that $\rho\le1/2$. This temporary labeling
does not assume an order between $S,T$. Under $S\ge r$, we will
prove $u<99/100$ before using a separate implication valid over
the full range of $\rho$. Define the uncapped expression for $h_1$
and the affine expression for $S$ obtained from it by
\begin{equation}
 \begin{aligned}
 u_c&=99/100,&
 \bar h_1(v)&=1-r(1+\sigma)/\rho+\sigma(1-r)v/\rho,\\
 S_0(v)&=\bar h_1(v)-r,& S_u(v)&=s_2-r<0,\\
 S^{\mathrm{raw}}(u,v)&=S_0(v)+uS_u(v).
 \end{aligned}
 \label{eq:preliminary-uncapped-source}
\end{equation}
The superscript denotes replacement of $h_1$ by its uncapped
formula. Since $h_1\le\bar h_1$, equation~\eqref{eq:scalar-caps}
gives $S\le S^{\mathrm{raw}}$. Its slope in $v$ at $u_c$ is positive:
\begin{equation}
 \begin{aligned}
 A_v&=\frac{\sigma(1-r)}{\rho}+u_c\chi_1
     =\frac{(1-u_c)\sigma(1-r)+u_cc}{\rho}>0,\\
 v_s&=\frac{r-S_0(0)-u_cS_u(0)}{A_v},\qquad
 S^{\mathrm{raw}}(u_c,v_s)=r.
 \end{aligned}
 \label{eq:preliminary-source-boundary}
\end{equation}
Thus $v_s$ is the value at which this uncapped expression equals
$r$. If $u\ge u_c$ and $S\ge r$, the negative $u$-slope
in~\eqref{eq:preliminary-uncapped-source} gives
\begin{equation}
 S^{\mathrm{raw}}(u_c,v)\ge S^{\mathrm{raw}}(u,v)\ge S\ge r,
 \qquad v\ge v_s.
 \label{eq:preliminary-source-order}
\end{equation}

We must obtain the left-segment condition with its original guards
before contradicting it. Write the central column numerator as
\begin{equation}
 \begin{aligned}
 N_C&=C_{\mathrm b}(v)+uC_{\mathrm s}(v),\\
 C_{\mathrm b}(v)&=\rho^2\chi_1+\sigma(\rho-2r)v,\qquad
 C_{\mathrm s}(v)=\rho^2c+\rho\sigma(\chi_2-r)v,\\
 (N_C)_v&=\sigma(\rho-2r)+\rho\sigma(\chi_2-r)u<0.
 \end{aligned}
 \label{eq:preliminary-column-affine}
\end{equation}
The subscripts $\mathrm b,\mathrm s$ distinguish the constant and
$u$-slope of the affine expression. The derivative sign follows
from $\rho\le1/2<2r$, $u>0$, and
\eqref{eq:global-coefficient-signs}. If $C_{\mathrm s}(v)\le0$,
equations~\eqref{eq:preliminary-source-order}
and~\eqref{eq:preliminary-column-affine} imply
$N_C(u,v)\le N_C(u_c,v)\le N_C(u_c,v_s)$.

If $C_{\mathrm s}(v)>0$, use $S\ge r$ to bound $u$ before
substitution into $N_C$. On $v\le v_0$, $h_1=\bar h_1$ and
$u\le(S_0-r)/(-S_u)$; on $v\ge v_0$, $h_1=1$ and
$u\le c/(-S_u)$. The denominator is positive by
\eqref{eq:small-payoff-bounds}. Consequently
\begin{equation}
 \begin{aligned}
 P_{\mathrm{raw}}&=(-S_u)C_{\mathrm b}+C_{\mathrm s}(S_0-r),
 &N_C&\le P_{\mathrm{raw}}/(-S_u)&& (v\le v_0),\\
 P_{\mathrm{clip}}&=(-S_u)C_{\mathrm b}+cC_{\mathrm s},
 &N_C&\le P_{\mathrm{clip}}/(-S_u)&& (v\ge v_0).
 \end{aligned}
 \label{eq:preliminary-column-substitution}
\end{equation}
These inequalities follow by replacing $u$ in
\eqref{eq:preliminary-column-affine} by its respective upper bound.
Both numerator polynomials have degree at most two in $v$.

To state their sign checks exactly, for a quadratic polynomial $P$
on an interval $[a,b]$ with $a\le b$, define
\begin{equation}
 \begin{aligned}
 \beta_0(P;a,b)&=P(a),\\
 \beta_1(P;a,b)&=P(a)+(b-a)P'(a)/2,\\
 \beta_2(P;a,b)&=P(b),\\
 P(a+(b-a)t)&=\beta_0(P;a,b)(1-t)^2
          +2\beta_1(P;a,b)t(1-t)+\beta_2(P;a,b)t^2.
 \end{aligned}
 \label{eq:quadratic-bernstein}
\end{equation}
For $0\le t\le1$, negative coefficients imply a negative polynomial
by~\eqref{eq:bernstein-sign-bound}. Unit
\nolinkurl{canonical-u-cap}\footnote{Checker:
\nolinkurl{proof/bound/simplify/replay-relaxed-source-knot-u-cap.py}.\label{fn:checker-canonical-u-cap}} reconstructs the following seven sign checks:
\begin{equation}
 \begin{gathered}
 N_C(u_c,v_s)<0,\\
 \beta_j(P_{\mathrm{raw}};v_s,v_0)<0\quad(j=0,1,2;\ v_s\le v_0),\\
 \beta_j(P_{\mathrm{clip}};v_0,1/k)<0\quad(j=0,1,2).
 \end{gathered}
 \label{eq:preliminary-column-signs}
\end{equation}
Their exact parameter domain is $619/2000\le r\le1/3$,
$r\le\rho\le1/2$, covered by
\begin{equation}
 r=619/2000+(1/3-619/2000)t_r,\qquad
 \rho=r+(1/2-r)t_\rho,\qquad 0\le t_r,t_\rho\le1.
 \label{eq:preliminary-sign-domain}
\end{equation}
It covers the required domain because $619/2000<r_0$.
The denominators cleared in these checks are positive products
of $\rho,\sigma,r,1-r$ and the factor
$J_{\mathrm{CN}}=100\rho A_v>0$. There are 352 strictly negative
rational Bernstein coefficients. The interval $[v_0,1/k]$ is
nonempty on this domain. If $v_s>v_0$, no uncapped-$h_1$ point can
satisfy $v\ge v_s$, so that case needs no first interval check.
Equations~\eqref{eq:preliminary-column-substitution}
and~\eqref{eq:preliminary-column-signs}, together with the
$C_{\mathrm s}\le0$ case, prove $N_C<0$. The original segment
alternative~\eqref{eq:mixing-branch-guards} now gives $L\ge0$.

To contradict that inequality, expand $L$ about $(u_c,v_s)$:
\begin{equation}
 \begin{aligned}
 L(u_c+d_u,v_s+d_v)
  &=b_{00}+b_{10}d_u+b_{01}d_v+b_{20}d_u^2
                +b_{11}d_ud_v+b_{21}d_u^2d_v,\\
 d_u&=u-u_c\ge0,\qquad d_v=v-v_s\ge0.
 \end{aligned}
 \label{eq:preliminary-shifted-determinant}
\end{equation}
Here $b_{ij}$ is the coefficient of $d_u^id_v^j$, equivalently the
corresponding derivative of $L$ at $(u_c,v_s)$ divided by $i!j!$.
The same verification unit reconstructs these six coefficients
and proves each strictly negative on
\eqref{eq:preliminary-sign-domain}. Every monomial is nonnegative
and $b_{00}<0$, so~\eqref{eq:preliminary-shifted-determinant}
gives $L<0$. We have proved
\begin{equation}
 \Sigma_r,\quad r_0\le r\le1/3,\quad\rho\le1/2,\quad S\ge r
 \quad\Longrightarrow\quad u<99/100.
 \label{eq:preliminary-u-cap}
\end{equation}

The next implication supplies the contradiction if $T\ge r$ also.
It uses the explicit capped formulas~\eqref{eq:scalar-caps} and
holds on the full stationary-weight range:
\begin{equation}
 \begin{gathered}
 r_0\le r\le1/3,\quad r\le\rho\le1-r,\quad
 \sigma=1-\rho,\quad 0\le u\le1,\quad v\ge0,\\
 S\ge r\quad\Longrightarrow\quad T<r.
 \end{gathered}
 \label{eq:full-cap-exclusion}
\end{equation}
In particular, it requires no second-round weight comparison.
Define the other uncapped expression and replace both capped
formulas in $S,T$ by their uncapped expressions:
\begin{equation}
 \begin{aligned}
 \bar h_2(u)&=1-r(1+\rho)/\sigma+\rho(1-r)u/\sigma,\\
 \bar S(u,v)&=\bar h_1(v)-r+u(s_2-r),\\
 \bar T(u,v)&=\bar h_2(u)-r+v(s_1-r).
 \end{aligned}
 \label{eq:uncapped-common-bounds}
\end{equation}
Thus $S\le\bar S$ and $T\le\bar T$. Subscripts $u,v$ on these
functions denote partial derivatives. Their slopes in $v$ have
the endpoint values
\begin{equation}
 \begin{aligned}
 \left.\bar S_v\right|_{u=0}&=\sigma(1-r)/\rho>0,&
 \left.\bar S_v\right|_{u=1}&=c/\rho>0,\\
 \left.\bar T_v\right|_{u=0}&=\chi_1-r<0,&
 \left.\bar T_v\right|_{u=1}&=\chi_1+\chi_2-r<0.
 \end{aligned}
 \label{eq:raw-cap-slopes}
\end{equation}
The slopes are affine in $u$, so these signs hold for $0\le u\le1$.
The first negative sign uses~\eqref{eq:global-coefficient-signs};
the second uses $\rho\sigma\le1/4$ to give
$\chi_1+\chi_2-r=2-3r-r/(\rho\sigma)\le2-7r<0$.
These arguments do not need bounds on $v$.

Because $\bar S_v>0$, the equation $\bar S(u,v)=r$ has the unique
solution
\begin{equation}
 v_{\mathrm b}(u)=\frac{r-\bar S(u,0)}{\bar S_v(u)}.
 \label{eq:raw-cap-boundary}
\end{equation}
Its numerator and denominator are affine in $u$, so the derivative
of their ratio has constant sign. To control $\bar T$ along this
boundary, the unit \nolinkurl{full-mass-both-high}\footnote{Checker:
\nolinkurl{proof/bound/simplify/replay-fullmass-both-high.py}.\label{fn:checker-full-mass-both-high}} checks
\begin{equation}
 \begin{gathered}
 0<v_{\mathrm b}(0),v_{\mathrm b}(1)<7/3,\qquad
 \bar T(1,v_{\mathrm b}(1))<r,\\
 \mathcal J_{\mathrm{cap}}
       =\bar T_u\bar S_v-\bar T_v\bar S_u>0
       \quad(0\le u\le1,\ 0\le v\le7/3).
 \end{gathered}
 \label{eq:cap-boundary-signs}
\end{equation}
The determinant $\mathcal J_{\mathrm{cap}}$ combines the two rows
of partial derivatives of $\bar T,\bar S$. Its positive denominator
is $\rho^2\sigma^2$; the endpoint expressions also have positive
rational denominators. The exact coefficient-check domain is
$r\in[r_0,1/3]$, $\rho\in[r,1-r]$, $u\in[0,1]$,
$v\in[0,7/3]$. Since $v_{\mathrm b}$ stays between its endpoint
values, this domain contains the whole boundary curve.
If $S\ge r$, then $\bar S\ge r$ and
\eqref{eq:raw-cap-boundary} imply $v\ge v_{\mathrm b}(u)$.
Using the negative slope~\eqref{eq:raw-cap-slopes} and differentiating
$\bar S(u,v_{\mathrm b}(u))=r$ gives
\begin{equation}
 \begin{aligned}
 T&\le\bar T(u,v)\le\bar T(u,v_{\mathrm b}(u)),\\
 \frac{d}{du}\bar T(u,v_{\mathrm b}(u))
   &=\frac{\mathcal J_{\mathrm{cap}}}{\bar S_v}>0.
 \end{aligned}
 \label{eq:cap-boundary-comparison}
\end{equation}
Hence $T\le\bar T(1,v_{\mathrm b}(1))<r$, proving
\eqref{eq:full-cap-exclusion}, even for $v>7/3$.

If $S,T\ge r$ in the temporary labeling with $\rho\le1/2$,
equation~\eqref{eq:preliminary-u-cap} makes
\eqref{eq:full-cap-exclusion} applicable and contradicts $T\ge r$.
The impossibility of both inequalities is invariant under exchanging
the players. Together with $\max\{S,T\}\ge r$, this proves
\begin{equation}
 \Sigma_r,\quad r_0\le r\le1/3
 \quad\Longrightarrow\quad
 (S>r>T)\ \text{or}\ (T>r>S).
 \label{eq:strict-bound-orientation}
\end{equation}
For strictness, $S=r>T$ would give
$H=\alpha r+(1-\alpha)T<r$, because $\alpha\le1-r<1$;
the case $T=r>S$ similarly contradicts $K\ge r$.

\subsection{Eliminating the second-round weights}
\label{app:weights}

The purpose of this reduction is to eliminate $\alpha,\gamma$
while preserving enough necessary inequalities for the later paths.
By~\eqref{eq:strict-bound-orientation}, exchange all player labels,
if necessary, to obtain $S>r>T$. This exchange acts simultaneously:
\begin{equation}
 (\rho,\sigma,\eta,\xi,\alpha,\gamma,u,v,s_1,s_2,h_1,h_2,S,T)
 \longleftrightarrow
 (\sigma,\rho,\xi,\eta,\gamma,\alpha,v,u,s_2,s_1,h_2,h_1,T,S).
 \label{eq:weight-exchange}
\end{equation}
The stationary row weight after this exchange can exceed $1/2$.
We therefore use the full domain~\eqref{eq:global-basic-domain}
throughout.

Define the two distances from $r$ and their ratio by
\begin{equation}
 m=S-r>0,\qquad n=r-T>0,\qquad
 z=\frac nm>0,\qquad A_0=\frac{n}{m+n}=\frac{r-T}{S-T}.
 \label{eq:weight-balancing-values}
\end{equation}
This is the same coefficient $A_0$ used in
\eqref{eq:app-observed-weights}; it balances $S,T$ at $r$.
Substitute~\eqref{eq:weight-balancing-values} into $H,K$
in~\eqref{eq:scalar-sources}. The results are
\begin{equation}
 H-r=\alpha(m+n)-n,\qquad K-r=m-\gamma(m+n),
 \qquad
 \alpha\ge A_0,\quad\gamma\le1-A_0,\quad0<A_0\le1-r.
 \label{eq:weight-source-restrictions}
\end{equation}
The two weight inequalities use $H,K\ge r$ and $m+n>0$;
the upper bound on $A_0$ also uses $\alpha\le1-r$
from~\eqref{eq:scalar-system}. They do not yet justify changing
the weights, because $E_1\ge0$ must survive.

To lower $\alpha$ while preserving $E_1\ge0$, we need its
$\alpha$-derivative to be negative throughout the allowed interval
for $\gamma$. Two signs will imply this:
\begin{equation}
 M=(1-r)S-(1-2r)T-rh_2/\xi<0,\qquad h_2>\xi.
 \label{eq:weight-derivative-signs}
\end{equation}
Their hypotheses differ. The sign $M<0$ needs
\eqref{eq:global-basic-domain} and $T\ge0$. The sign $h_2>\xi$
also uses the initial nonnegativity
$s_1,s_2,h_1,h_2,S,T\ge0$, $S\ge r$, and the original
$H\ge r$ with $\alpha\le1-r$. These are conditions
of~\eqref{eq:scalar-system} and~\eqref{eq:scalar-sources}, before
any later path drops them. We explain the checks before using
the two signs to move the weights.

First suppose $h_2=\bar h_2$, using its uncapped formula
in~\eqref{eq:uncapped-common-bounds}. For either formula of $h_1$,
the functions $T,M$ are affine in $u$. Write
\begin{equation}
 \begin{aligned}
 T&=t_0+t_1u,\qquad M=M_0+M_1u,\\
 t_0&=\bar h_2(0)-r+v(\chi_1-r),\qquad
 t_1=\rho(1-r)/\sigma+v\chi_2>0,\\
 M_1&=(1-r)(s_2-r)-(1-2r)t_1
                        -r(1+v)\rho(1-r)/\sigma<0.
 \end{aligned}
 \label{eq:weight-M-affine}
\end{equation}
Here $M_0$ is the constant coefficient obtained by formally
evaluating this affine formula at $u=0$. To justify $t_1>0$,
if $\chi_2\ge0$ both terms are nonnegative and the first is
positive. If $\chi_2<0$, $v\le1/k$ gives
$t_1\ge\rho(1-r)/\sigma+\chi_2/k=(1-r)c/r>0$.
The sign of $M_1$ then follows from
\eqref{eq:small-payoff-bounds}, $c>0$, and
$1/\xi=1+v$ in~\eqref{eq:relative-weights}.
Using $T\ge0$ in~\eqref{eq:weight-M-affine} gives the exact elimination
\begin{equation}
 M=M_0-\frac{M_1t_0}{t_1}+\frac{M_1}{t_1}T
       \le M_0-\frac{M_1t_0}{t_1}.
 \label{eq:weight-M-elimination}
\end{equation}
It remains to check the sign of the right-hand side after clearing
its positive denominators. The unit \nolinkurl{mass-M-small}\footnote{Checker:
\nolinkurl{proof/bound/global/replay_m_chart_certificate.py}.\label{fn:checker-mass-M-small}} checks
$619/2000\le r\le1/3$, $r\le\rho\le1/2$; the unit
\nolinkurl{mass-M}\footnote{Checker:
\nolinkurl{proof/bound/global/replay_m_chart_full_mass_certificate.py}.\label{fn:checker-mass-M}} adds $r_0\le r\le1/3$,
$1/2\le\rho\le1-r$. In both cases the $h_1$ intervals are
$k\le v\le v_0$ and $v_0\le v\le1/k$. They are in this order
because $v_0\le1/k$ is equivalent to $r^2\le\sigma c$,
which follows from $\sigma\ge r$ and $c\ge r$.
The smaller-$\rho$ checks have strictly negative Bernstein
coefficients. In each of the two larger-$\rho$ checks, three terms
are grouped into a strictly negative quadratic; all other
coefficients are negative. The grouped coefficients
$b_0,b_1,b_2$ satisfy $b_0,b_2<0$ and
$60b_0b_2-36b_1^2>0$, proving negativity of
$b_0(1-t)^2+6b_1t(1-t)+15b_2t^2$.
Thus~\eqref{eq:weight-M-elimination} proves $M<0$
on the uncapped-$h_2$ branch.

On the other branch $h_2=1$, the definition of the cap implies
$u\ge u_C$, where
\begin{equation}
 u_C=\frac{r(1+\rho)}{\rho(1-r)},\qquad
 M\le(1-r)[h_1-r+u_C(s_2-r)]-r(1+v).
 \label{eq:weight-M-capped}
\end{equation}
Indeed, $S=h_1-r+u(s_2-r)$ decreases with $u$
by~\eqref{eq:small-payoff-bounds}, and the term
$-(1-2r)T$ is nonpositive because $T\ge0$.
The right-hand side of~\eqref{eq:weight-M-capped} is affine in $v$
on each formula of $h_1$. The unit \nolinkurl{mass-M}\textsuperscript{\ref{fn:checker-mass-M}} checks its
three endpoint values at $v=k,v_0,1/k$ are negative on
$r_0\le r\le1/3$, $r\le\rho\le1-r$. This completes $M<0$
on both $h_2$ branches.

For the other derivative sign, the unit \nolinkurl{mass-h_2-xi}\footnote{Checker:
\nolinkurl{proof/bound/simplify/replay-h_2-xi-certificates.py}.\label{fn:checker-mass-h_2-xi}}
excludes $h_2\le\xi$ under the hypotheses stated for
\eqref{eq:weight-derivative-signs}. Its sufficient contradiction
is the following conditional sign:
\begin{equation}
 \begin{gathered}
 s_1,s_2,h_1,h_2,S,T\ge0,\quad S\ge r,\quad h_2\le\xi\\
 \Longrightarrow\quad
 Z_\xi=(1-r)S+(1+r)T-r<0,
 \end{gathered}
 \label{eq:weight-h2-xi-exclusion}
\end{equation}
on~\eqref{eq:global-basic-domain}. To specify the substitutions
in this check, put $x=\xi$ and regard $h_2$ temporarily as a
coordinate. Since $h_2\le x\le1-r<1$, its uncapped formula applies:
\begin{equation}
 \begin{aligned}
 u&=\frac{\sigma h_2-\sigma+r(1+\rho)}{\rho(1-r)},\qquad
 v=\frac{1-x}{x},\\
 s_1&=c_0+c_1h_2,\qquad
 c_0=\chi_1-\frac{\chi_2[\sigma-r(1+\rho)]}{\rho(1-r)},\quad
 c_1=\frac{\sigma(1-r)-r}{\rho(1-r)}.
 \end{aligned}
 \label{eq:weight-h2-xi-substitution}
\end{equation}
The weight ranges give $c_1<1$, $c_0<r$, and
$\partial_{h_2}T=1+vc_1>0$. For $c_1<1$, use
$(1-2\rho)(1-r)<r$, a consequence of
\eqref{eq:global-coefficient-signs}. For $\rho\le1/2$,
$\chi_2\ge0$, $\sigma-r(1+\rho)\ge0$, and
$\chi_1\le1-3r<r$ give $c_0<r$. For $\rho\ge1/2$,
$\chi_1\le19/70$. If the product subtracted in $c_0$ is negative,
then $0<\chi_2\le1/10$,
$0<r(1+\rho)-\sigma<1/20$, and
$\rho(1-r)\ge1/3$; hence
$c_0\le19/70+3/200=401/1400<3/10\le r$.
If that product is nonnegative, $c_0\le\chi_1<r$ directly.
Finally, $c_1\ge0$ when $\rho\le1/2$; otherwise
$c_1\ge-2r^2/(1-r)$ and $v\le1/k$ give $1+vc_1\ge1-2r>0$.

Increasing $h_2$ to $x$ therefore gives the upper bound
\begin{equation}
 T\le\widehat T(x)
     =x-r/x+(1-x)(c_0+c_1x)/x,\qquad
 \widehat T'(x)=1-c_1+(r-c_0)/x^2>0.
 \label{eq:weight-h2-xi-T}
\end{equation}
If $h_1=1$, the switch formula~\eqref{eq:app-cap-switch} gives
$x\le\theta=\sigma(1-r)/(\sigma+r)$. The first two sign checks
in \nolinkurl{mass-h_2-xi}\textsuperscript{\ref{fn:checker-mass-h_2-xi}} certify $\widehat T(\theta)<0$,
contradicting $T\ge0$. Their rational boxes are
$r\in[3/10,1/3]$, $\rho\in[3/10,1/2]$ and
$r\in[3/10,1/3]$, $\rho\in[1/2,7/10]$.
The numerators have positive Bernstein coefficients and their
denominators are negative.

Thus $h_1$ is uncapped. Equations~\eqref{eq:scalar-caps}
and~\eqref{eq:small-payoff-bounds} give
$S<h_1-r$, so $S\ge r$ implies $h_1>2r$. Solving these two
strict bounds on the uncapped formula for $h_1$ gives
\begin{equation}
 \theta<x<\theta_2,\qquad
 \theta=\frac{\sigma(1-r)}{\sigma+r},\qquad
 \theta_2=\frac{\sigma(1-r)}{1+r-2\rho(1-r)}.
 \label{eq:weight-h2-xi-range}
\end{equation}
The denominator of $\theta_2$ is positive on the stated domain.
After~\eqref{eq:weight-h2-xi-substitution}, differentiation at
fixed $x$ gives
\[
 \partial_{h_2}Z_\xi
       =(\sigma/\rho)(s_2-r)+(1+r)(1+vc_1)>0.
\]
For $\rho\le1/2$, $s_2\ge0$ and $c_1\ge0$ give a lower bound
$1+r-\sigma r/\rho\ge2r>0$. For $\rho\ge1/2$, use
$\sigma/\rho\le1$, $s_2\ge0$, and $vc_1\ge-2r$ to get
$1-2r-2r^2\ge1/9>0$. Thus substituting $h_2=x$ increases
$Z_\xi$. The remaining two sign checks substitute
$x=\theta+(\theta_2-\theta)t$, $0\le t\le1$, and prove that even
this increased expression is negative. Their exact boxes are
$r\in[4/13,1/3]$, $\rho\in[3/10,1/2]$, $t\in[0,1]$ and
$r\in[619/2000,1/3]$, $\rho\in[1/2,7/10]$, $t\in[0,1]$.
After clearing positive denominators, the numerators of $-Z_\xi$
have positive Bernstein coefficients. These four boxes cover all
weights needed for~\eqref{eq:weight-h2-xi-exclusion}.

In our ordered case $S>r>T$, the original $H\ge r$ and
$\alpha\le1-r$ imply
$r\le H\le(1-r)S+rT$. Together with $T\ge0$, this gives
$Z_\xi\ge T\ge0$, contradicting
\eqref{eq:weight-h2-xi-exclusion}. Hence $h_2>\xi$, completing
\eqref{eq:weight-derivative-signs}.

Now differentiate the particular expression $E_1$ in
\eqref{eq:scalar-sources}. Using $1/\xi=1+v$ gives
\begin{equation}
 \begin{aligned}
 \partial_\alpha E_1
   &=(1-\gamma)(S-T)+\gamma(1+T-h_2/\xi)-r,\\
 \left.\partial_\alpha E_1\right|_{\gamma=r}&=M<0,\\
 \left.\partial_\alpha E_1\right|_{\gamma=1-A_0}
   &=\frac{(S-r)(1-h_2/\xi)+T(T-r)}{S-T}<0.
 \end{aligned}
 \label{eq:weight-E1-derivative}
\end{equation}
The second endpoint is negative because $S-r>0$, $S-T>0$,
$0\le T<r$, and $h_2>\xi$. The derivative is affine in $\gamma$,
so it is negative throughout $r\le\gamma\le1-A_0$.
Lowering $\alpha$ to $A_0$ therefore cannot decrease $E_1$.
It cannot decrease $U=\alpha s_1+(1-\alpha)h_2$ either:
$s_1<r$ by~\eqref{eq:small-payoff-bounds}, and
$h_2>\xi\ge r$ by~\eqref{eq:weight-derivative-signs}.

After this move $E_1$ is affine in $\gamma$, with
\begin{equation}
 E_1(A_0,0)=-A_0r<0.
 \label{eq:weight-E1-zero}
\end{equation}
Its nonnegativity at the present $\gamma>0$ therefore forces its
$\gamma$-slope to be positive. Raising $\gamma$ to $1-A_0$ preserves
$E_1\ge0$ and does not affect $U$. We have justified
$\alpha=A_0$, $\gamma=1-A_0$ by necessary comparisons, without
assuming the observed numerical equalities.

The inequalities needed next are $U(A_0)\ge r$, $A_0\le1-r$,
and $E_1(A_0,1-A_0)\ge0$. Their positive rescalings are
\begin{equation}
 \begin{aligned}
 G&=(m+n)[U(A_0)-r]
       =m(h_2-r)-n(r-s_1)\ge0,\\
 B_{\mathrm{mass}}&=(m+n)[(1-r)-A_0]
       =(1-r)m-rn\ge0,\\
 E&=\frac{E_1(A_0,1-A_0)}{(1-A_0)^2}
       =v(h_2-r)+z[c+v(s_1-2r)]-rz^2\ge0.
 \end{aligned}
 \label{eq:relaxed}
\end{equation}
Thus $B_{\mathrm{mass}}$ is exactly the upper coefficient bound
$A_0\le1-r$, while $B_{\mathrm{mix}}$ is the distinct segment
determinant in~\eqref{eq:mixing-determinants}. At the initial point,
changing only $\alpha,\gamma$ leaves all caps and segment guards
unchanged. For the subsequent paths we use the following necessary
domain, with the functions still defined by~\eqref{eq:scalar-caps}:
\begin{equation}
 \begin{gathered}
 r_0\le r\le1/3,\quad r\le\rho\le1-r,\quad
 k\le u,v\le1/k,\quad \sigma=1-\rho,\quad k=r/(1-r),\\
 m=S-r>0,\quad n=r-T>0,\quad T\ge0,\quad s_2\ge0,\quad
 (\rho\le1/2\ \Longrightarrow\ s_1\ge0),\\
 G\ge0,\quad B_{\mathrm{mass}}\ge0,\quad E\ge0,\\
 (N_R\ge0\ \land\ L\ge0)
 \quad\text{or}\quad
 (N_R<0\le N_C\ \land\ B_{\mathrm{mix}}\ge0).
 \end{gathered}
 \label{eq:reduced-domain}
\end{equation}
The original point also satisfies $s_1\ge0$ for $\rho\ge1/2$,
but the paths on that half use the weaker system without it.
The guarded alternatives specify how a segment inequality is chosen
at the start; each later path will state which guards it needs
throughout the move. We have dropped $E_2$ and the remaining
weighted comparisons. We have also dropped the lower bound
$A_0\ge r$, which has not been proved for this reduction.
The second-round weights have therefore disappeared, leaving only
$(r,\rho,u,v)$.

We next prove $u<1$, so that $h_2$ always uses its uncapped formula
in this reduced problem. For $\rho\le1/2$, assume $u\ge u_c$.
The column-numerator checks
\eqref{eq:preliminary-column-signs}, using $S>r$, force $N_C<0$.
The opposite guarded alternative in~\eqref{eq:reduced-domain}
is then impossible, and the first gives $L\ge0$. The negative
coefficients~\eqref{eq:preliminary-shifted-determinant} contradict
this inequality. Thus $u<u_c<1$ on this half.

For $\rho\ge1/2$, replace both capped expressions by
\eqref{eq:uncapped-common-bounds} and define
\begin{equation}
 \begin{aligned}
 B_{\mathrm{mass}}^{\mathrm{raw}}
    &=(1-r)\bar S+r\bar T-r
      =\beta_0^{\mathrm{mass}}+\beta_1^{\mathrm{mass}}v,\\
 L&=\ell_0+\ell_1v.
 \end{aligned}
 \label{eq:weight-u-barrier-coefficients}
\end{equation}
The four coefficients are the constant and linear coefficients in
$v$, and depend on $(r,\rho,u)$. Since $S\le\bar S$ and $T\le\bar T$,
$B_{\mathrm{mass}}\le B_{\mathrm{mass}}^{\mathrm{raw}}$.
The unit \nolinkurl{fixed-v-and-u1}\footnote{Checker:
\nolinkurl{proof/bound/global/replay_preclip_left_projection_signs.py}.\label{fn:checker-fixed-v-and-u1}} certifies, on
$619/2000\le r\le1/3$, $1/2\le\rho\le1-r$, $1\le u\le1/k$,
\begin{equation}
 \beta_1^{\mathrm{mass}}>0,\qquad \ell_1<0,\qquad
 \Omega=\beta_0^{\mathrm{mass}}\ell_1
                  -\beta_1^{\mathrm{mass}}\ell_0>0.
 \label{eq:weight-left-u-barrier}
\end{equation}
If $L\ge0$, its affine formula gives $v\le-\ell_0/\ell_1$.
Substitution into the increasing affine expression
in~\eqref{eq:weight-u-barrier-coefficients} yields
\[
 B_{\mathrm{mass}}\le B_{\mathrm{mass}}^{\mathrm{raw}}
       \le\Omega/\ell_1<0,
\]
contradicting~\eqref{eq:reduced-domain}.

The same unit supplies the explicit alternative for
$B_{\mathrm{mix}}\ge0$. Write
$B_{\mathrm{mix}}=b_0^{\mathrm{mix}}+
b_1^{\mathrm{mix}}v+b_2^{\mathrm{mix}}v^2$ by collecting powers of $v$.
On the same $(r,\rho,u)$ domain, it certifies
\begin{equation}
 \begin{aligned}
 Q_0&=(\beta_1^{\mathrm{mass}})^2b_0^{\mathrm{mix}}
   -\beta_0^{\mathrm{mass}}\beta_1^{\mathrm{mass}}b_1^{\mathrm{mix}}
   +(\beta_0^{\mathrm{mass}})^2b_2^{\mathrm{mix}}<0,\\
 Q_1&=\beta_1^{\mathrm{mass}}b_1^{\mathrm{mix}}
                    -2\beta_0^{\mathrm{mass}}b_2^{\mathrm{mix}}<0,\\
 Q_2&=b_2^{\mathrm{mix}}<0.
 \end{aligned}
 \label{eq:weight-opposite-u-barrier}
\end{equation}
The nonnegative coefficient bound in~\eqref{eq:reduced-domain}
implies $B_{\mathrm{mass}}^{\mathrm{raw}}\ge0$, hence
$v\ge v_W=-\beta_0^{\mathrm{mass}}/\beta_1^{\mathrm{mass}}$.
Putting $d=v-v_W\ge0$ gives the exact identity
\[
 B_{\mathrm{mix}}(v_W+d)
   =\frac{Q_0}{(\beta_1^{\mathrm{mass}})^2}
        +\frac{Q_1}{\beta_1^{\mathrm{mass}}}d+Q_2d^2<0.
\]
This contradicts the opposite segment inequality. Neither argument
uses a sign of $N_C$, so both are valid after the weight change.
We have proved $u<1$ throughout~\eqref{eq:reduced-domain}.

Finally, reaching the cap $h_2=1$ would require
\begin{equation}
 u\ge\frac{r(1+\rho)}{\rho(1-r)}
       \ge\frac{r(2-r)}{(1-r)^2}>1.
 \label{eq:weight-h2-cap-inactive}
\end{equation}
The middle inequality uses $\rho\le1-r$; the last is equivalent
to $4r-2r^2-1>0$, which is already positive at $r=3/10$ and
increases for $r\le1/3$. Since $u<1$, $h_2$ is uncapped.
Only the change of formula for $h_1$ at $v=v_0$ remains.
The next paths show how every remaining point reaches that boundary
or a boundary system that is excluded directly.

\subsection{Paths to the switch boundary}
\label{app:paths}

We now show that a point satisfying the reduced necessary conditions
can be moved to $v=v_0$, or to one of two boundaries excluded below.
The moves vary scalar variables and preserve the specified
inequalities, producing new points in the necessary system.
Throughout, $r$ and $\rho$ are fixed. With $m,n,z,A_0$ defined in
\eqref{eq:weight-balancing-values}, the conditions used by the paths are
\begin{equation}
 \begin{gathered}
 r_0\le r\le1/3,\qquad r\le\rho\le1-r,\qquad
 k\le u<1,\quad k\le v\le1/k,\qquad S>r>T,\\
 m=S-r>0,\quad n=r-T>0,\quad z=n/m,\quad
 A_0=z/(1+z),\qquad T,s_2\ge0,\\
 G\ge0,\quad B_{\mathrm{mass}}\ge0,\quad E\ge0,\qquad
 L\ge0\ \text{or}\
 \bigl(N_R<0\le N_C\ \text{and}\ B_{\mathrm{mix}}\ge0\bigr).
 \end{gathered}
 \label{eq:path-conditions}
\end{equation}
The three inequalities involving $G,B_{\mathrm{mass}},E$ are
\eqref{eq:relaxed}; the determinant polynomials are defined in
\eqref{eq:mixing-determinants}. When $\rho\le1/2$, we also use
$s_1\ge0$. No such lower bound is imposed for $\rho>1/2$.
For differentiating $E$, it is convenient to define
\begin{equation}
 F=E/z,\qquad R_s=r-s_1>0,\qquad A_h=h_2-r>0.
 \label{eq:path-derivative-quantities}
\end{equation}
Since $z>0$, $F$ and $E$ have the same sign. The symbol $F$ here
is distinct from the candidate regrets $F_j^i$.
The preliminary bound $s_1<r$ gives $R_s>0$, and
$G=m(A_h-zR_s)\ge0$ gives $A_h\ge zR_s>0$.

The switch value lies within the allowed range of $v$. Indeed,
$v_0/k=(1+\sigma)/\sigma>1$, and
$v_0\le1/k$ is equivalent to
$r^2(1+\sigma)\le\sigma(1-r)^2$.
The latter follows from $\sigma\ge r$ and $r\le1/3$.
Thus moving toward $v_0$ cannot leave the range in
\eqref{eq:path-conditions} through an endpoint of $v$.

First suppose $v\ge v_0$, so the definition of $h_1$ in
\eqref{eq:scalar-caps} gives $h_1=1$. We keep $u$ fixed and
decrease $v$ to $v_0$. On the domain \eqref{eq:path-conditions}
with $v\ge v_0$, the unit \nolinkurl{postclip-envelope}\footnote{Checker:
\nolinkurl{proof/bound/global/replay_postclip_crossing_certificate.py}.\label{fn:checker-postclip-envelope}} gives
\begin{equation}
 z_v>0,\qquad F=0\ \Longrightarrow\ F_v<0.
 \label{eq:capped-path-signs}
\end{equation}
A subscript denotes a partial derivative with the other variables
fixed. At a first crossing from $F\ge0$ to $F<0$ as $v$ decreases,
the directional derivative would be nonpositive, whereas
\eqref{eq:capped-path-signs} gives $-F_v>0$. Hence $E\ge0$ is
preserved, including when the initial point satisfies $E=0$.
The variables $s_1,h_2$ are constant during this move, and
\begin{equation}
 G=m(A_h-zR_s),\qquad
 B_{\mathrm{mass}}=m[(1-r)-rz].
 \label{eq:path-positive-factors}
\end{equation}
Since $z$ decreases, both inequalities in
\eqref{eq:path-positive-factors} remain valid while $m>0$.
On the same domain, using $T\ge0$ and $u<1$, the unit
\nolinkurl{postclip-knots}\footnote{Checker:
\nolinkurl{proof/bound/global/replay_postclip_knot_direction.py}.\label{fn:checker-postclip-knots}} gives the determinant signs
\begin{equation}
 L_v<0,\qquad (B_{\mathrm{mix}})_v<0,\qquad
 (N_R)_v\ge0,\qquad (N_C)_v<0.
 \label{eq:capped-determinant-signs}
\end{equation}
Decreasing $v$ therefore preserves $L\ge0$ in the first alternative
of \eqref{eq:path-conditions}. In the second alternative it preserves
$B_{\mathrm{mix}}\ge0$ and $N_R<0\le N_C$.
The first alternative uses $L\ge0$ alone and does not require the
initial signs of $N_R,N_C$ to persist.

The unit \nolinkurl{postclip-orientation}\footnote{Checker:
\nolinkurl{proof/bound/global/replay_postclip_orientation_certificate.py}.\label{fn:checker-postclip-orientation}} gives $T<r$ and
checks the endpoint identities on this domain. Since
$T_v=s_1-r<0$, decreasing $v$ increases $T$ and preserves $T\ge0$.
The inequality $s_2\ge0$ is also preserved. If $\chi_1\le0$,
decreasing $v$ increases $s_2$. If $\chi_1>0$, then $\rho>k$ and
the minimum of $s_2$ during the move is attained at $v_0$; the identity
\begin{equation}
 \rho\sigma(1-r)\left.s_2\right|_{v=v_0}
 =r(1-3r)+(1-r)(\rho-k)(1-r-\rho)\ge0
 \label{eq:capped-s2-endpoint}
\end{equation}
establishes that minimum's nonnegativity.
At a possible limiting point, $n=0$ would mean $T=r$, contrary
to the cap bound. If $m=0$ while $n>0$, then $z=n/m$ would diverge,
contrary to its decrease in \eqref{eq:capped-path-signs}.
The denominator $m=S-r$ therefore remains positive throughout
this move, including its endpoint.
The endpoint satisfies $S>r>T$, $0<A_0\le1-r$, $h_2>r$, $E\ge0$,
and one determinant alternative in \eqref{eq:path-conditions}.

Now suppose $v<v_0$, so $h_1$ equals the affine expression inside
the minimum in \eqref{eq:scalar-caps}. We keep $v$ fixed and
increase $u$ until a determinant reaches zero. The unit
\nolinkurl{fixed-v-and-u1}\textsuperscript{\ref{fn:checker-fixed-v-and-u1}} reconstructs the source and derivative
signs used in this move.
Here $S_u=s_2-r<0$,
$T_u=\rho(1-r)/\sigma+v\chi_2$, and
\begin{equation}
 z_u=\frac{-T_u+z(r-s_2)}m,\qquad
 T_uR_s-A_h(r-s_2)>0,\qquad
 \partial_u(A_h/R_s)>0.
 \label{eq:fixed-v-source-signs}
\end{equation}
The inequality $G\ge0$ in \eqref{eq:path-positive-factors} is
$z\le A_h/R_s$. Substituting it into
\eqref{eq:fixed-v-source-signs} gives $z_u<0$; the increasing upper
bound $A_h/R_s$ then preserves $G\ge0$.
For $E\ge0$, differentiation of $F=E/z$ gives
\begin{equation}
 F_u=v\bigl(\chi_2+(h_2)_u/z\bigr)
       -(r+vA_h/z^2)z_u,
 \qquad (h_2)_uR_s+\chi_2A_h>0.
 \label{eq:fixed-v-envelope-signs}
\end{equation}
When $\chi_2\ge0$, the first bracket is positive because
$(h_2)_u=\rho(1-r)/\sigma>0$. When $\chi_2<0$, the positive
quantities $A_h,R_s,z$ and $z\le A_h/R_s$ give
\[
 \chi_2+(h_2)_u/z
 \ge[\chi_2A_h+(h_2)_uR_s]/A_h>0.
\]
Both terms in \eqref{eq:fixed-v-envelope-signs} are thus positive,
so $F_u>0$ preserves $E\ge0$. The decrease of $z$ also preserves
$B_{\mathrm{mass}}\ge0$ through \eqref{eq:path-positive-factors}.
These sign checks cover the full range $r\le\rho\le1-r$ and allow
the affine expression for $s_1$ to become negative when $\rho>1/2$.

During this move, $s_2$ is constant and the checked sign $T_u>0$
preserves $T,s_2\ge0$. For $\rho\le1/2$, $\chi_2\ge0$ also
preserves $s_1\ge0$. In the alternative $N_R<0\le N_C$, the
checked sign $(N_R)_u<0$ preserves $N_R<0$. If $N_C$ reached zero,
then \eqref{eq:mixing-determinants} would give
$B_{\mathrm{mix}}=rvN_R<0$, contradicting the determinant inequality.
The previously proved implication $S\ge r,\ u\le1\Rightarrow T<r$
excludes $n=0$. A limit with $m=0$ and $n>0$ would contradict the
decrease of $z$. Finally, the bounds using $B_{\mathrm{mass}}$
and the determinant at $u=1$ exclude that endpoint, as established
in the weight reduction. Thus the move reaches $L=0$ or
$B_{\mathrm{mix}}=0$ before any other condition can fail.

For $\rho\le1/2$, the unit \nolinkurl{canonical-alignment}\footnote{Checker:
\nolinkurl{proof/bound/preclip_canonical_alignment.py}.\label{fn:checker-canonical-alignment}} gives
$S\ge r,\ u\le1,\ v<v_0\Rightarrow N_C<0$ on the range in
\eqref{eq:path-conditions}; the first segment condition therefore gives
$L\ge0$. We choose a contact approached from $L>0$ as $u$ increases,
so $L_u\le0$ there. If the initial point already has $L=0$ and
$L_u\le0$, it is used immediately. If instead $L=0$ and $L_u>0$,
the move first enters $L>0$ and continues to the next contact.
We now justify the additional bounds needed to continue along
$L=0$. They are proved before the continuation checker is applied.

First write the polynomial $L$, which is affine in $v$, as
\begin{equation}
 L=\ell_0+\ell_1v,\qquad
 \ell_0=\sigma\chi_2(\sigma+\rho u)(\chi_1+cu),\qquad \ell_1=L_v.
 \label{eq:left-affine-graph}
\end{equation}
The inequalities $\chi_2\ge0$ and $s_1\ge0$ give
$\chi_1+cu=s_1+(\rho r/\sigma)u>0$.
Consequently $\ell_0>0$, except at
$r=1/3$, $\rho=\sigma=1/2$. At that corner
$L=-u^2v/18<0$, so no contact is possible. At a contact with $v>0$,
\begin{equation}
 \ell_1=-\ell_0/v<0,\qquad
 d_L=-\rho L_v=\rho\ell_0/v>0,\qquad v=\rho\ell_0/d_L.
 \label{eq:left-graph-denominator}
\end{equation}
This proves that division in the affine graph is valid before any
additional bound on $z$ is used.

To obtain $A_0>r$, define
\begin{equation}
 W_2=rS+(1-r)T-r,\qquad A_0-r=-W_2/(S-T).
 \label{eq:opposite-weight-test}
\end{equation}
On the present domain, the checked bound $T_u\ge r$ and $s_2\ge0$
imply
\[
 (W_2)_u=r(s_2-r)+(1-r)T_u
 \ge rs_2+r(1-2r)>0.
\]
The unit \nolinkurl{canonical-A_0-stop}\footnote{Checker:
\nolinkurl{proof/bound/simplify/replay-A_0-mass-stop-exclusion.py}.\label{fn:checker-canonical-A_0-stop}} gives $W_2(9/10,k)<0$ and
$W_2(9/10,v_0)<0$. Since $W_2$ is affine in $v$, it is negative
at $u=9/10$ throughout $k\le v\le v_0$.
Thus $W_2\ge0$ requires $u>9/10$.
Writing $W_2=w_0+w_1v$, where $w_0,w_1$ are its constant and
linear coefficients, the checks on $9/10<u<1$ give
\begin{equation}
 w_0<0,\qquad \ell_1<0,\qquad
 \Omega_2=w_0\ell_1-w_1\ell_0>0.
 \label{eq:opposite-weight-coefficients}
\end{equation}
If $W_2\ge0$, then $w_1>0$ and
$Lw_1=\ell_1W_2-\Omega_2<0$, contradicting $L=0$.
By \eqref{eq:opposite-weight-test}, this proves
\begin{equation}
 A_0>r,\qquad z>k.
 \label{eq:left-lower-weight}
\end{equation}
The equivalence uses $A_0=z/(1+z)$ and $k=r/(1-r)$.

To bound $z$ from above, use the quadratic obtained from $E$ in
\eqref{eq:relaxed}:
\begin{equation}
 \mathcal E(t)=vA_h+[c-v(r+R_s)]t-rt^2\quad(t\in\mathbb R),
 \qquad E=\mathcal E(z).
 \label{eq:envelope-quadratic}
\end{equation}
On the graph in \eqref{eq:left-graph-denominator}, the unit
\nolinkurl{canonical-E1}\footnote{Checker:
\nolinkurl{proof/bound/simplify/replay-arc-e1-source-budget.py}.\label{fn:checker-canonical-E1}} proves
\begin{equation}
 \mathcal P_{\mathrm{arc}}
 =\rho^2\sigma^2d_L[-\mathcal E(1)-2G]>0.
 \label{eq:envelope-at-one}
\end{equation}
The identity defines $\mathcal P_{\mathrm{arc}}$ after graph
substitution; its clearing factor is positive.
Its 600 rational Bernstein coefficients are checked on
$r\in[619/2000,1/3]$, $\rho\in[3/10,1/2]$, $u\in[2/5,1]$,
a box containing the present domain because $r\ge r_0>619/2000$,
$\rho\ge r>3/10$, and $u\ge k>2/5$.
Of the 600 coefficients, 582 are individually positive. The
remaining terms form six quadratic groups. For each group with
coefficients $b_2,b_3,b_4$, the checker verifies $b_2,b_4>0$ and
$3b_2b_4-2b_3^2>0$. For its unit coordinate $w\in[0,1]$, the factor
\[
 6b_2(1-w)^2+4b_3w(1-w)+b_4w^2
\]
is positive because its matrix determinant is
$2(3b_2b_4-2b_3^2)>0$. The other Bernstein factors are nonnegative.
This proves \eqref{eq:envelope-at-one} without requiring each
individual coefficient to be positive.
Since $G\ge0$, it follows that $\mathcal E(1)<0$.
The identity
$\mathcal E(z)=z\mathcal E(1)+(1-z)(vA_h+rz)$
would then give $E<0$ for $z\ge1$. Hence
\begin{equation}
 z<1,\qquad B_{\mathrm{mass}}=m[(1-r)-rz]>mc>0.
 \label{eq:left-upper-weight}
\end{equation}
The second inequality shows that the mass condition is automatic
on this part of $L=0$.

One further bound excludes small $v$. On the branch $v<v_0$,
$S_u=s_2-r<0$ and $S_v>0$, so $u\ge k$ and $v\le4/5$ would imply
$S\le\left.S\right|_{u=k,v=4/5}$. The unit \nolinkurl{canonical-boundary-stops}\footnote{Checker:
\nolinkurl{proof/bound/simplify/replay-preclip-stop-events.py}.\label{fn:checker-canonical-boundary-stops}} checks the fifteen
coefficients proving $r-\left.S\right|_{u=k,v=4/5}>0$. Consequently,
\begin{equation}
 S\ge r\quad\Longrightarrow\quad v>4/5.
 \label{eq:left-v-lower}
\end{equation}

We can now continue the contact along $L=0$ toward $v_0$.
The preceding calculations justify all conditions of the
5,098-leaf continuation certificate:
\begin{equation}
 \begin{gathered}
 r_0\le r\le1/3,\quad r\le\rho\le1/2,\quad
 L=0,\quad L_u\le0,\quad d_L>0,\\
 m,n>0,\quad T,s_1,s_2\ge0,\quad G,E\ge0,\quad
 k\le u<1,\quad k<z<1,\quad 4/5<v<v_0,\quad k\le v\le1/k.
 \end{gathered}
 \label{eq:continuation-domain}
\end{equation}
On this domain, the unit \nolinkurl{canonical-upper-arc}\footnote{Checker:
\nolinkurl{proof/bound/kink/replay_preclip_bernstein_interval.py}.\label{fn:checker-canonical-upper-arc}} gives
the exact interval and Bernstein checks
\begin{equation}
 L_u<0,\quad L_v<0,\quad u>k,\quad s_1>0,\quad s_2>0,
 \qquad
 E=0\ \Longrightarrow\ E_v-E_uL_v/L_u>0.
 \label{eq:continuation-signs}
\end{equation}
The implicit-function theorem~\cite{rudin1976}
gives a curve $L=0$ with increasing
$v$ and $du/dv=-L_v/L_u<0$. At $E=0$, the last sign in
\eqref{eq:continuation-signs} is precisely $dE/dv>0$ and excludes
exit through $E=0$.

For completeness, we check the other possible exits.
The calculations establishing \eqref{eq:left-lower-weight},
\eqref{eq:left-upper-weight}, and \eqref{eq:left-v-lower} apply
at every point satisfying the same inequalities on $L=0$.
If $G=0$, then $A_h=zR_s$ and
\[
 E=z[c-r(v+z)]<0,
\]
because $v+z>4/5+k>c/r$ for $r\ge r_0$, contradicting $E\ge0$.
If $T=0$, then $z=r/(S-r)<1$ gives $S>2r$, whereas the separate
cap check in \nolinkurl{canonical-boundary-stops}\textsuperscript{\ref{fn:checker-canonical-boundary-stops}} at $T=0$ gives $S<2r$.
The strict signs in \eqref{eq:continuation-signs} exclude
$u=k$, $s_1=0$, $s_2=0$, $L_u=0$, and $L_v=0$.
The implication $S\ge r,\ u\le1\Rightarrow T<r$ excludes $n=0$,
and $z<1$ then excludes $m=0$. Equation~\eqref{eq:left-upper-weight}
excludes $B_{\mathrm{mass}}=0$.
The domain bounds and these strict exclusions prevent a limiting
stop before $v_0$; the continuation therefore reaches $v=v_0$.

For $\rho\ge1/2$, we instead exclude the two contact boundaries
directly. Their common conditions are
\begin{equation}
 \begin{gathered}
 r_0\le r\le1/3,\quad 1/2\le\rho\le1-r,\quad
 k\le u<1,\quad k\le v\le v_0,\\
 S>r>T,\quad T,s_2\ge0,\quad G,B_{\mathrm{mass}},E\ge0.
 \end{gathered}
 \label{eq:large-weight-contact-domain}
\end{equation}
The unit \nolinkurl{high-rho-left}\footnote{Checker:
\nolinkurl{proof/bound/global/replay_preclip_left_interval.py}.\label{fn:checker-high-rho-left}} adds $L=0$ to
\eqref{eq:large-weight-contact-domain} and needs no signs on
$N_R,N_C$. The unit \nolinkurl{high-rho-below}\footnote{Checker:
\nolinkurl{proof/bound/global/replay_preclip_below_interval.py}.\label{fn:checker-high-rho-below}} adds
$B_{\mathrm{mix}}=0$ and $N_R<0\le N_C$.
Both domains allow $s_1<0$ and $A_0<r$.
The two exclusion certificates have 19,771 and 4,945 leaves;
the continuation certificate has 5,098 leaves.
Each leaf is checked by exact rational interval or Bernstein bounds.
Substitution in an affine graph such as $L=0$ is made only after
its linear coefficient has a certified nonzero sign. A box containing
a possible zero coefficient is split or checked without division.
The independent replayers check all leaf inequalities and the
complete binary partitions, including these division conditions.
Together, the paths and the two exclusions cover every point
satisfying \eqref{eq:path-conditions}.

\subsection{Exclusion on the switch boundary}
\label{app:switch}

We now exclude the endpoints $v=v_0$ produced by the preceding
paths when $r>r_*$. The proof must exclude simultaneous inequalities:
the nonexistence of a common equality root alone would not suffice.
We use only the following conditions on this boundary:
\begin{equation}
 \begin{gathered}
 r_0\le r\le1/3,\quad r\le\rho\le1-r,\quad
 v=v_0,\quad 1/2<\eta\le1-r,\\
 S>r>T,\quad 0<A_0\le1-r,\quad h_2>r,\quad E\ge0,
 \end{gathered}
 \label{eq:switch-domain}
\end{equation}
together with one determinant alternative from
\eqref{eq:path-conditions}. The condition $\eta>1/2$ follows from
$u<1$. Thus $h_2$ equals its affine expression in
\eqref{eq:scalar-caps}, as in the definition of $N_2$.
Neither $E_2\ge0$ nor $s_1,s_2\ge0$ is required here.

The polynomial definitions in \eqref{eq:root-polys} and the
determinant definitions in \eqref{eq:mixing-determinants} agree
on this boundary. For the second determinant, define $N_B$ by
\begin{equation}
 [q_\dagger(1-r)-J_R](C_k-r)+q_\dagger r(R_k-r)
       =\frac{N_B(r,\rho,\eta)}{\sigma D_cJ^2}.
 \label{eq:second-boundary-numerator}
\end{equation}
The left side is the intersection determinant for the segment
toward $(0,q_\dagger)$; $D_c,J,J_R,C_k,R_k$ were defined in the
boundary polynomial calculation. Exact substitution gives
\begin{equation}
 \begin{aligned}
 N_1&=\rho\sigma\eta^2(1-r)L,\\
 N_B&=\sigma^2\eta(1-r)^2B_{\mathrm{mix}},\\
 \Psi(A_0)&=(1-A_0)^2E=N_2/D_2.
 \end{aligned}
 \label{eq:switch-polynomial-scales}
\end{equation}
Every factor multiplying a determinant, and $D_2$, is positive
under \eqref{eq:switch-domain}. Hence $L\ge0$ gives $N_1\ge0$,
$B_{\mathrm{mix}}\ge0$ gives $N_B\ge0$, and $E\ge0$ gives $N_2\ge0$.

The alternative requiring $N_B$ can occur only when $\rho>1/2$.
To see this on $r\le\rho\le1/2$, write
$T=T_0+T_1/\eta$ on $v=v_0$, where $T_0,T_1$ are its constant
and reciprocal coefficients. Define
$\eta_T=T_1/(r-T_0)$ and the polynomial
$\widehat C=\sigma D_cJ(C_k-r)$.
The unit \nolinkurl{full-rho-kink-face}\footnote{Checker:
\nolinkurl{proof/bound/critical/replay_full_rho_kink_face.py}.\label{fn:checker-full-rho-kink-face}} checks the boundary
polynomial normalizations in~\eqref{eq:root-polys} and
\eqref{eq:second-boundary-numerator}, and gives the following
signs on this smaller-weight domain:
\begin{equation}
 T_1>0,\quad r-T_0>0,\quad
 \partial_\eta\widehat C<0,\quad \widehat C(\eta_T)<0.
 \label{eq:small-switch-orientation}
\end{equation}
Since $T<r$ implies $\eta>\eta_T$, these inequalities give
$C_k<r$, or $N_C<0$ by \eqref{eq:mixing-numerators}.
The alternative $N_R<0\le N_C$ is therefore impossible there.

Both determinant alternatives will use the same sign argument
for $\Psi(A_0)$ as $\eta$ increases. With its cap coefficients
fixed, the quadratic $\Psi$ satisfies
$\Psi(0)=v(h_2-r)>0$ and $\Psi(1)=-r<0$.
It has a unique zero in $(0,1)$, and $\Psi_a<0$ at that zero.
Holding $a$ fixed while varying $\eta$ gives
\begin{equation}
 \left.\partial_\eta\Psi(a)\right|_a
 =-\frac{(1-a)v}{\eta^2}
   \left[a\chi_2+(1-a)\frac{\rho(1-r)}\sigma\right]<0
 \qquad(0\le a\le1-r).
 \label{eq:switch-fixed-argument-derivative}
\end{equation}
The bracket is affine in $a$ and has positive endpoint values
$\rho(1-r)/\sigma$ and $(1-r)(1-2r)$.
Consequently, on any interval satisfying the conditions in
\eqref{eq:switch-domain} other than $E\ge0$, and also
$(A_0)_\eta>0$, a zero of $\Psi(A_0)$ satisfies
\begin{equation}
 \frac{d}{d\eta}\Psi(A_0)
 =\left.\partial_\eta\Psi(a)\right|_{a=A_0}
      +\Psi_a(A_0)(A_0)_\eta<0.
 \label{eq:switch-zero-derivative}
\end{equation}
A negative value therefore cannot become nonnegative as $\eta$
increases on that interval. Below we check the interval conditions
separately for $N_1$ and $N_B$.

We first treat $N_1\ge0$. Write
\begin{equation}
 N_1=a_2\eta^2+b_2\eta+c_2,\qquad
 \Delta=(b_2^2-4a_2c_2)/\sigma^2,
 \label{eq:switch-discriminant}
\end{equation}
where $a_2,b_2,c_2$ are functions of $(r,\rho)$.
The quadratic is strictly concave throughout the physical domain.
Indeed, expansion gives
\[
 a_2=r\bigl[(1-r)\sigma^3(1-2\sigma)-r^2\rho^2(1+\sigma)\bigr].
\]
For $\rho\le1/2$, its first term is nonpositive and its second
is negative. For $\rho\ge1/2$, the bounds
$\sigma^3(1-2\sigma)\le1/32$ and $\rho^2(1+\sigma)\ge1/4$ give
$a_2/r\le(1-r)/32-r^2/4<0$ for $r\ge3/10$.

Before using a root of $N_1$, the unit \nolinkurl{full-rho-kink-face}\textsuperscript{\ref{fn:checker-full-rho-kink-face}}
excludes two larger parts of its domain:
\begin{equation}
 \begin{aligned}
 r_0\le r\le1/3,\quad r\le\rho\le21/50
    &\ \Longrightarrow\ b_2^2-4a_2c_2<0,\\
 8/25\le r\le1/3,\quad r\le\rho\le1/2,\quad r\le\eta\le1-r
    &\ \Longrightarrow\ N_1<0,\\
 8/25\le r\le1/3,\quad 1/2\le\rho\le1-r
    &\ \Longrightarrow\\
       (N_1)_\eta(1-r)>0,\quad N_1(1-r)<0.
 \end{aligned}
 \label{eq:switch-preliminary-signs}
\end{equation}
The first line excludes real roots. The second directly contradicts
$N_1\ge0$. For the third, concavity gives
$(N_1)_\eta(\eta)\ge(N_1)_\eta(1-r)>0$ when $\eta\le1-r$;
thus $N_1(\eta)\le N_1(1-r)<0$ there as well.

The remaining parameter pairs lie in the rational rectangles
\begin{equation}
 \mathcal C_-=[r_0,8/25]\times[21/50,1/2],\qquad
 \mathcal C_+=[r_0,8/25]\times[1/2,1-r_0],
 \label{eq:switch-rectangles}
\end{equation}
with coordinates $(r,\rho)$; $r\le\rho\le1-r$ is imposed for
a physical point. On both rectangles the checks give
$N_1(1/2)<0$ and $(N_1)_\eta(1/2)>0$.
Therefore $N_1\ge0$ requires $\Delta\ge0$ and
\begin{equation}
 \eta\ge\eta_-(r,\rho)>1/2,\qquad
 \eta_-(r,\rho)=\frac{-b_2+\sigma\sqrt\Delta}{2a_2}.
 \label{eq:switch-smaller-root}
\end{equation}
The displayed expression is the smaller root because $a_2<0$;
when $\Delta=0$, the two roots coincide.

The derivative checks of \nolinkurl{full-rho-kink-face}\textsuperscript{\ref{fn:checker-full-rho-kink-face}} on
\eqref{eq:switch-rectangles} give
\begin{equation}
 S_\eta>0,\quad T_\eta<0,\quad \left.S\right|_{\eta=1/2}>r>\left.T\right|_{\eta=1/2},\qquad
 (A_0)_\eta>0.
 \label{eq:first-switch-interval-signs}
\end{equation}
The last sign holds throughout $[\eta_-,\eta]$ for a proposed
feasible endpoint $\eta$. For the lower rectangle, positivity of $(A_0)_\eta$ uses
$\Delta\ge0$. For the upper rectangle, it uses
$N_1(1-r)\ge0$, which follows from $N_1(\eta)\ge0$ and the
checked positive derivative $(N_1)_\eta(1-r)$.
Throughout $[\eta_-,\eta]$, the first three signs give $S>r>T$.
The increasing $A_0$ stays positive and at most its final value
$1-r$. The explicit formula for $h_2$ gives $(h_2)_\eta<0$, so
$h_2$ stays above its final value $r$.
These facts justify applying \eqref{eq:switch-zero-derivative}
on the whole interval.

It remains to establish a negative starting value at $\eta_-$.
The resultant~\cite{cox2015}, which vanishes whenever $N_1,N_2$ share a root in
$\eta$, satisfies the exact identity
\begin{equation}
 \operatorname{Res}_{\eta}(N_1,N_2)
       =\sigma^6\rho^2r^2R_{\mathrm{resultant}}(r,\rho).
 \label{eq:switch-resultant}
\end{equation}
This defines the polynomial factor $R_{\mathrm{resultant}}$;
\nolinkurl{full-rho-kink-face}\textsuperscript{\ref{fn:checker-full-rho-kink-face}} reconstructs the resultant identity.
At the isolated solution $(r_*,\rho_*,\eta_*)$ of
\eqref{eq:natural-root} in \eqref{eq:app-root-box}, exact algebra gives
\begin{equation}
 R_{\mathrm{resultant}}(r_*,\rho)=(\rho-\rho_*)^2W(\rho),\qquad \deg W=18,
 \qquad
 \det\frac{\partial(N_1,N_2)}{\partial(r,\eta)}
          (r_*,\rho_*,\eta_*)\ne0.
 \label{eq:switch-resultant-at-root}
\end{equation}
The determinant permits solving locally for $(r,\eta)$ as functions
of $\rho$. The equation $H_{\mathrm{root}}=0$ then gives $r'(\rho_*)=0$.
Differentiating the resultant along this curve proves that its
$\rho$ derivative also vanishes there, yielding the double factor.
This determinant differs from the $3\times3$ Jacobian used for
root isolation.
The unit \nolinkurl{full-rho-kink-face}\textsuperscript{\ref{fn:checker-full-rho-kink-face}} reconstructs all nineteen
Bernstein coefficients of $W$ on $[r_*,1-r_*]$ and checks that
they are strictly negative. It also checks the exact signs
$\partial_r R_{\mathrm{resultant}}<0$ on both rectangles in
\eqref{eq:switch-rectangles}. Hence
\begin{equation}
 (r,\rho)\in\mathcal C_-\cup\mathcal C_+,\quad
 r>r_*,\quad r\le\rho\le1-r
 \quad\Longrightarrow\quad R_{\mathrm{resultant}}(r,\rho)<0.
 \label{eq:switch-resultant-negative}
\end{equation}
Indeed, such a $\rho$ lies in $[r_*,1-r_*]$, and increasing $r$
from $r_*$ strictly decreases the resultant factor.
Thus $N_1,N_2$ cannot have a common root in this domain.

To determine the sign, use the exact reference check in
\nolinkurl{full-rho-kink-face}\textsuperscript{\ref{fn:checker-full-rho-kink-face}}:
\begin{equation}
 0.55278962<\eta_-(r_*,1/2)<0.55278965,\qquad
 N_2(r_*,1/2,\eta_-)<0.
 \label{eq:switch-reference-sign}
\end{equation}
The signs $\Delta_\rho>0$ and $\Delta_r<0$ on
\eqref{eq:switch-rectangles} connect every real smaller-root point
to this reference. In the lower rectangle, first increase $\rho$
to $1/2$, then decrease $r$ to $r_*$. Both moves increase $\Delta$.
In the upper rectangle, first decrease $r$ to $r_*$, then decrease
$\rho$ to $1/2$. The last move stays in $\Delta>0$ because
$\Delta(r_*,\rho)\ge\Delta(r_*,1/2)>0$ for $\rho\ge1/2$.
The paths obey $r\le\rho\le1-r$. Their portions at $r>r_*$ have
nonzero resultant by \eqref{eq:switch-resultant-negative}; the
portion at $r=r_*$ has $\rho\ge1/2>\rho_*$ and hence nonzero
resultant by \eqref{eq:switch-resultant-at-root}.
Continuity of $\eta_-$ and $N_2$ therefore propagates the negative
reference sign to every such point. Points with $\Delta=0$, where
the two roots coincide, follow by continuity; $N_2=0$ there would
also contradict the nonzero resultant.
Thus \eqref{eq:switch-polynomial-scales} gives $\Psi(A_0)<0$
at $\eta_-$. Equations~\eqref{eq:first-switch-interval-signs} and
\eqref{eq:switch-zero-derivative} exclude a nonnegative value at
any larger feasible $\eta$. This excludes the simultaneous
conditions $N_1\ge0$, $N_2\ge0$.

We now treat $N_B\ge0$ on $1/2\le\rho\le1-r$.
The polynomial in \eqref{eq:second-boundary-numerator} is affine
in $\eta$. On $r_0\le r\le1/3$, $1/2\le\rho\le1-r$,
\nolinkurl{full-rho-kink-face}\textsuperscript{\ref{fn:checker-full-rho-kink-face}} gives
\begin{equation}
 (N_B)_\eta>0,\quad N_B(1/2)<0,\qquad
 \eta_B=-\frac{N_B(r,\rho,0)}{(N_B)_\eta}>1/2,
 \qquad N_B\ge0\ \Longrightarrow\ \eta\ge\eta_B.
 \label{eq:second-switch-root}
\end{equation}
The zero $\eta_B$ is uniquely determined by this affine equation.
Define $\mathcal G_B$ by the exact substitution identity
\begin{equation}
 N_2(r,\rho,\eta_B)
       =\frac{\sigma^3r\mathcal G_B(r,\rho)}{[(N_B)_\eta]^3}.
 \label{eq:second-switch-substitution}
\end{equation}
The denominator is positive. To check $\mathcal G_B$ on its full
physical triangle, substitute
$r=r_0+(1/3-r_0)s$ and $\rho=1/2+(1/2-r)t$, where $s,t\in[0,1]$.
After clearing positive rational constants,
\nolinkurl{full-rho-kink-face}\textsuperscript{\ref{fn:checker-full-rho-kink-face}} reconstructs all 612 Bernstein
coefficients and checks that they are strictly negative. Hence
\begin{equation}
 N_2(r,\rho,\eta_B)<0\qquad
 (r_0\le r\le1/3,\ 1/2\le\rho\le1-r).
 \label{eq:second-switch-negative}
\end{equation}

The interval conditions for this alternative must be checked on
the whole triangle, since \eqref{eq:switch-rectangles} covers less.
The checks of \nolinkurl{full-rho-kink-face}\textsuperscript{\ref{fn:checker-full-rho-kink-face}} on this triangle give
\begin{equation}
 \begin{gathered}
 S_\eta>0,\quad T_\eta<0,\quad \left.S\right|_{\eta=1/2}\ge r>\left.T\right|_{\eta=1/2},\\
 N_B(1-r)\ge0\quad\Longrightarrow\quad
 (A_0)_\eta>0\text{ for every }\eta\in[\eta_B,1-r].
 \end{gathered}
 \label{eq:second-switch-interval-signs}
\end{equation}
A feasible endpoint has $\eta_B\le\eta\le1-r$, so the increasing
$N_B$ gives $N_B(1-r)\ge0$.
The equality $\left.S\right|_{\eta=1/2}=r$ occurs only at $(r,\rho)=(1/3,2/3)$;
$\eta_B>1/2$ then ensures $S>r$ on the interval used.
As in the first alternative, $A_0$ increases toward a final value
at most $1-r$, and $h_2$ decreases toward a final value greater
than $r$. Thus all hypotheses of \eqref{eq:switch-zero-derivative}
hold from $\eta_B$ to the proposed endpoint. The starting sign
\eqref{eq:second-switch-negative} therefore excludes $N_2\ge0$
at that endpoint. No additional candidate algebraic root is needed.

The two alternatives arising from the paths in~\eqref{eq:path-conditions}
are now both excluded. The alternative $L\ge0$ gives $N_1\ge0$
and $N_2\ge0$, which the first switch argument excludes. The other
alternative has $B_{\mathrm{mix}}\ge0$ and $N_R<0\le N_C$;
\eqref{eq:small-switch-orientation} places it in the larger-weight
domain, where the second switch argument excludes $N_B,N_2\ge0$.
These conclusions use exactly the conditions preserved by the paths.
Combining them with the two direct contact exclusions proves
\begin{equation}
 r_*<r\le1/3\quad\Longrightarrow\quad
 \nexists\boldsymbol\theta,\ \Sigma_r(\boldsymbol\theta).
 \label{eq:global-exclusion}
\end{equation}

The proof uses computation for three stated purposes. Numerical
optimization suggests the boundary and the active equations.
Symbolic calculation constructs the specified nonnegative combinations,
cut points, determinant numerators, and resultant identities.
Exact checks establish the signs, root enclosure, and regional
coverage required by the mathematical argument.
The repository\codefootref records the inputs and dependencies of
29 verification units. Some units recompute fixed algebraic checks;
dedicated independent replayers reconstruct certificate identities,
coefficient signs, or all leaves and partitions of boundary trees.
Their mathematical premises are the building-block properties,
LP duality, payoff linearity, maximum attainment, and universal
instantiation with the specified strategy types. These checks certify
the arithmetic of the stated proof; they are not a completed
formalization in a theorem-prover kernel. None is a step of
Algorithm~\ref{alg:main}.

\section{Transfer of the approximation errors}
\label{app:errors}

We now prove that the actual errors $\delta_f$ and $\tau$ enter
the regret bound once each. This establishes the implication in
Lemma~\ref{lem:error-transfer} using the same necessary system
\eqref{eq:scalar-system} as in the ideal analysis. The error contracts
for the stationary-point and mixing operations come
from~\cite{ts2008,dfm2023,automating2025,mixing2025}.
Suppose an actual execution satisfies
\begin{equation}
 f(\boldsymbol x_{\mathrm{out}})>r+\delta_f+\tau,
 \qquad r_*<r<1/3.
 \label{eq:finite-error-failure}
\end{equation}
Output selection makes every candidate exceed this threshold.
Substituting any $p,q\in[0,1]$ into \eqref{eq:om} subtracts $\tau$
once. The error-free endpoint comparison at $p=q=0$ also gives
$g\ge f(\boldsymbol x_o)$. Hence
\begin{equation}
 \forall p,q\in[0,1],\quad f(\boldsymbol x(p,q))>r+\delta_f,
 \qquad g>r+\delta_f+\tau.
 \label{eq:finite-error-rectangle-failure}
\end{equation}
Primal optimality of $\sm(y_s^1)$, $\sm(y_s^2)$,
$\sm(y_1^1)$, and $\sm(y_1^2)$ gives the same strict
$r+\delta_f$ lower bound throughout each call's domain with its input
strategy fixed. All four LPs are solved exactly, so their
duality and complementary-slackness properties add no error.

To use stationarity, substitute $(x_s^1,y_s^2)$ and
$(y_s^1,x_s^2)$ into \eqref{eq:properties-stationary}.
With $\lambda_{\mathrm{payoff}},\mu_{\mathrm{payoff}}$ the payoff differences in~\eqref{eq:stationary-payoff-differences}, this yields
\begin{equation}
 G_s=g-\delta_f>r,\qquad
 \rho\lambda_{\mathrm{payoff}}\ge G_s,\qquad \sigma\mu_{\mathrm{payoff}}\ge G_s,
 \qquad r<\rho<1-r.
 \label{eq:finite-error-stationary-value}
\end{equation}
The last interval follows because $\lambda_{\mathrm{payoff}},\mu_{\mathrm{payoff}}\le1$ and
$\rho+\sigma=1$. In particular, both stationary weights are positive.
The notation $G_s$ denotes this single error-adjusted stationary
value; it differs from the polynomial $G$ in \eqref{eq:relaxed}.
The error-adjusted lower-bound values are
\begin{equation}
 \lambda_0=G_s/\rho,\qquad \mu_0=G_s/\sigma.
 \label{eq:finite-error-gap-lower-values}
\end{equation}
They reduce to~\eqref{eq:stationary-gap-lower-values} when
$\delta_f=0$.
Equation~\eqref{eq:ts-best-response} gives the actual errors
\begin{equation}
 \begin{aligned}
 \tau_1&=B_1(x_s^2)-u_1(y_s^1,x_s^2),\\
 \tau_2&=B_2(x_s^1)-u_2(x_s^1,y_s^2),
 \end{aligned}
 \qquad 0\le\tau_1,\tau_2\le\delta_f.
 \label{eq:finite-error-best-responses}
\end{equation}

The dual inequality~\eqref{eq:first-row-instance} for
$\sm(y_s^1)$ is exact. Instantiating it at a column best response
to $y_s^1$ and at $x_s^2$, respectively, gives
\begin{equation}
 F_1^1\le\eta,\qquad
 F_1^1\le1-\eta+\eta\tau_1\le1-\eta+\eta\delta_f.
 \label{eq:finite-error-first-weight}
\end{equation}
Together with $F_1^1>r+\delta_f$, these inequalities give
$\eta>r$ and
$r<1-\eta-\delta_f+\eta\tau_1\le1-\eta$.
Thus $r<\eta<1-r$. The corresponding instantiations of the
dual inequality for $\sm(y_s^2)$ give $r<\xi<1-r$.
Positive first-round weights permit the exact best-response
conclusion \eqref{eq:alignment} at $x_1^2,x_1^1$.
Instantiating the dual inequality for $\sm(y_1^1)$ at $x_1^2$
and at a column best response to $y_1^1$ gives
$F_2^1\le1-\alpha$ and $F_2^1\le\alpha$; the corresponding
instantiations for $\sm(y_1^2)$ are symmetric. We therefore have
\begin{equation}
 r<\rho,\eta,\xi,\alpha,\gamma<1-r.
 \label{eq:finite-error-weight-ranges}
\end{equation}

The row regret of $y_1^1$, the dual strategy returned by
$\sm(y_s^1)$, against $x_s^2$ must be bounded before taking the
minimum with one. Set
$d=B_1(x_s^2)-u_1(y_1^1,x_s^2)\ge0$.
Instantiate the dual inequality~\eqref{eq:first-row-instance}
for $\sm(y_s^1)$ at $x_s^2$. This gives
\begin{equation}
 \begin{aligned}
 r+\delta_f<F_1^1&\le\eta(\tau_1-d)+(1-\eta),\\
 \eta d&<1-\eta-r-\delta_f+\eta\tau_1\le1-\eta-r.
 \end{aligned}
 \label{eq:finite-error-old-payoff-loss}
\end{equation}
Since normalized payoffs also give $d\le1$, we obtain the same
$d\le d_1=\min\{1,(1-\eta-r)/\eta\}$ as in the ideal argument.
The symmetric instantiation of the dual inequality for
$\sm(y_s^2)$ at $x_s^1$ gives the ideal bound $d_2$ for the column
regret. In each case the error is canceled before the minimum
is taken; an additional subtraction from the capped quantity is
unnecessary. The unit \nolinkurl{near-BR-delta-ledger}\footnote{Checker:
\nolinkurl{proof/bound/critical/replay_delta_f_lift_ledger.py}.\label{fn:checker-near-BR-delta-ledger}}
reconstructs this cancellation identity and both cases of the
minimum with one.

The payoff bounds can now be evaluated at
\begin{equation}
 \boldsymbol\lambda_s=(r,G_s,\rho,\eta,\xi).
 \label{eq:finite-error-payoff-parameters}
\end{equation}
For example, instantiate the dual inequality~\eqref{eq:first-row-instance}
for $\sm(y_s^1)$ at $y_s^2$ and use the stationary payoff lower bounds in
\eqref{eq:finite-error-stationary-value}:
\begin{equation}
 \begin{aligned}
 r+\delta_f<F_1^1
 &\le\eta[u_1(y_1^1,y_s^2)-G_s/\rho]
                   +(1-\eta)(1-G_s/\sigma),\\
 1-u_1(y_1^1,y_s^2)
 &<s_1(\boldsymbol\lambda_s)-\delta_f/\eta\le s_1(\boldsymbol\lambda_s).
 \end{aligned}
 \label{eq:finite-error-first-payoff}
\end{equation}
Here $s_1(\boldsymbol\lambda_s)$ is the function in
\eqref{eq:payoff-bound-s} evaluated at the tuple
\eqref{eq:finite-error-payoff-parameters}. Instantiating the dual inequality for
$\sm(y_s^2)$ at $y_s^1$ gives $s_2(\boldsymbol\lambda_s)$.
For the cross payoff, substitute $(y_1^1,x_s^2)$ into
\eqref{eq:properties-stationary}. Using
\eqref{eq:finite-error-old-payoff-loss} gives
\[
 G_s\le\rho(d-\tau_1)
       +\sigma[u_2(y_1^1,y_s^2)-u_2(y_1^1,x_s^2)],
 \qquad
 u_2(y_1^1,y_s^2)\ge G_s/\sigma-(\rho/\sigma)d_1.
\]
This is the lower bound giving $h_2(\boldsymbol\lambda_s)$ in
\eqref{eq:payoff-bound-h}; the symmetric substitution gives $h_1(\boldsymbol\lambda_s)$.
Instantiating the dual inequalities for $\sm(y_s^1)$ at $y_1^2$
and for $\sm(y_s^2)$ at $y_1^1$ gives $S(\boldsymbol\lambda_s),T(\boldsymbol\lambda_s)$,
respectively, by~\eqref{eq:payoff-bound-st}, using $F_1^i>r$ and
these payoff bounds. There is no new stationary error in either
instantiation.

The joint regret comparisons also have favorable error signs.
For arbitrary $x^{1\prime}\in\Delta_{n_1}$ and
$x^{2\prime}\in\Delta_{n_2}$, define
\begin{equation}
 \begin{aligned}
 \ell_1^s&=B_1(x_s^2)-u_1(x^{1\prime},x_s^2),&
 \ell_1^y&=B_1(y_s^2)-u_1(x^{1\prime},y_s^2),\\
 \ell_2^s&=B_2(x_s^1)-u_2(x_s^1,x^{2\prime}),&
 \ell_2^y&=B_2(y_s^1)-u_2(y_s^1,x^{2\prime}).
 \end{aligned}
 \label{eq:finite-error-payoff-losses}
\end{equation}
All four quantities are nonnegative. The column comparison follows
by adding $\eta/\rho$ times stationarity instantiated at
$(x_s^1,x^{2\prime})$ to the dual inequality for $\sm(y_s^1)$
instantiated at the same $x^{2\prime}$, and bounding
$u_1(y_1^1,x^{2\prime})-u_1(x_s^1,x^{2\prime})\le1$.
The corresponding calculation using the dual inequality for
$\sm(y_s^2)$ is symmetric. The resulting inequalities are
\begin{equation}
 \begin{aligned}
 \frac{\eta\sigma}{\rho}\ell_2^s+(1-\eta)\ell_2^y
 &\ge r-\eta(1-G_s/\rho)+\frac{\eta\sigma}{\rho}\tau_2,\\
 \frac{\xi\rho}{\sigma}\ell_1^s+(1-\xi)\ell_1^y
 &\ge r-\xi(1-G_s/\sigma)+\frac{\xi\rho}{\sigma}\tau_1.
 \end{aligned}
 \label{eq:finite-error-joint-losses}
\end{equation}
The last term on each right side is nonnegative, so omitting it
weakens the lower bound and gives exactly the ideal joint
regret comparison evaluated at $g=G_s$.
Applying the same minimum argument as in
\eqref{eq:gap-lower-bounds} therefore bounds the actual Jensen
gaps from below by $j_r(\boldsymbol\lambda_s,q),j_c(\boldsymbol\lambda_s,p)$.

For the rectangle regrets, use \eqref{eq:mixing-regrets} with
$g=G_s+\delta_f$. The endpoint regrets against the stationary
opponent are $\tau_1,\tau_2$, while
\eqref{eq:finite-error-stationary-value} gives the opposite endpoint
bounds $1-G_s/\rho,1-G_s/\sigma$.
The error contributions satisfy, for every $p,q\in[0,1]$,
\begin{equation}
 \begin{aligned}
 \delta_f(1-p)(1-q)+\tau_1p(1-q)
     &\le\delta_f(1-q)\le\delta_f,\\
 \delta_f(1-p)(1-q)+\tau_2(1-p)q
     &\le\delta_f(1-p)\le\delta_f.
 \end{aligned}
 \label{eq:finite-error-rectangle-budget}
\end{equation}
Together with \eqref{eq:finite-error-joint-losses} and the
functions in \eqref{eq:rectangle-upper-bounds}, this gives
\begin{equation}
 \begin{aligned}
 \forall p,q\in[0,1],\quad f_1(\boldsymbol x(p,q))
     &\le R_{G_s}(p,q)+\delta_f(1-q),\\
 \forall p,q\in[0,1],\quad f_2(\boldsymbol x(p,q))
     &\le C_{G_s}(p,q)+\delta_f(1-p).
 \end{aligned}
 \label{eq:finite-error-rectangle-bounds}
\end{equation}
Here $R_{G_s},C_{G_s}$ are the functions $R_g,C_g$ evaluated at
$g=G_s$. If their maximum were at most $r$ at any $(p,q)$,
\eqref{eq:finite-error-rectangle-budget}--\eqref{eq:finite-error-rectangle-bounds}
would give $f(\boldsymbol x(p,q))\le r+\delta_f$, contradicting
\eqref{eq:finite-error-rectangle-failure}.
Thus all four segment predicates \eqref{eq:segment-test} hold
with the common parameter $G_s$.

The remaining ideal necessary inequalities follow from the same
payoff bounds and the exact $\sm$ dual comparisons. In particular,
all six bounds $s_i(\boldsymbol\lambda_s),h_i(\boldsymbol\lambda_s),S(\boldsymbol\lambda_s),T(\boldsymbol\lambda_s)$ are nonnegative,
$U_{G_s},V_{G_s},H_{G_s},K_{G_s}\ge r$, and
$E_{1,G_s},E_{2,G_s}\ge0$.
The shared-envelope calculation \eqref{eq:common-payoff}--\eqref{eq:common-envelope}
uses only these bounds and the same actual $\alpha,\gamma$;
it introduces no further stationary error.
Instantiating the dual inequalities for $\sm(y_s^1)$ at $y_s^2$
and for $\sm(y_s^2)$ at $y_s^1$ also gives the parameter range
\begin{equation}
 r<G_s\le\min\left\{
    \frac{\rho\sigma(1-r)}{D_c},
    \frac{\rho\sigma(1-r)}{D_r}\right\}.
 \label{eq:finite-error-projection-domain}
\end{equation}
Equations~\eqref{eq:finite-error-weight-ranges} and
\eqref{eq:finite-error-projection-domain} meet the hypotheses
\eqref{eq:g-projection-domain}. Since
$(r_*,1/3)\subset[3/10,1/3]$, the proved projection of the common
parameter from $G_s$ to $r$ preserves the four segment predicates
and the other necessary inequalities. At this projected point,
replacing $d_i$ by the uncapped bounds used in~\eqref{eq:scalar-caps}
increases the corresponding upper bounds. This gives the same bare
functions and necessary system $\Sigma_r$ as in the ideal derivation.

The independent unit \nolinkurl{near-BR-delta-ledger}\textsuperscript{\ref{fn:checker-near-BR-delta-ledger}} checks
the specified 34 nonnegative
combinations, seven identities, the two minimum cases in
\eqref{eq:finite-error-old-payoff-loss}, and the source substitutions
with common weights. The mathematical substitutions above and
those exact checks establish
\begin{equation}
 f(\boldsymbol x_{\mathrm{out}})>r+\delta_f+\tau
 \quad\Longrightarrow\quad
 \exists\boldsymbol\theta,\ \Sigma_r(\boldsymbol\theta),
 \qquad r_*<r<1/3.
 \label{eq:finite-error-transfer}
\end{equation}
Equation~\eqref{eq:global-exclusion} contradicts its conclusion.
If $f(\boldsymbol x_{\mathrm{out}})>r_*+\delta_f+\tau$, choose
$r$ strictly between $r_*$ and
$\min\{f(\boldsymbol x_{\mathrm{out}})-\delta_f-\tau,1/3\}$.
This would satisfy \eqref{eq:finite-error-failure}, so
\begin{equation}
 f(\boldsymbol x_{\mathrm{out}})\le r_*+\delta_f+\tau.
 \label{eq:finite-error-bound}
\end{equation}
Taking $\delta_f=\tau=\delta/2$ proves Theorem~\ref{thm:main};
the polynomial running time follows from the fixed composition
in Section~\ref{sec:complete-algorithm}.
If $\delta>1$, normalized regrets are at most one and any profile
suffices. A scalar point at $r_*$ has been checked, but no game
realizing it has been established. The conclusion is an upper
bound on the algorithm's regret, without a tightness claim.


\begin{thebibliography}{99}

\bibitem{nash1951}
John Nash.
\newblock Non-cooperative games.
\newblock \emph{Annals of Mathematics}, 54(2):286--295, 1951.
\newblock \href{https://doi.org/10.2307/1969529}{doi:10.2307/1969529}.

\bibitem{daskalakis2009}
Constantinos Daskalakis, Paul W. Goldberg, and Christos H. Papadimitriou.
\newblock The complexity of computing a Nash equilibrium.
\newblock \emph{SIAM Journal on Computing}, 39(1):195--259, 2009.
\newblock \href{https://doi.org/10.1137/070699652}{doi:10.1137/070699652}.

\bibitem{chen2009}
Xi Chen, Xiaotie Deng, and Shang-Hua Teng.
\newblock Settling the complexity of computing two-player Nash equilibria.
\newblock \emph{Journal of the ACM}, 56(3):14:1--14:57, 2009.
\newblock \href{https://doi.org/10.1145/1516512.1516516}{doi:10.1145/1516512.1516516}.

\bibitem{lipton2003}
Richard J. Lipton, Evangelos Markakis, and Aranyak Mehta.
\newblock Playing large games using simple strategies.
\newblock In \emph{Proceedings of the 4th ACM Conference on Electronic
Commerce}, pages 36--41, 2003.
\newblock \href{https://doi.org/10.1145/779928.779933}{doi:10.1145/779928.779933}.

\bibitem{babichenko2017}
Yakov Babichenko, Siddharth Barman, and Ron Peretz.
\newblock Empirical distribution of equilibrium play and its testing application.
\newblock \emph{Mathematics of Operations Research}, 42(1):15--29, 2017.
\newblock \href{https://doi.org/10.1287/moor.2016.0794}{doi:10.1287/moor.2016.0794}.

\bibitem{rubinstein2016}
Aviad Rubinstein.
\newblock Settling the complexity of computing approximate two-player Nash
equilibria.
\newblock In \emph{Proceedings of the 57th IEEE Symposium on Foundations of
Computer Science}, pages 258--265, 2016.
\newblock \href{https://doi.org/10.1109/FOCS.2016.35}{doi:10.1109/FOCS.2016.35}.

\bibitem{survey2024}
Hanyu Li, Wenhan Huang, Zhijian Duan, David Henry Mguni, Kun Shao, Jun Wang,
and Xiaotie Deng.
\newblock A survey on algorithms for Nash equilibria in finite normal-form games.
\newblock \emph{Computer Science Review}, 51:100613, 2024.
\newblock \href{https://doi.org/10.1016/j.cosrev.2023.100613}{doi:10.1016/j.cosrev.2023.100613}.

\bibitem{kontogiannis2009}
Spyros C. Kontogiannis, Panagiota N. Panagopoulou, and Paul G. Spirakis.
\newblock Polynomial algorithms for approximating Nash equilibria of
bimatrix games.
\newblock \emph{Theoretical Computer Science}, 410(17):1599--1606, 2009.
\newblock \href{https://doi.org/10.1016/j.tcs.2008.12.033}{doi:10.1016/j.tcs.2008.12.033}.

\bibitem{daskalakisnote2009}
Constantinos Daskalakis, Aranyak Mehta, and Christos H. Papadimitriou.
\newblock A note on approximate Nash equilibria.
\newblock \emph{Theoretical Computer Science}, 410(17):1581--1588, 2009.
\newblock \href{https://doi.org/10.1016/j.tcs.2008.12.031}{doi:10.1016/j.tcs.2008.12.031}.

\bibitem{daskalakisprogress2007}
Constantinos Daskalakis, Aranyak Mehta, and Christos H. Papadimitriou.
\newblock Progress in approximate Nash equilibria.
\newblock In \emph{Proceedings of the 8th ACM Conference on Electronic
Commerce}, pages 355--358, 2007.
\newblock \href{https://doi.org/10.1145/1250910.1250962}{doi:10.1145/1250910.1250962}.

\bibitem{bbm2010}
Hartwig Bosse, Jaroslaw Byrka, and Evangelos Markakis.
\newblock New algorithms for approximate Nash equilibria in bimatrix games.
\newblock \emph{Theoretical Computer Science}, 411(1):164--173, 2010.
\newblock \href{https://doi.org/10.1016/j.tcs.2009.09.023}{doi:10.1016/j.tcs.2009.09.023}.

\bibitem{ts2008}
Haralampos Tsaknakis and Paul G. Spirakis.
\newblock An optimization approach for approximate Nash equilibria.
\newblock \emph{Internet Mathematics}, 5(4):365--382, 2008.
\newblock \href{https://doi.org/10.1080/15427951.2008.10129172}{doi:10.1080/15427951.2008.10129172}.

\bibitem{dfm2023}
Argyrios Deligkas, Michail Fasoulakis, and Evangelos Markakis.
\newblock A polynomial-time algorithm for 1/3-approximate Nash equilibria in
bimatrix games.
\newblock \emph{ACM Transactions on Algorithms}, 19(4):31:1--31:17, 2023.
\newblock \href{https://doi.org/10.1145/3606697}{doi:10.1145/3606697}.

\bibitem{tightness2023}
Zhaohua Chen, Xiaotie Deng, Wenhan Huang, Hanyu Li, and Yuhao Li.
\newblock On tightness of Tsaknakis--Spirakis descent methods for approximate
Nash equilibria.
\newblock \emph{Information and Computation}, 293:105046, 2023.
\newblock \href{https://doi.org/10.1016/j.ic.2023.105046}{doi:10.1016/j.ic.2023.105046}.

\bibitem{computeraided2026}
Xiaotie Deng, Dongchen Li, and Hanyu Li.
\newblock A computer-aided approach for approximate Nash equilibria.
\newblock In \emph{Proceedings of the 20th International Conference on Web
and Internet Economics (WINE 2024)}, volume 15534 of \emph{Lecture Notes in
Computer Science}, pages 314--331, 2026.
\newblock \href{https://doi.org/10.1007/978-3-032-08560-3_18}{doi:10.1007/978-3-032-08560-3\_18}.

\bibitem{legone2026}
Hanyu Li, Dongchen Li, and Xiaotie Deng.
\newblock Discovering expert-level Nash equilibrium algorithms with large
language models.
\newblock \emph{Nature Communications}, 17:7168, 2026.
\newblock \href{https://doi.org/10.1038/s41467-026-74003-1}{doi:10.1038/s41467-026-74003-1}.

\bibitem{hemon2008}
S{\'e}bastien H{\'e}mon, Michel de Rougemont, and Miklos Santha.
\newblock Approximate Nash equilibria for multi-player games.
\newblock In \emph{Proceedings of the 1st International Symposium on
Algorithmic Game Theory}, volume 4997 of \emph{Lecture Notes in Computer
Science}, pages 267--278, 2008.
\newblock \href{https://doi.org/10.1007/978-3-540-79309-0_24}{doi:10.1007/978-3-540-79309-0\_24}.

\bibitem{automating2025}
Xiaotie Deng, Dongchen Li, and Hanyu Li.
\newblock Automating approximation analysis for Nash equilibria algorithms in
two-player games.
\newblock \emph{Information and Computation}, 307:105362, 2025.
\newblock \href{https://doi.org/10.1016/j.ic.2025.105362}{doi:10.1016/j.ic.2025.105362}.

\bibitem{mixing2025}
Xiaotie Deng, Dongchen Li, and Hanyu Li.
\newblock On the optimal mixing problem of approximate Nash equilibria in
bimatrix games.
\newblock \emph{Theoretical Computer Science}, 1031:115072, 2025.
\newblock \href{https://doi.org/10.1016/j.tcs.2025.115072}{doi:10.1016/j.tcs.2025.115072}.

\bibitem{gondauri2026}
Davit Gondauri.
\newblock A deterministic polynomial-time approximation below one third for
bimatrix games.
\newblock Preprint, 2026.
\newblock \href{https://zenodo.org/records/23077766}{Zenodo record 23077766}.

\bibitem{boyd2004}
Stephen Boyd and Lieven Vandenberghe.
\newblock \emph{Convex Optimization}.
\newblock Cambridge University Press, 2004.
\newblock \href{https://web.stanford.edu/~boyd/cvxbook/}{Author-hosted book and supplementary materials}.

\bibitem{basu2006}
Saugata Basu, Richard Pollack, and Marie-Fran\c{c}oise Roy.
\newblock \emph{Algorithms in Real Algebraic Geometry}, second edition.
\newblock Algorithms and Computation in Mathematics, volume 10.
Springer, 2006.
\newblock \href{https://doi.org/10.1007/3-540-33099-2}{doi:10.1007/3-540-33099-2}.

\bibitem{scipy2020}
Pauli Virtanen, Ralf Gommers, Travis E. Oliphant, et al.
\newblock SciPy 1.0: fundamental algorithms for scientific computing in Python.
\newblock \emph{Nature Methods}, 17:261--272, 2020.
\newblock \href{https://doi.org/10.1038/s41592-019-0686-2}{doi:10.1038/s41592-019-0686-2}.

\bibitem{kraft1988}
Dieter Kraft.
\newblock A software package for sequential quadratic programming.
\newblock Technical Report DFVLR-FB 88-28, Institut f\"ur Dynamik der
Flugsysteme, DFVLR, Oberpfaffenhofen, 1988.
\newblock \href{https://degenerateconic.com/uploads/2018/03/DFVLR_FB_88_28.pdf}{Technical report}.

\bibitem{nocedal2006}
Jorge Nocedal and Stephen J. Wright.
\newblock \emph{Numerical Optimization}, second edition.
\newblock Springer, 2006.
\newblock \href{https://doi.org/10.1007/978-0-387-40065-5}{doi:10.1007/978-0-387-40065-5}.

\bibitem{moore2009}
Ramon E. Moore, R. Baker Kearfott, and Michael J. Cloud.
\newblock \emph{Introduction to Interval Analysis}.
\newblock Society for Industrial and Applied Mathematics, 2009.
\newblock \href{https://doi.org/10.1137/1.9780898717716}{doi:10.1137/1.9780898717716}.

\bibitem{krawczyk1969}
R. Krawczyk.
\newblock Newton-Algorithmen zur Bestimmung von Nullstellen mit Fehlerschranken.
\newblock \emph{Computing}, 4:187--201, 1969.
\newblock \href{https://doi.org/10.1007/BF02234767}{doi:10.1007/BF02234767}.

\bibitem{cox2015}
David A. Cox, John Little, and Donal O'Shea.
\newblock \emph{Ideals, Varieties, and Algorithms: An Introduction to
Computational Algebraic Geometry and Commutative Algebra}, fourth edition.
\newblock Springer, 2015.
\newblock \href{https://doi.org/10.1007/978-3-319-16721-3}{doi:10.1007/978-3-319-16721-3}.

\bibitem{sympy2017}
Aaron Meurer, Christopher P. Smith, Mateusz Paprocki, et al.
\newblock SymPy: symbolic computing in Python.
\newblock \emph{PeerJ Computer Science}, 3:e103, 2017.
\newblock \href{https://doi.org/10.7717/peerj-cs.103}{doi:10.7717/peerj-cs.103}.

\bibitem{python3}
Python Software Foundation.
\newblock \emph{Python 3 Documentation}.
\newblock \href{https://docs.python.org/3/}{Official documentation};
\href{https://docs.python.org/3/library/fractions.html}{\nolinkurl{fractions}---rational numbers}.
\newblock Accessed October 4, 2026.

\bibitem{garloff1986}
J\"urgen Garloff.
\newblock Convergent bounds for the range of multivariate polynomials.
\newblock In Karl Nickel, editor, \emph{Interval Mathematics 1985},
volume 212 of \emph{Lecture Notes in Computer Science}, pages 37--56.
Springer, 1986.
\newblock \href{https://doi.org/10.1007/3-540-16437-5_5}{doi:10.1007/3-540-16437-5\_5}.

\bibitem{numpy2020}
Charles R. Harris, K. Jarrod Millman, St{\'e}fan J. van der Walt, et al.
\newblock Array programming with NumPy.
\newblock \emph{Nature}, 585:357--362, 2020.
\newblock \href{https://doi.org/10.1038/s41586-020-2649-2}{doi:10.1038/s41586-020-2649-2}.

\bibitem{rudin1976}
Walter Rudin.
\newblock \emph{Principles of Mathematical Analysis}, third edition.
\newblock McGraw-Hill, 1976.
\newblock \href{https://www.mheducation.co.uk/principles-of-mathematical-analysis-int-l-ed-9780070856134-emea}{Publisher's book page}.

\end{thebibliography}
\end{document}